\documentclass{article}
\usepackage{kr}

\usepackage{times}
\usepackage{soul}
\usepackage{url}
\usepackage[hidelinks]{hyperref}
\usepackage[utf8]{inputenc}
\usepackage[small]{caption}
\usepackage{graphicx}
\usepackage{amsmath}
\usepackage{amsthm}
\usepackage{booktabs}
\usepackage{algorithm}
\usepackage{algorithmic}
\title{On the Semantics of Generative SPARQL}

\author{%
    Author name
    \affiliations
    Affiliation
    \emails
    email@example.com    
}
\author{%
Ratan Bahadur Thapa$^{1}$\thanks{Corresponding author: ratan.thapa@ki.uni-stuttgart.de} \, \and\,
Steffen Staab$^{1,2}$\\
\affiliations
$^1$Institute for Artificial Intelligence, University of Stuttgart, Stuttgart, Germany\\
$^2$University of Southampton, United Kingdom\\
\emails
\{ratan.thapa, steffen.staab\}@ki.uni-stuttgart.de
}

\usepackage{amsmath}
\usepackage{enumitem}
\usepackage{xcolor}
\usepackage{color}
\usepackage{xspace}
\usepackage{amssymb}
\usepackage{mathtools}

\usepackage{stmaryrd}

\usepackage{amsthm}

\theoremstyle{plain}
\newtheorem{theorem}{Theorem}
\newtheorem{lemma}{Lemma}
\newtheorem{proposition}{Proposition}
\newtheorem{corollary}{Corollary}

\theoremstyle{definition}
\newtheorem{definition}{Definition}
\newtheorem{example}{Example}

\newtheorem{assumption}{Assumption}

\theoremstyle{remark}
\newtheorem{remark}{Remark}

\newcommand\m[1]{\ensuremath{\mathbf{#1}}}
\newcommand\ma[1]{\ensuremath{\mathcal{#1}}}

\newcommand*\wrt{\text{w.r.t.}\xspace}

\newcommand*\ie{\text{i.e.,}\xspace}

\newcommand*\eg{\text{e.g.,}\xspace}

\newcommand\tx[1]{\texttt{#1}}

\newcommand*\simtext{\tx{Sim}\xspace}

\newcommand*\stringSim{\ensuremath{\tx{Str}_\simeq\xspace}}

\newcommand*\iris{\tx{IRIs}\xspace}
\newcommand*\bgps{\tx{BGP}\xspace}

\newcommand*\filter[2]{\ensuremath{\tx{Filter}_{\,#2}(#1)}}
\newcommand*\union[2]{\ensuremath{#1\,\tx{Union}\,#2}}
\newcommand*\join[2]{\ensuremath{#1\, \tx{Join}\, #2}}
\newcommand*\minus[2]{\ensuremath{#1\, \tx{Minus}\, #2}}
\newcommand*\proj[2]{\ensuremath{\tx{Proj}_{\,#2}(#1)}}

\newcommand*\opt[3]{\ensuremath{#1\, \tx{Opt}_{\,#3}\,#2}}
\newcommand*\diff[3]{\ensuremath{#1\,\tx{Diff}_{\,#3}\,#2}}
\newcommand*\genbind[3]{\ensuremath{\tx{GenOp}(#1,#2,#3)}}

\newcommand*\promptstr[1]{\ensuremath{\pi(#1)}\xspace}
\newcommand*\sourceRdf{\ensuremath{\tx{rdf}\xspace}}
\newcommand*\sourceGen{\ensuremath{\tx{gen}\xspace}}

\newcommand*\bound[1]{\ensuremath{\tx{bound}(#1)}}
\newcommand*\var[1]{\ensuremath{\tx{var}(#1)}}
\newcommand*\domain[1]{\ensuremath{\tx{dom}_{#1}}\xspace}
\newcommand*\holder[1]{\ensuremath{\tx{input}(#1)}}
\newcommand*\target[1]{\ensuremath{\tx{out}(#1)}}
\newcommand*\compact[3]{\ensuremath{\tx{Compat}_{#1}(#2,#3)}}
\newcommand*\size[1]{\ensuremath{|#1|}}
\newcommand*\occurrence[2]{\ensuremath{#1\preceq#2}}
\newcommand*\canonical{\tx{Canon}\xspace}
\newcommand*\lexical{\tx{Lex}\xspace}
\newcommand*\evalTwo[2]{\ensuremath{ {\llbracket #1\rrbracket}_{#2}\xspace}}
\newcommand*\evalThree[3]{{\llbracket #1\rrbracket}_{#2}^{#3}{\xspace}}
\newcommand*\genPattern[1]{\ensuremath{\tx{gen}(#1)}}
\newcommand*\directedGraph[1]{\ensuremath{\overrightarrow{#1}}}
\newcommand*\Cand[1]{\ensuremath{\tx{Cand}(#1)\xspace}}
\newcommand*\Consequence[1]{\ensuremath{\Phi_{#1}}}
\newcommand*\ContextS[2]{\ensuremath{#1|_{#2}}}

\newcommand*\rdf{\ensuremath{\mathrm{Rdf}\xspace}}
\newcommand*\gen{\ensuremath{\mathrm{Gen}\xspace}}
\newcommand*\dom{\ensuremath{\mathrm{Dom}\xspace}}
\newcommand*\Call[1]{\ensuremath{\tx{Call}_{#1}\xspace}}
\newcommand*\Parse[1]{\ensuremath{\tx{Parse}_{#1}\xspace}}
\newcommand*\Sol{\ensuremath{\mathrm{Sol}\xspace}}
\newcommand*\lfp{\ensuremath{\mathrm{lfp}\xspace}}

\newcommand*\JJoin{\ensuremath{\mathrm{Join}\xspace}}
\newcommand*\ffalse{\ensuremath{\mathrm{false}\xspace}}
\newcommand*\ttrue{\ensuremath{\mathrm{true}\xspace}}
\newcommand*\eval{\ensuremath{\mathrm{Eval}\xspace}}
\newcommand*\undeff{\ensuremath{\mathrm{undef}\xspace}}
\newcommand*\GenOp{\ensuremath{\mathrm{GenOp}\xspace}}
\newcommand*\genop{\ensuremath{\mathrm{GenOp}\xspace}}
\newcommand*\Gen{\ensuremath{\mathrm{Gen}\xspace}}
\newcommand*\sem{\ensuremath{\tx{Sem}\xspace}}

\newcommand*\rep{\ensuremath{\tx{rep}\xspace}}
\newcommand*\emptypat{\ensuremath{\{\}\xspace}}
\newcommand{\Self}{\mathsf{Self}}

\begin{document}

\maketitle

\begin{abstract}
We extend SPARQL with a generative query construct, called \tx{GenOp}, whose evaluation calls a language model and produces typed solution mappings. We define the semantics of the GenOp in the query in a way that maintains the fixed-dataset assumption, on which formal semantics of SPARQL build, and extend solution mappings with values generated by the language model.  
We formalize the semantics of the extended language over these mappings using a compatibility relation that generalizes equality and supports similarity-based matching between RDF terms and generated values.
We analyze the semantic consequences of generative query patterns, focusing on mapping-level recursion induced by the reuse of generated bindings. Under deterministic bounded generation and finite candidate coverage assumptions, we characterize acyclic and stratified fragments with fixpoint semantics, establish algebraic equivalence and semantics-preserving rewrite rules, and provide an executable evaluation method; and we show that  data and combined complexity coincide with those of standard SPARQL.
\end{abstract}

\section{Introduction}
SPARQL  \cite{prud2008sparql} is the de facto standard for querying RDF \cite{manola2004rdf}. The semantics of SPARQL query patterns is defined in terms of (partial) solution mappings, with operators analogous to relational-algebra operators \cite{codd1970relational} and evaluated with respect to a fixed RDF graph \cite{perez2006semantics,perez2009semantics}. All semantic characterizations, equivalence results, and algebraic rewrite rules are stated relative to an underlying RDF graph that remains fixed during query evaluation \cite{schmidt2010foundations}.

Recent knowledge graph centric applications increasingly require \emph{exploratory enrichment} over existing datasets \cite{hogan2021knowledge}, such as synthesizing missing attributes, proposing related entities, or producing textual descriptions. Large language models (LLMs) can generate such values on demand, but current practice treats calls to language models operationally, as external procedures, middleware invocations, or prompt-based post-processing, outside the query algebra \cite{zhaohybrid,tan2025can,khabiri2025declarative,saeed2024querying,dai2024uqe,molli2025sparqllm}. As a result, their interaction with SPARQL operators is not governed by the declarative semantics of the query language, and standard questions about query equivalence, recursion, and termination remain outside the scope of these works.

One way to formally capture the semantics of operators that generate new values at query time is to relax the fixed-dataset assumption and  allow the queried graph to be discovered during evaluation, as in link-traversal querying \cite{hartig2009executing}. However, such approaches rely on dynamic data discovery and require specialized semantics that preclude algebraic equivalence, recursion, and fixpoint reasoning under a fixed graph \cite{hartig2009executing,florescu1998database,hernandez2019deep}. In contrast, we aim to preserve the fixed-dataset assumption of SPARQL while \emph{extending only the space of solution mappings} under a bounded generation assumption.

We integrate generative inference into the SPARQL algebra through a generative query pattern $\tx{GenOp}$. A $\tx{GenOp}$ instantiates a prompt template using the currently bound variables, invokes a specified generative model, parses finitely many outputs, and returns them as bindings for designated output variables. 

At the semantic level, our central design choice is to lift generation to the level of solution mappings rather than RDF graphs. To support algebraic composition with standard SPARQL operators, we define query semantics over \emph{typed solution mappings} that distinguish RDF-derived bindings from generated bindings. Each binding records its origin, keeping evaluation relative to a fixed RDF graph. To combine symbolic and generated values, we replace equality-based compatibility with a generalized compatibility relation: it coincides with equality on RDF bindings and allows similarity-based matching between RDF terms and generated values, following the general spirit of similarity-based querying\cite{kiefer2007fundamentals,zheng2016semantic,ferrada2024similarity}, via deterministic
lexicalization and a variable-indexed similarity threshold.

Although $\tx{GenOp}$ is non-recursive syntactically, cyclic dependencies may arise when generated bindings are reused as inputs to other $\tx{GenOp}$s. Such cycles occur solely at the level of extended solution mappings and induce self-referential constraints on admissible solutions. We formalize these dependencies via a consequence operator over sets of typed mappings and interpret cyclic queries using fixpoint semantics. 
We restrict ourselves to the deterministic bounded-generation assumption, where each fixed model specifier and instantiated prompt returns a fixed finite parsed output set, yielding a finite generated-value universe for each query instance. Under this assumption, query evaluation reduces to fixpoint computation over a finite mapping domain and therefore terminates by standard fixpoint arguments for finite lattices \cite{ullman1988principles,abiteboul1995foundations,green2013datalog}. Let us explore the difficulty of including value-generating LLM calls into SPARQL query processing with the following example. 
\begin{example}\label{example_one}
Let $\ma G=\{\tx{:Paris a :City}\}$ be an RDF graph, and consider the following
pseudo-generative SPARQL query:
\begin{verbatim}
SELECT ?x ?y ?z WHERE {
  ?x a :City .
  GenOp("Suggest a topic related to ?x 
                and ?z" AS ?y, GPT-4o).
  GenOp("Suggest a city related to 
     topic ?y" AS ?z, Gemini 1.5 Pro).
}
\end{verbatim}

Ideally, the first $\tx{GenOp}$ would behave analogously to a basic graph pattern: it should use currently available bindings and otherwise range over admissible combinations of variable assignments. Typical procedural implementations (\eg \cite{florescu1998database}) require all input variables of a $\tx{GenOp}$ to be bound before invocation. In this query, however, each $\tx{GenOp}$ depends on a value produced by the other, so neither can be evaluated under such an execution model. The problem is therefore that procedural execution enforces an evaluation order that blocks mutually dependent generative calls; our semantics instead provides a declarative interpretation that admits jointly supported bindings.

Under our semantics (see Appendix~\ref{app:example1-exec}), evaluation is defined over sets of typed solution mappings, and the basic graph pattern yields the initial mapping $\mu_0 = {?x \mapsto \tx{:Paris}}$. Each $\tx{GenOp}$ induces a mapping transformer. A typed mapping $\mu \supseteq \mu_0$ assigning $?y$ and $?z$ is admissible iff $?y$ is among the finitely many outputs generated from the first $\tx{GenOp}$ prompt instantiated with $\mu(?x)$ and $\mu(?z)$, and $?z$ is analogously supported by the second $\tx{GenOp}$ prompt instantiated with $\mu(?y)$.

The two $\tx{GenOp}$ occurrences therefore induce a cyclic dependency over the
mapping domain. The declarative solution set is given by the relevant fixpoint
semantics, while executable evaluation validates finite candidates against the
recursive equations under the candidate-coverage assumptions stated in
Sect.~\ref{sec:exec-fixpoint}.

\end{example}

Overall, we isolate a semantic foundation for hybrid symbolic–generative querying in which generative operators are formalized as \emph{context-sensitive transformers of solution mappings}. Their interaction is defined declaratively by fixpoint semantics over finite, typed mapping domains, supporting formal reasoning about correctness, termination, and equivalence for queries that incorporate generative inference.

\textbf{Our contributions are as follows}. (i) We extend SPARQL with a generative operator \tx{GenOp} and define a denotational semantics over typed solution mappings together with a generalized compatibility relation. Generated values remain confined to solution mappings and do not induce fact derivation, distinguishing the semantics from existential-rule languages and chase-based formalisms\cite{abiteboul1995foundations,cali2010query,cali2010datalog+}. \emph{Recursion is interpreted over mappings} rather than by model construction, distinguishing our approach from answer-set and stable-model approaches \cite{przymusinski1990well,niemela1999logic}.
(ii) We characterize acyclic, stratified, and cyclic generative patterns and provide fixpoint semantics for the recursive cases. Recursive admissibility is defined by a self-support condition. The underlying fixpoint theory is classical \cite{ullman1988principles,abiteboul1995foundations}, while its application to \emph{context-dependent},  deterministic set-valued mapping transformers over solution mappings is \emph{new}. (iii) We establish algebraic equivalence results and semantics-preserving rewrite rules for monotone and stratified fragments with respect to the fixpoint semantics over typed solution mappings, identifying both the preserved standard SPARQL algebraic laws and \emph{new algebraic rewritings} in the presence of context-dependent generative patterns. (iv) We provide an executable realization of the recursive semantics by reducing evaluation to self-support checking over finite candidate mapping sets localized to strongly connected components of the dependency graph. This evaluation avoids chase-based reasoning and yields soundness and completeness relative to candidate coverage, following established treatments of logic programs with external operators \cite{eiter2006effective}.Under the bounded-generation assumption and finite candidate coverage assumptions,  both data and combined complexity coincide with those of standard SPARQL evaluation\cite{perez2006semantics} (Section~\ref{sec:complexity}).

\section{Related Work}
The formal semantics of SPARQL as an algebra over solution mappings, together with
results on containment, entailment and optimization under a fixed RDF dataset, are well
established \cite{perez2009semantics,arenas2011querying,chekol2016containment,kontchakov2016expressibility,kaminski2017query}.
Our work builds on this foundation but departs from it by allowing solution
mappings to be extended by generative operators during evaluation.

A distinct line of SPARQL extension for the Web of Linked Data has formalized semantics and computability under full-Web and reachability-based interpretations \cite{hartig2009executing,hartig2013squin}, examined the empirical behavior and limitations of link-traversal query execution over heterogeneous Web data \cite{umbrich2015link}, and proposed dedicated Web-of-Data query languages that separate navigation from query evaluation, such as LDQL \cite{hartig2016ldql}.  These approaches address open-world data discovery and computability over potentially unbounded graphs.

Recent systems integrate LLMs into knowledge-graph querying pipelines. SparqLLM invokes LLMs as auxiliary services during SPARQL execution \cite{molli2025sparqllm}, while other work uses declarative orchestration of LLM-assisted API and database calls around existing query engines \cite{khabiri2025declarative}. Related approaches combine retrieval and generation in pipeline architectures \cite{yang2025comprehensive}. In all cases, language model invocation remains external to the query algebra and is treated operationally, with generated outputs consumed procedurally rather than defined by the query semantics.

Related relational data-management work has also studied SQL and SQL-adjacent query processing with language-model calls, including hybrid querying over relational databases and LLMs, SQL-style querying of LLMs, and query engines over unstructured data \cite{zhaohybrid,saeed2024querying,dai2024uqe,tan2025can}. These works establish model invocation as a useful operational primitive in query processing and study how such calls can be exposed through SQL-style syntax, user-defined functions, physical operators, or dedicated execution engines. In all cases, focus is primarily execution-oriented and SQL-centric: model calls are incorporated as callable components of a query pipeline, without developing a formal denotational semantics for generated bindings inside a solution-mapping algebra.

Fixpoint semantics for recursive queries originate in deductive databases and
logic programming \cite{ullman1988principles,abiteboul1995foundations}, with
stratified and well-founded semantics providing principled treatments of
negation \cite{przymusinski1990well,van1989alternating}. Existential rules and
Datalog$^\pm$ extend these ideas via domain-expanding inference
\cite{cali2010query,cali2010datalog+}. Our setting differs fundamentally: the RDF
domain remains fixed, no facts are derived, and recursion is resolved over a finite
space of typed solution mappings. Treatments of logic programs with external
operators \cite{eiter2006effective} are conceptually closest, but does not address algebraic query languages or mapping-level recursion under fixed-dataset semantics.

\section{Syntax and Semantics}\label{sect:syntax_semantics}
We now introduce the syntax of generative SPARQL and define its denotational semantics in terms of typed solution mappings.

Let  $\ma V$, $\ma I$,  $\ma B$ and $\ma L$ be countably infinite disjoint sets of \emph{variables}, \emph{Internationalized Resource Identifiers} (\iris), \emph{blank nodes} and \emph{literals}, respectively. Let $\ma T=\ma I\cup\ma B\cup \ma L$ be a set of RDF terms, and let $\ma T_\tx{gen}$ be a disjoint set of \emph{generated terms}.  An \emph{RDF triple} is an element of $(\ma I\cup\ma B)\times\ma I\times\ma T$, and a \emph{triple pattern} is an element of $(\ma I\cup\ma L\cup\ma V)\times (\ma I\cup\ma V)\times (\ma I\cup\ma L\cup\ma V)$. An \emph{RDF graph} $\ma G$ is a finite set of RDF triples, and a \emph{basic graph pattern (\bgps)} is a finite set of triple patterns.


A generative SPARQL query is a \emph{graph pattern} $\ma P $ defined by the following grammar:
\begin{align}\label{eq:query-grammar}
     \ma P\coloneqq &\ \bgps \mid\ \genbind{\promptstr{X}}{Y}{\ma M} \mid\ \filter{\ma P}{\ma F} \nonumber  \\
     &\ \mid\ \proj{\ma P}{L} 
     \mid \union{\ma P}{\ma P} \mid   \join{\ma{P}}{\ma{P}} \\ \nonumber
     &\,  \mid\ \minus{\ma P}{\ma P}\mid \diff{\ma P}{\ma P}{\ma F}\mid\ \opt{\ma P}{\ma P}{\ma F},
 \end{align}
 where  $\genbind{\pi(X)}{Y}{\ma M}$ is a \emph{generative pattern} with prompt template $\pi$ containing placeholders $X\subseteq\ma V$, output variables $Y\subseteq\ma V$ with $X\cap Y=\varnothing$, and generative model specifier $\ma M$ (abstract); $L\subseteq\ma V$ is a projection list, and $\ma F$ is a filter formula built from atoms $\bound{v}$, $(v=c)$, $(v=v')$ s.t. $v,v'\in\ma V$ and $c\in\ma T$ using logical  connectives $\wedge$ and $\neg$. The set of variables in a graph pattern \ma P is denoted by \var{\ma P}, and the size of \ma P, \ie the number of symbols in its writing, by \size{\ma P}.

\begin{definition}
Let $\ma P$ be a graph pattern. A \bgps or generative $\tx{GenOP}$ pattern $Q$ is said to occur in $\ma P$, written \occurrence{Q}{\ma P}, if $Q$ is a syntactic subpattern of $\ma P$, defined inductively by: (i) if $\ma P = Q$, then $Q \preceq \ma P$; (ii) if $\ma P = \filter{\ma P'}{\ma F}$ or $\ma P = \proj{\ma P'}{L}$, then $Q \preceq \ma P$ iff $Q \preceq \ma P'$; (iii) if $\ma P = \ma P_1\, \tx{Op}\, \ma P_2$ for $\tx{Op}\in\{\tx{Join},\tx{Union},\tx{Minus},\tx{Diff},\tx{Opt}\}$, then $Q \preceq \ma P$ iff $Q \preceq \ma P_1$ or $Q \preceq \ma P_2$.
\end{definition}

Let $\emptypat$ be the empty graph pattern, with $\evalTwo{\emptypat}{\ma G}=\{\emptyset\}$ and $\var{\emptypat}=\emptyset$.

\begin{definition}\label{defn_surroundin_patterns}
Let $\ma P$ be a graph pattern and let $Q\preceq\ma P$ be a designated
occurrence. The \emph{surrounding pattern} of $Q$ in $\ma P$, denoted
$\ma P\setminus\{Q\}$, is obtained by replacing only that occurrence by the
empty pattern $\emptypat$ and then
simplifying the result by the following pruning $\leadsto$ rules:
\begin{enumerate}
    \item $\filter{\emptypat}{\ma F}\leadsto \emptypat$ and $
\proj{\emptypat}{L}\leadsto \emptypat$;
\item For the $\tx{Op}\in\{\tx{Join},\tx{Union},\tx{Minus},\tx{Diff},\tx{Opt}\}$ and graph pattern $\ma P'$,
\[\emptypat\ \tx{Op}\ \ma P'\leadsto \ma P'\quad\text{and}\quad\ma P'\ \tx{Op}\ \emptypat\leadsto \ma P';\]
\end{enumerate}
provided that the resulting graph pattern is \emph{well formed}.
\end{definition}

\begin{definition}(Well formed)
A graph pattern $\ma P$ defined by the grammar \ref{eq:query-grammar} is \emph{well-formed} if it satisfies the following conditions: (i) for every $g=\genbind{\pi(X)}{Y}{\ma M}$  occurrence such that $g \preceq \ma P$, $X \subseteq \var{\ma P \setminus \{g\}} \setminus Y$; (ii) for every filter $\filter{\ma P}{\ma F}$, $\var{\ma F}\subseteq \var{\ma P}$; and (iii) for every projection $\proj{\ma P}{L}$, $L\subseteq \var{\ma P}$.
\end{definition}

We now define solution mappings, typing, and compatibility.

\begin{definition}
A \emph{typed term} is a pair $(t,\tau)$ where either (i) $t\in \ma T$ and $\tau=\sourceRdf$, or (ii) $t\in \ma T_{\tx{gen}}$ and $\tau=\sourceGen$. A \emph{typed solution mapping} is a partial function $\mu:\ma V\rightharpoonup (\ma T\times\{\sourceRdf\})\cup (\ma T_{\tx{gen}}\times\{\sourceGen\})$ with domain $\domain{\mu}\subseteq \ma V$. 
\end{definition}


Let $\Sigma$ be a finite alphabet and $\Sigma^*$ be the set of all finite strings over $\Sigma$. 

\begin{definition}
A \emph{canonical lexicalization} consists of fixed deterministic functions
\[\canonical : \ma T \to \Sigma^* \quad \text{and} \quad\lexical : \ma T_\tx{gen} \to \Sigma^*,\]
where $\canonical$ maps each RDF term to a canonical string and $\lexical$ maps each generated term to its textual realization.
\end{definition}

\begin{remark}
In an implementation, literals can be lexicalized by their lexical forms. IRIs
can be lexicalized by a deterministic label policy, for example by using an
associated \emph{rdfs:label} when available and otherwise a stable rendering of
the IRI. Blank nodes can be lexicalized by a stable system-defined identifier or
by a deterministic bounded-neighbourhood description. We require only that the
lexicalization policy is fixed during evaluation; the semantics does not require
labels or blank-node descriptions to be semantically complete.
\end{remark}
\begin{definition}
 Let $\stringSim:\Sigma^*\times\Sigma^*\to[0,1]$ be a symmetric similarity function, and let $\theta:\ma V\to(0,1]$ be a variable-specific \emph{threshold function} (default $1$). For any two typed terms $(t,\tau),(t',\tau')$, we define $\simtext$ as,
\begin{multline*}
    \simtext((t,\tau),(t',\tau')) =\\
\begin{cases}
1 \hspace{3.3cm}\tx{if} &(\tau=\tau'=\sourceRdf) \wedge (t=t')\\
\stringSim(\canonical(t),\lexical(t')) &\tx{if }\tau=\sourceRdf,\tau'=\sourceGen\\
\stringSim(\lexical(t),\canonical(t')) &\tx{if }\tau=\sourceGen,\ \tau'=\sourceRdf\\
\stringSim(\lexical(t),\lexical(t')) &\tx{if }\tau=\tau'=\sourceGen
\end{cases}
\end{multline*}
\end{definition}

\begin{definition}[Variable-indexed similarity]
For each variable $v\in\ma V$, let $\approx_v$ be the binary relation on typed terms defined by
\[\alpha \approx_v \beta \quad\text{iff}\quad \simtext(\alpha,\beta)\ge \theta(v).\]
\end{definition}
Since only finitely many typed terms can bind a given variable under the bounded-generation assumption, the reflexive and transitive closure of $\approx_v$ therefore induces finitely many equivalence classes; we fix a deterministic total order $\preceq$ on typed terms, and write $\rep_v(\alpha)$ for the $\preceq$ -minimal representative of the equivalence class of $\alpha$ under $\approx_v$.

\begin{definition} 
For a variable $v$ and typed terms $\alpha,\beta$, we define \emph{canonical compatibility},
\[\compact{v}{\alpha}{\beta}=\top \quad\text{iff}\quad\rep_v(\alpha)=\rep_v(\beta).\]
\end{definition}

\begin{definition}
\label{def:compat-union}
Two typed solution mappings $\mu_1$ and $\mu_2$ are \emph{compatible}, written
$\mu_1\sim\mu_2$, iff for every variable
$v\in\domain{\mu_1}\cap\domain{\mu_2}$,
\[\compact{v}{\mu_1(v)}{\mu_2(v)}=\top.\]
If $\mu_1\sim\mu_2$, their \emph{canonical union} $\mu_1\uplus\mu_2$ is the
mapping with domain $\domain{\mu_1}\cup\domain{\mu_2}$ defined by
\[(\mu_1\uplus\mu_2)(v)=
\begin{cases}
    \rep_v(\mu_1(v)) & \text{if } v\in\domain{\mu_1}\cap\domain{\mu_2},\\
    \mu_1(v) & \text{if } v\in\domain{\mu_1}\setminus\domain{\mu_2},\\
    \mu_2(v) & \text{if } v\in\domain{\mu_2}\setminus\domain{\mu_1}.
\end{cases}\]
If $\mu_1\not\sim\mu_2$, the union $\mu_1\uplus\mu_2$ is undefined.
\end{definition}

By slight abuse of notation, for a graph pattern $\ma P$ and mapping $\mu$, we write $\mu(\ma P)$ for the set of triples obtained by substituting the variables $\var{\ma P}$ in $\ma P$ according to $\mu$.

\paragraph{\bgps evaluation.\,}
 Let $\ma P$ be a \bgps and \ma G an RDF graph . The evaluation of $\ma P$ over $\ma G$ is the set,
\[\evalTwo{\ma P}{\ma G} = \{\, \mu \mid \domain{\mu}=\var{\ma P}\wedge \mu(\ma P)\subseteq\ma G\,\},
\]
where each binding $\mu(v)=(t,\sourceRdf)$ with $t\in\ma T$.

\paragraph{Generative patterns.}
We model a generative pattern as a set-valued mapping that transforms an instantiated prompt into finitely many tuples of generated terms via a generative step followed by parsing, while leaving the surrounding query algebra unchanged.

A model specifier $\ma M$ induces finite-valued operations $\Call{\ma M}:\Sigma^*\to 2^{\Sigma^*}$ and $\Parse{\ma M}$, where $\Call{\ma M}$ returns answer strings and $\Parse{\ma M}$ returns tuples of generated terms. For every fixed query instance $(\ma P,\ma G)$, we assume that iterating the generative occurrences of $\ma P$ from the RDF bindings in $\ma G$ yields finitely many generated terms. For complexity results, we assume constants $n_1,n_2$ such that \[ |\Call{\ma M}(s)|\le n_1 \qquad\text{and}\qquad |\Parse{\ma M}(a)|\le n_2 \] for all relevant $\ma M,s,a$, and that generated terms have bounded representation length.

\begin{definition}\label{defn-genpattern}
Let $g=\genbind{\pi(X)}{Y}{\ma M}$, and let $Y=\langle y_1,\dots,y_k\rangle$ be a fixed ordering, where $k=\size{Y}$. Let $\Call{\ma M}:\Sigma^*\to 2^{\Sigma^*}$ and $\Parse{\ma M}:\Sigma^*\to 2^{\ma T_{\tx{gen}}^{|Y|}}$ return finite sets.  Let  \ma G be an RDF graph and $S$  a context supply, \ie a set of typed solution mappings, such that every $\mu_c\in S$ satisfies $X\subseteq\domain{\mu_c}$. The evaluation of $g$  with respect to  $S$ over \ma G is the set:
\[\evalThree{g}{\ma G}{S} =
\bigcup_{\mu_c\in S}
\left\{
\mu_c\uplus\nu
 \middle|\
\begin{array}{l}
\m a\in\Call{\ma M}(\pi^{[\mu_c]}),\\
(v_1,\dots,v_k)\in\Parse{\ma M}(\m a),\\
\nu=\{y_i\mapsto (v_i,\sourceGen)\mid 1\le i\le k\},\\
\mu_c\sim\nu,
\end{array}
\right\}\]

where $\pi^{[\mu_c]}$ is the prompt obtained from $\pi$ by replacing each placeholder in $X$ by the lexicalization of its binding under $\mu_c$.
\end{definition}



\begin{example}
Let $\ma G=\{\tx{:Paris a :City}\}$ and the query:
\begin{verbatim}
SELECT ?x ?y ?z WHERE {
  ?x a :City .
  GenOp("Describe the city ?x" AS ?y,
                              GPT-4o).
  }
\end{verbatim}
By evaluating \bgps, we get context supply $\mu_c=\{?x\mapsto(\tx{:Paris},\sourceRdf)\}$ for the \tx{GenOp}. Then, prompt instantiation produces $\pi^{[\mu_c]}=\text{``Describe the city Paris''}$. Assume,
\[\Call{\tx{GPT-4o}}(\pi^{[\mu_c]})=\{\text{``capital of France''},\text{``cultural center''}\},\]
and
\begin{multline*}
    \Parse{\tx{GPT-4o}}(\text{``capital of France''})=\{(t_1)\}\\
\Parse{\tx{GPT-4o}}(\text{``cultural center''})=\{(t_2)\}.
\end{multline*}
By the semantics of generative patterns,
\begin{multline*}
    \evalThree{g}{\ma G}{\{\mu_c\}}=\big\{
\mu_c\uplus\{?y\mapsto(t_1,\sourceGen)\},\\
\mu_c\uplus\{?y\mapsto(t_2,\sourceGen)\}
\big\}.
\end{multline*}
\end{example}

\subsubsection{Filter .\,}

Let $\mu$ be a typed solution mapping. The truth value of a filter formula $\ma F$ under $\mu$, denoted $\ma F^{\mu}\in\{\top,\bot,\varepsilon\}$ (where $\varepsilon$ denotes an error), is defined inductively as follows.

\begin{enumerate}
    \item ${\bound{v}}^{\mu}=\top$ if $v\in\domain{\mu}$, and ${\bound{v}}^{\mu}=\bot$ otherwise.

    \item For $c\in\ma T$, $(v=c)^{\mu}=\varepsilon$ if $v\notin\domain{\mu}$;  
    $(v=c)^{\mu}=\top$ if $\mu(v)=(c,\sourceRdf)$ or $\compact{v}{\mu(v)}{(c,\sourceRdf)}=\top$;  otherwise, $(v=c)^{\mu}=\bot$.

    \item $(v=v')^{\mu}=\varepsilon$ if $\{v,v'\}\nsubseteq\domain{\mu}$;  
    $(v=v')^{\mu}=\top$ if $\mu(v)=\mu(v')=(t,\sourceRdf)$ for some $t\in\ma T$, or if $\compact{v}{\mu(v)}{\mu(v')}=\top$;  
    otherwise, $(v=v')^{\mu}=\bot$.

    \item $(\neg \ma F)^{\mu}=\top$ if $\ma F^{\mu}=\bot$;  
    $(\neg \ma F)^{\mu}=\bot$ if $\ma F^{\mu}=\top$;  
    $(\neg \ma F)^{\mu}=\varepsilon$ if $\ma F^{\mu}=\varepsilon$.

    \item $(\ma F_1\wedge \ma F_2)^{\mu}={\ma  F_1}^{\mu}\wedge {\ma F_2}^{\mu}$.
\end{enumerate}

\paragraph{Graph Patterns.\,}
The evaluation of a graph pattern $\ma P$ over an RDF graph $\ma G$, denoted
$\evalThree{\ma P}{\ma G}{S}$, returns a set of typed solution mappings and is defined recursively below relative to a context supply $S$. When the context supply is fixed or immaterial, we write $\evalTwo{\ma P}{\ma G}$ as shorthand for $\evalThree{\ma P}{\ma G}{S}$. The clauses below define a \emph{single-pass} (one-step) evaluation; cyclic generative dependencies are handled by the fixpoint semantics in Section~\ref{sec:fixpoint}.

\begin{enumerate}
  \item If $\ma P$ is a \bgps, then $\evalTwo{\ma P}{\ma G}$ is defined as above. If $\ma P$ is a \tx{GenOp} pattern, its single-pass evaluation is $\evalThree{\ma P}{\ma G}{S}$, where $S$ provides the solution mappings used to instantiate the  placeholder variables of its prompt template.
  \item $\evalTwo{\filter{\ma P}{\ma F}}{\ma G} = \{\mu\in\evalTwo{\ma P}{\ma G} \mid F^\mu=\top\}$;
  \item $\evalTwo{\proj{\ma P}{L}}{\ma G}=\{\mu_{|L}\mid\mu\in\evalTwo{\ma P}{\ma G}\}$;
  \item 
  $\begin{aligned}
  \evalTwo{\join{\ma P_1}{\ma P_2}}{\ma G} =  \{\mu_1\uplus\mu_2 \mid \mu_1\in\evalTwo{\ma P_1}{\ma G}, \mu_2\in\evalTwo{\ma P_2}{\ma G}, \\ \mu_1\sim\mu_2\};
  \end{aligned}$
  \item $\evalTwo{\union{\ma P_1}{\ma P_2}}{\ma G}= \evalTwo{P_1}{\ma G}\cup \evalTwo{P_2}{\ma G}$;
  \item $\begin{aligned}\evalTwo{\diff{\ma P_1}{\ma P_2}{\ma F}}{\ma G} =\{\mu_1\in\evalTwo{\ma P_1}{\ma G} \mid \forall \mu_2\in\evalTwo{\ma P_2}{\ma G}: \\ \mu_1\not\sim\mu_2\,\text{or}\,  \ma F^{\mu_1\uplus \mu_2}=\bot \}\end{aligned}$;
  \item $\evalTwo{\opt{\ma P_1}{\ma P_2}{\ma F}}{\ma G}= \evalTwo{\join{\ma P_1}{\ma P_2}}{\ma G}\cup\evalTwo{\diff{\ma P_1}{\ma P_2}{\ma F}}{\ma G}$;
  \item $\begin{aligned} \evalTwo{\minus{\ma P_1}{\ma P_2}}{\ma G} = \{\mu_1\in\evalTwo{\ma P_1}{\ma G} \mid \forall \mu_2\in\evalTwo{\ma P_2}{\ma G}: \\ \text{either}\, \mu_1\not\sim\mu_2\, \text{or}\, \domain{\mu_1}\cap\domain{\mu_2
  }=\varnothing\};
  \end{aligned}$
\end{enumerate}
where $\mu_{|L}$ is the restriction of $\mu$ to $L$. If $g \preceq \ma P$, the context supply for $g$ in a single evaluation step is the set $S = \evalTwo{\ma P \setminus \{g\}}{\ma G}$.

\begin{proposition}
Let $\ma P$ be a well-formed graph pattern that contains no occurrence of \tx{GenOp}. Then, for any RDF graph $\ma G$, the evaluation of $\ma P$ under the extended semantics coincides with standard SPARQL evaluation under set semantics: $\evalTwo{\ma P}{\ma G} = {\evalTwo{\ma P}{\ma G}}|_{\text{standard SPARQL}}$.
\end{proposition}

\section{Fragments and Fixpoint Semantics}
\label{sec:fixpoint}

The preceding semantics characterize the meaning of generative SPARQL queries independently of any evaluation strategy. Section~\ref{sec:exec-fixpoint} presents an executable realization of these semantics under bounded generation.

\subsection{Dependency Graph }
 Let $\ma P$ be a generative SPARQL graph pattern and let $\genPattern{\ma P}=\{g\mid g\preceq \ma P\text{ and }g\text{ is a }\tx{GenOp}\text{ occurrence}\}$ denote the set of \tx{GenOp} occurrences in $\ma P$. For each $g=\genbind{\pi(X)}{Y}{\ma M}$ in $\genPattern{\ma P}$, let $\holder{g}$ denote its set of placeholder variables $X$ and $\target{g}$ its set of output variables $Y$.

\begin{definition}[Dependency graph]
The \emph{dependency graph} of $\ma P$ is the directed graph
$\directedGraph{\ma D}_{\ma P} = (\genPattern{\ma P}, \to_{\ma P})$,
where for $g_i, g_j \in \genPattern{\ma P}$,
\[
g_i \to_{\ma P} g_j
\quad\text{iff}\quad
\target{g_i} \cap \holder{g_j} \neq \varnothing .
\]
\end{definition}

\begin{definition}
A nonempty subset $C \subseteq \genPattern{\ma P}$ is a \emph{strongly connected component (\tx{SCC})} of $\ma P$ iff:
\begin{enumerate}
\item for all $g, g' \in C$, there exists a finite sequence $g = g_0, \dots, g_k = g'$ in $C$ such that $g_i \to_{\ma P} g_{i+1}$ for all $i < k$, and a sequence from $g'$ to $g$; and
\item $C$ is maximal \wrt set inclusion under this property.
\end{enumerate}
The \tx{SCC}s of $\ma P$ form a partition of $\genPattern{\ma P}$.
\end{definition}

\begin{definition}[Acyclicity]
A graph pattern $\ma P$ is \emph{acyclic} iff $\directedGraph{\ma D}_{\ma P}$ is acyclic.
\end{definition}

\noindent If $\ma P$ is acyclic, all \tx{GenOp} occurrences can be evaluated in a
single pass following a topological order of
$\directedGraph{\ma D}_{\ma P}$, and the result coincides with the
fixpoint semantics induced by $\Consequence{\ma P}$ (Cf. Section \ref{sec:lfp}).

\begin{definition}(Signed Dependency Edge)
Let $g_i \to_{\ma P} g_j$ be a dependency edge and let $V_{ij}= \target{g_i} \cap \holder{g_j}$. Let $g_j\preceq \ma P$. We classify the edge by the syntactic occurrences of variable $v\in V_{ij}$ in the surrounding pattern $\ma P \setminus \{g_j\}$.
\begin{enumerate}
    \item The edge is \emph{positive} if every syntactic occurrence of every
  $v\in V_{ij}$ in $\ma P\setminus\{g_j\}$ is in a \emph{positive position}, namely
  under only the operators \tx{Join}, \tx{Union}, \tx{Proj}, the \tx{Join} branch  of $\tx{Opt}$ or within a
  \tx{Filter} where $v$ does not occur under $\neg$.
  \item The edge is \emph{negative} if there exists some $v\in V_{ij}$ that has a
  syntactic occurrence in $\ma P\setminus\{g_j\}$ in a \emph{negative position}, \ie  within the scope of \tx{Minus}, \tx{Diff}, the \tx{Diff}  branch of $\tx{Opt}$ or under $\neg$ in a
  \tx{Filter}.
\end{enumerate}
\end{definition}

\begin{example}
Let $g_1$ and $g_2$ be the two \tx{GenOp} occurrences in
\begin{verbatim}
SELECT ?x ?y ?z WHERE {
  ?x a :City .
  GenOp("Find topic for ?x" AS ?y, GPT-4o).
  GenOp("Find city for ?y" AS ?z, GPT-4o).
}
\end{verbatim}
The variable $?y$ occurs as a placeholder in the prompt of $g_2$ and only in
positive positions in the surrounding pattern of $g_2$; hence $g_1\to g_2$ is
a positive dependency. In the variant obtained by adding
\begin{verbatim}
MINUS { ?x :hasTopic ?y . }
\end{verbatim}
the variable $?y$ occurs under a negative operator in the surrounding pattern
of $g_2$. Thus, $g_1\to g_2$ is a negative dependency, and any stratification
as in Definition~\ref{defn_stratification} must satisfy
$\sigma(g_1)<\sigma(g_2)$.
\end{example}

\begin{definition}\label{defn_stratification}
A \emph{stratification} of $\ma P$ is a function
$\sigma : \genPattern{\ma P} \to \mathbb{N}$ such that for every
$g_i \to_{\ma P} g_j$:
(i) if the edge is positive then $\sigma(g_i) \le \sigma(g_j)$;
(ii) if the edge is negative then $\sigma(g_i) < \sigma(g_j)$.
The pattern $\ma P$ is \emph{stratified} if such a function exists.
\end{definition}

Acyclic patterns admit single-pass evaluation. Stratified patterns admit stratum-wise fixpoint semantics. Non-stratified patterns require well-founded semantics.

\subsection{Consequence Operator}
\label{sec:immediate-operator}

We handle recursion among \tx{GenOp} patterns via an immediate consequence operator, assuming bounded generation so that, for fixed $\ma P$ and $\ma G$, the set of generated terms, and thus relevant typed solution mappings, is finite.

Let $\Omega_{\ma P}$ be the finite set of all typed solution mappings over $\var{\ma P}$ whose values range over $\ma T \cup \ma T_{\tx{gen}}$. 

To expose recursive dependencies, evaluation is parameterized by a \emph{context supply} $S\subseteq\Omega_{\ma P}$. 
For a \tx{GenOp} pattern $g$, only those mappings in $S$ that provide bindings for all placeholders $\holder{g}$ are useful to instantiate the prompt, \ie context set,
\[\ContextS{S}{g} = \{\mu \in S \mid \holder{g}\subseteq\domain{\mu}\},\]
and the evaluation of $g$ relative to $S$ is then $\evalThree{g}{\ma G}{\ContextS{S}{g}}$.

\begin{definition}
The immediate consequence operator of $\ma P$ is the function
\[\Consequence{\ma P} : 2^{\Omega_{\ma P}} \to 2^{\Omega_{\ma P}},
\,\,
\Consequence{\ma P}(S) := \evalThree{\ma P}{\ma G}{S}.\]
\end{definition}

The $\Consequence{\ma P}(S)$ evaluates the query once while allowing each
\tx{GenOp} to draw its placeholder bindings from the current supply $S$.
Recursive behavior is captured by iterating $\Consequence{\ma P}$ to a fixpoint.

\begin{proposition}
If $S,S'\subseteq\Omega_{\ma P}$ satisfy
$\ContextS{S}{g}=\ContextS{S'}{g}$ for every \tx{GenOp} pattern $g$ in $\ma P$,
then $\Consequence{\ma P}(S)=\Consequence{\ma P}(S')$.
\end{proposition}

\subsection{Monotonicity and Least Fixpoint Semantics}
\label{sec:lfp}
We identify conditions for a least-fixpoint semantics by establishing monotonicity of the consequence operator.


A generative SPARQL pattern $\ma P$ is monotone w.r.t.\ the context supply if, for every \tx{GenOp} $g\preceq\ma P$, variables in $\target{g}$ occur only in positive positions in $\ma P$, so added generated bindings cannot invalidate existing solution mappings.

\begin{lemma}
\label{lem:monotone}
If $\ma P$ is monotone with respect to the context supply, then the consequence
operator $\Consequence{\ma P}$ is monotone on the complete lattice
$(2^{\Omega_{\ma P}},\subseteq)$:
\[
\forall S,S'\in 2^{\Omega_{\ma P}}:\, S \subseteq S' \;\longrightarrow\; \Consequence{\ma P}(S) \subseteq \Consequence{\ma P}(S').
\]
\end{lemma}

\begin{proof}
Expanding the context supply $S$ only increases the set of mappings available to
instantiate placeholders of \tx{GenOp} patterns. Since each \tx{GenOp} is defined
as a union over all admissible contexts, additional contexts can only yield
additional generated mappings. By the monotonicity condition, variables
introduced by \tx{GenOp} patterns occur only in positive positions. Hence,
previously derived mappings cannot be invalidated, and
$\Consequence{\ma P}$ is monotone.
\end{proof}

If $\Consequence{\ma P}$ is monotone, it admits a least fixpoint with respect to
set inclusion. Under the bounded-generation assumption, this fixpoint can be
computed by iterating $\Consequence{\ma P}$ from the empty set.

\begin{definition}[Least fixpoint semantics]
Let $\Consequence{\ma P}$ be monotone, and the ascending sequence
$\{S_n\}_{n \geq 0}$ is defined by
\[
S_0 = \varnothing, \qquad S_{n+1} = \Consequence{\ma P}(S_n).
\]
The \emph{least fixpoint} of $\Consequence{\ma P}$ is $\lfp(\Consequence{\ma P}) = \bigcup_{n \geq 0} S_n$,
and the meaning of $\ma P$ under least-fixpoint semantics is
$\evalTwo{\ma P}{\ma G}^{\lfp} := \lfp(\Consequence{\ma P})$.
\end{definition}

\begin{theorem}(Existence and characterization)
\label{thm:lfp}
If $\Consequence{\ma P}$ is monotone on $(2^{\Omega_{\ma P}},\subseteq)$, then it
has a least fixpoint, and
\[
\lfp(\Consequence{\ma P}) = \bigcup_{n \geq 0} \Consequence{\ma P}^n(\varnothing).
\]
Moreover, if $\Omega_{\ma P}$ is finite, the sequence stabilizes after at most
$|\Omega_{\ma P}|$ iterations.
\end{theorem}

For acyclic patterns, the iteration stabilizes after one step (see \ref{app:acyclic}), and
the fixpoint semantics coincides with standard compositional evaluation.

\begin{theorem}
Let $S^\ast = \lfp(\Consequence{\ma P})$. If $S$ is any set of typed solution
mappings such that $\Consequence{\ma P}(S) \subseteq S$, then $S^\ast \subseteq S$.
Thus, $S^\ast$ is the minimal fixpoint model of $\ma P$.
\end{theorem}


\subsection{Stratified Semantics }
\label{sec:stratified}

Let $\ma P$ be a generative SPARQL pattern and let $\sigma:\genPattern{\ma P}\to\mathbb N$ be a stratification. For $i\in\mathbb N$, let
\[\Gen^{\leq i}(\ma P)=\{g\in\genPattern{\ma P}\mid \sigma(g)\leq i\}.\]
Let $\ma P^{\leq i}$ denote the pattern obtained from $\ma P$ by \emph{disabling} all \tx{GenOp} occurrences not in $\Gen^{\leq i}(\ma P)$, that is, replacing each such occurrence by a pattern whose evaluation is empty for all context supply $S$, while leaving the remaining algebraic structure unchanged.

\begin{definition}(Stratum-wise operator)
For each $i$, let $\Consequence{\ma P^{\leq  i}}:2^{\Omega_{\ma P}}\to 2^{\Omega_{\ma P}}$ defined by $\Consequence{\ma P^{\leq  i}}(S):= \evalThree{\ma P^{\leq  i}}{\ma G}{S}$.
\end{definition}

\begin{definition}(Stratified semantics)
Let $\ma P$ be stratified with stratification $\sigma:\genPattern{\ma P}\to\mathbb N$. Let $S_{-1} := \varnothing$. For each stratum index $i=0,1,2,\dots$, consider the operator
\[S \mapsto \Consequence{\ma P^{\leq i}}(S \cup S_{i-1}),\]
and let $S_i$ be its least fixpoint:
\[S_i := \lfp\big( S \mapsto \Consequence{\ma P^{\leq i}}(S \cup S_{i-1}) \big).\]
Let $m = \max\{\sigma(g)\mid g\in\genPattern{\ma P}\}$, with $m=0$ if $\genPattern{\ma P}=\varnothing$. The meaning of $\ma P$ under stratified semantics is given by $\evalThree{\ma P}{\ma G}{\mathrm{strat}}= S_m$.
\end{definition}


\begin{theorem}
\label{thm:strat}
If $\ma P$ is stratified and $\Omega_{\ma P}$ is finite (bounded generation),
then the stratified semantics
$\evalTwo{\ma P}{\ma G}^{\mathrm{strat}}$
is well-defined and unique.
\end{theorem}



\begin{theorem}
Let $\ma P$ be stratified.
For each stratum $i$, the least fixpoint $S_i$ is independent of the order in
which \tx{GenOp} occurrences in stratum $i$ are evaluated.
\end{theorem}


\subsection{Well-Founded Semantics for Non-Stratified Patterns}
\label{sec:wfs}

For query patterns that are not stratified, we adopt the standard three-valued semantics based on an alternating fixpoint construction \cite{przymusinski1990well,van1989alternating}, instantiated to typed solution mappings.


Let $\Cand{\ma P} \subseteq \Omega_{\ma P}$ be a finite candidate set.
assumed to be a sound over-approximation of all typed solution mappings that may arise during evaluation of $\ma P$ under bounded generation. A partial interpretation is a pair $(I^+, I^-)$ with $I^+, I^- \subseteq \Cand{\ma P}$ and $I^+ \cap I^- = \varnothing$, where mappings in $I^+$ are true, mappings in $I^-$ are false, and all others are undefined.

\begin{definition}
Let $J \subseteq \Cand{\ma P}$.
The \emph{reduct consequence operator} $\Consequence{\ma P}^J : 2^{\Cand{\ma P}} \to 2^{\Cand{\ma P}}$ is defined as
\[\Consequence{\ma P}^J(S) \coloneqq \evalThree{\ma P}{\ma G}{S},\]
where all negative constructs (\ie \tx{MINUS}, \tx{DIFF}, and negated filters) are evaluated with respect to the fixed set $J$.
\end{definition}

Under fixed $J$, the operator $S \mapsto \Consequence{\ma P}^J(S)$ is monotone,
since all negative tests are evaluated against a constant context.

\begin{definition}[Alternating fixpoint sequence]
\label{def:alt}
Let $(I_0^+, I_0^-) \coloneqq (\varnothing, \varnothing)$.
For $n \geq 0$, 
\begin{align*}
I_{n+1}^+ &\coloneqq
\lfp\big(S \mapsto \Consequence{\ma P}^{\Cand{\ma P}\setminus I_n^-}(S)\big),\\
I_{n+1}^- &\coloneqq
\Cand{\ma P} \setminus
\lfp\big(S \mapsto \Consequence{\ma P}^{I_{n+1}^+}(S)\big).
\end{align*}
\end{definition}

\begin{theorem}(Convergence and uniqueness)
\label{thm:wfs}
If $\Cand{\ma P}$ is finite, then the alternating fixpoint sequence in
Definition~\ref{def:alt} stabilizes after finitely many steps at a pair
$(I^+, I^-)$.
The induced three-valued interpretation is unique and coincides with the
well-founded model of $\ma P$ over $\ma G$ relative to $\Cand{\ma P}$.
\end{theorem}


\noindent



Non-stratified generative SPARQL patterns, \ie patterns with cyclic negative dependencies, fall outside the monotone fragment. Under the bounded-generation assumption, such patterns admit a three-valued well-founded semantics obtained by the standard alternating-fixpoint construction over the finite mapping lattice; the construction follows classical well-founded semantics.

\section{Algebraic Equivalences and Rewrite Rules}
\label{sec:rewrites}
We prove equivalence results for the monotone or stratified fragments, where denotational and fixpoint semantics coincide.

\begin{definition}
Two generative SPARQL patterns $\ma P$ and $\ma P'$ are \emph{equivalent}, written $\ma P \equiv \ma P'$, if they produce the same set of solution mappings for every RDF graph $\ma G$ and every context-supply set $S$, \ie
$\evalThree{\ma P}{\ma G}{S} = \evalThree{\ma P'}{\ma G}{S}$.
\end{definition}


\paragraph{SPARQL-specific equivalences.} Standard SPARQL equivalences that preserve the context supplied to \tx{GenOp} patterns remain valid: 
\begin{multline*}
    \join{\ma P_1}{\ma P_2} \equiv \join{\ma P_2} {\ma P_1};\,\\
 \join{(\join{\ma P_1}{\ma P_2})}{\ma P_3} \equiv \join{\ma P_1}{(\join{\ma P_2}{\ma P_3})};\\
\join{\ma P_1}{(\union{\ma P_2}{\ma P_3})} \equiv \\ \union{(\join{\ma P_1}{\ma P_2})}{(\join{\ma P_1}{\ma P_3})};\\
\filter{P}{\ma F_1\wedge \ma F_2}\equiv \filter{\filter{P}{\ma F_2}}{\ma F_1};\\
\filter{\join{\ma P_1}{\ma P_2}}{F}\equiv \join{\filter{\ma P_1}{F}}{\ma P_2}\\ \text{if }\var{F}\subseteq\var{\ma P_1}; \,\text{and}\\
\minus{\ma P_1}{\ma P_2}\equiv \ma P_1\ \ \text{if }\var{\ma P_1}\cap\var{\ma P_2}=\varnothing.
\end{multline*}
The proofs follow those of SPARQL, with equality on shared variables replaced by compatibility (cf. Appendix~\ref{sec:equivalences}).

\paragraph{\tx{GenOp}-specific rewrites.}
A generative pattern $g=\genbind{\pi(X)}{Y}{\ma M}$ can be evaluated only using context mappings that bind $X$; moving $g$ can therefore change the context set and the result.

We first illustrate why unrestricted reordering of GenOp patterns is unsound.
\begin{example}[Unsafe GenOp reordering]
\label{ex:unsafe-reorder}
Let $g=\genbind{\pi(\{?x\})}{\{?y\}}{M}$ and consider
\[\ma P_1=\{(?s\ \tx{:p}\ ?x)\},\,\,\ma P_2=\{(?s\ \tx{:q}\ ?z)\}.\]
Assume $\ma G$ is such that $\evalTwo{\ma P_1}{\ma G}\neq\varnothing$ and $\evalTwo{\ma P_2}{\ma G}\neq\varnothing$, but no mapping in $\evalTwo{\ma P_2}{\ma G}$ binds $?x$.
In $\join{(\join{\ma P_1}{\ma P_2})}{g}$, the left subpattern $\join{\ma P_1}{\ma P_2}$ yields mappings binding $?x$, so the context set for $g$ is non-empty and $g$ may produce bindings for $?y$.
In contrast, in $\join{\ma P_1}{(\join{\ma P_2}{g})}$, the left subpattern $\ma P_2$ yields no mapping binding $?x$, hence the context set for $g$ is empty and $\evalTwo{ \join{P_2}{g}}{\ma G}=\varnothing$.
Thus, reordering a \tx{GenOp} across joins is not semantics-preserving in general.
\end{example}

We next give sufficient conditions guaranteeing that common rewrites involving \tx{GenOp} patterns are sound.

\begin{proposition}
\label{prop:genop-base}
Let $g=\genbind{\pi(\varnothing)}{Y}{M}$ be a base-mode \tx{GenOp}.
Then, $\,\forall\,\ma G,\forall\,S: \evalThree{g}{\ma G}{S}=\evalThree{g}{\ma G}{\{\varnothing\}}$.
\end{proposition}
\begin{proof}
By the base-mode convention, whenever $\holder{g}=\varnothing$ the evaluation of $g$
ignores the supplied context and uses the fixed context set $\{\varnothing\}$.
Hence $\evalThree{g}{\ma G}{S}=\evalThree{g}{\ma G}{\{\varnothing\}}$ for all $S$.
\end{proof}

\begin{corollary}
\label{cor:genop-base-join}
Let $g=\genbind{\pi(\varnothing)}{Y}{M}$ and $\Omega_g = \evalThree{g}{\ma G}{\varnothing}$. Let $\ma P$ be a query pattern that contains no \tx{GenOp} occurrences. Then, $\forall\,\ma G,\forall\,S,S':$
\[\evalThree{\join{\ma P}{g}}{\ma G}{S}
\;=\; \evalThree{\join{\ma P}{g}}{\ma G}{S'} \;=\;\evalTwo{\join{\ma P}{\Omega_g}}{\ma G}.\]
\end{corollary}



\begin{proposition}
\label{prop:union-distribution-branch}
Let $\ma P_1,\ma P_2$ be patterns and $g$ a \tx{GenOp}. For all $\ma G,S$ and $i\in\{1,2\}$, let $S_i := \evalThree{\ma P_i}{\ma G}{S}$ and $\ContextS{S_i}{g} := \{\mu\in S_i \mid \holder{g}\subseteq\domain{\mu}\}$. If $\evalThree{g}{\ma G}{\ContextS{S_1}{g}\cup \ContextS{S_2}{g}} =
\evalThree{g}{\ma G}{\ContextS{S_1}{g}}\cup\evalThree{g}{\ma G}{\ContextS{S_2}{g}}$, then
$\evalThree{\join{(\union{\ma P_1}{\ma P_2})}{g}}{\ma G}{S} = \evalThree{\union{(\join{\ma P_1}{g})}{(\join{\ma P_2}{g})}}{\ma G}{S}.$
\end{proposition}


\begin{proof}
In $\join{(\union{\ma P_1}{\ma P_2})}{g}$, the \tx{GenOp} $g$ sees exactly the union
of the solution mappings produced by the two branches, so its context is $(\ContextS{S_1}{g}\cup \ContextS{S_2}{g})$. By the separability assumption, evaluating $g$ on this combined context produces precisely the union of the outputs obtained by evaluating $g$ on each branch context separately. Joining these outputs back with their respective branch solutions, therefore, yields exactly the union of $\join{\ma P_1}{g}$ and $\join{\ma P_2}{g}$.
\end{proof}

\begin{proposition}
\label{prop:filter-pushdown}
Let $g=\genbind{\pi(X)}{Y}{\ma M}$ and let $\ma F$ be a filter with $\var{\ma F}\subseteq X$. Let $\ma P$ be a pattern with $Y\cap\var{\ma P}=\varnothing$ and 
$\var{\ma F}\subseteq\var{\ma P}$.
Then,
\[\evalTwo{\join{\filter{\ma P}{\ma F}}{g}}{\ma G}
=\evalTwo{\filter{\join{\ma P}{g}}{\ma F}}{\ma G}.\]
\end{proposition}

\begin{proof}
Let $S=\evalTwo{\ma P}{\ma G}$ and $S^+=\{\mu\in S\mid \ma F^\mu=\top\}$.
In $\join{\filter{\ma P}{\ma F}}{g}$, the context supply for $g$ is $S^+$; in
$\join{\ma P}{g}$, it is $S$. By the \tx{GenOp} clause,
$\evalThree{g}{\ma G}{S}=\bigcup_{\mu\in S}\evalThree{g}{\ma G}{\{\mu\}}$, hence the
extensions generated from $S^+$ are exactly those generated from the $S^+$, contexts
inside the run on $S$. Every generated mapping has the form $\mu\uplus\nu$ and does not change the truth value of $\ma F$ because $\var{\ma F}\subseteq X$ and $g$ only adds fresh bindings on $Y$ (up to representative normalization). Thus, $\ma F^\mu=\top$ iff $\ma F^{\mu\uplus\nu}=\top$, so filtering after the join removes precisely the
contributions from $\mu\notin S^+$, yielding the same result as filtering before.
\end{proof}

\medskip
\begin{example}
 Let $g = \GenOp(\pi(\{?x\}), \{?y\}, M)$, and suppose $\ma P_1$ binds $?x$ while $\ma P_2$ does not. Then:
\[\evalTwo{\join{\ma P_1}{(\join{\ma P_2}{g})}}{\ma G} \neq \evalTwo{\join{(\join{\ma P_1}{\ma P_2})}{g}}{\ma G
}\]   
\end{example}

\begin{proposition}(Safe Join Reordering )
\label{prop:genop-extract}
Let $g=\genbind{\pi(X)}{Y}{\ma M}$ and  let $\ma P_1,\ma P_2$ be patterns such that $X\subseteq \var{\ma P_1}$ and  $\var{\ma P_2}\cap X=\varnothing$. Then, for all graph $\ma G$ and all context-supply set $S$,
\[\evalThree{\join{(\join{\ma P_1}{\ma P_2})}{g}}{\ma G}{S}\;=\;
\evalThree{\join{(\join{\ma P_1}{g})}{\ma P_2}}{\ma G}{S}.\]
\end{proposition}

Topological evaluation orderings preserve semantics when the generative dependency graph is acyclic.
\begin{theorem}
\label{thm:topo-invariance}
Let $\ma P$ be a pattern whose generative dependency graph
$\directedGraph{\ma D}_{\ma P}$ is acyclic. Let $\gamma_1$ and $\gamma_2$ be two topological orders of $\genPattern{\ma P}$. For each order, let $\Omega_{\gamma_1}$ (resp.\ $\Omega_{\gamma_2}$) be the set of typed solution mappings obtained by evaluating all \tx{GenOp} occurrences according to that order, starting from the same non-generative base evaluation (cf.\ Appendix~\ref{app:acyclic}). Then, $\Omega_{\gamma_1} = \Omega_{\gamma_2}$.
\end{theorem}

\section{Executable Fixpoint Evaluation}\label{sec:exec-fixpoint}
Recursive evaluation is declaratively given by a fixpoint construction. With
generative patterns, consequences depend on finite-valued external model calls.
We therefore use candidate-model validation for executable fixpoint evaluation,
following the treatment of external operators in logic
programs~\cite{eiter2006effective}.

Let $\ma P$ be a monotone query pattern with immediate consequence operator $\Consequence{\ma P}$. Declaratively, the meaning of $\ma P$ is $\lfp(\Consequence{\ma P})$. Executable evaluation checks the \emph{self-support} condition $\mu\in \Consequence{\ma P}(\{\mu\})$ on finitely many candidates; soundness holds unconditionally, while completeness is relative to candidate coverage.

\emph{Evaluation Roadmap.} Query evaluation proceeds by invoking $\tx{EvalSCC}(\ma G,\ma P,C,\eta,K,B,T,R)$ (Algorithm~\ref{alg:evalscc}) for each \tx{SCC} $C$ of $\ma P$ and each outer binding $\eta$. For a fixed call, $\tx{EvalSCC}$ first identifies the generated variables $Y_C$ and required placeholders $X_C$ (Sect. 6.1), then uses bounded
proposal ($K,B$) to enumerate a finite candidate set $\tx{Cand}_{C}(\eta)\subseteq\Omega_{\ma P}$ (Sect. 6.2). Each candidate is subsequently checked for self-support by repeated validation ($R$), with optional bounded repairs ($T$). The output $R_C(\eta)$ consists exactly of those candidates that
satisfy self-support. A concrete execution trace is given in Appendix~\ref{app:example1-exec}.

\subsection{\tx{SCC}-local view and self-support}
Recursion in generative SPARQL arises only through generative patterns. We therefore localize recursion using strongly connected components (\tx{SCC}s).

Let $C$ be a \emph{\tx{SCC}} of the pattern $\ma P$, \ie $C \subseteq \genPattern{\ma P}$. Let,
\[Y_C = \bigcup_{g\in C}\target{g}\, \quad\text{and}\quad X_C=(\bigcup_{g\in C}\holder{g})\setminus Y_C.\]

\noindent Here, $Y_C$ are the variables generated within $C$, while $X_C$ are the external
placeholders required by operators in $C$. An \emph{outer binding} for $C$ is a
mapping $\eta$ with $X_C\subseteq\domain{\eta}$.


\begin{definition}
\label{def:self-support}
Let $C$ be an \tx{SCC} of $\ma P$. A mapping $\mu\in\Omega_{\ma P}$ is
\emph{self-supporting for $C$} (denoted $\tx{SS}_C(\mu)$) if
$X_C\subseteq\domain(\mu)$ and for every generative pattern
$g=\genbind{\pi(X)}{Y}{\ma M}\in C$,
\[
\mu \in \evalThree{g}{\ma G}{\{\mu_{|X}\}} .
\]
\end{definition}

The $\tx{SS}_C(\mu)$ requires that each \tx{GenOp} in the \tx{SCC}
reproduces the bindings already present in $\mu$ when evaluated on the
placeholder bindings induced by $\mu$ itself. This condition replaces
incremental rule firing as the operational criterion for recursive admissibility.

\subsection{Executable Interface}
We separate executable evaluation into \emph{proposal}, which heuristically enumerates a finite search space, and \emph{validation}, which enforces semantic correctness.
\begin{enumerate}
\item \emph{Validation Oracle.} For a generative pattern $g=\genbind{\pi(X)}{Y}{\ma M}$ and a \emph{total}
placeholder binding $\theta$ over $X$, we assume a Boolean procedure
\[\tx{Validate}(g,\theta,\nu;R)\in\{\top,\bot\}\]
that returns $\top$ iff, under fixed decoding, parsing, and normalization
settings repeated $R$ times, the extension $(\theta\uplus\nu)$ belongs to
$\evalThree{g}{\ma G}{\{\theta\}}$. Thus, \tx{Validate} implements a bounded
membership test for the denotation of $g$ on fully instantiated prompts.

\item \emph{Proposal Routine.} $\tx{Propose}(g,\theta,K)$ returns up to $K$ candidate bindings $\nu$ for variables in $\target{g}\cap Y_C$ under placeholder binding $\theta$. Proposal may be invoked on partially instantiated contexts by filling missing placeholders with fixed sentinel tokens (or by fill in the blanks \cite{donahue2020enablinglanguagemodelsblanks}); sentinels serve only to diversify candidate enumeration
and have no semantic status. Only fully instantiated bindings are validated.
\item \emph{Candidate sets and executable \tx{SCC} answers.} Let $C$ be an SCC and an $\eta$ outer binding.  For each $v\in Y_C$, maintain a finite domain $D[v]\subseteq \ma{T}_{\tx{gen}}\times\{\sourceGen\}$ with $\size{D[v]}\leq B$. These domains induce a finite candidate set
\begin{multline*}
    \tx{Cand}_{\,C}(\eta)= \{\eta\uplus\nu \ \mid \ \nu:Y_C\to \ma{T}_{\tx{gen}}\times\{\sourceGen\},\\ \forall v\in Y_C:\nu(v)\in D[v], \ \eta\sim\nu\Bigr\}.
\end{multline*}
By construction, $\tx{Cand}_{\,C}(\eta)\subseteq\Omega_{\ma{P}}$ ((typed solution mappings)). The candidate set is an operational search space and is not defined in terms of the consequence operator
$\Consequence{C}$; it may therefore contain mappings that are not semantically admissible.
The outer binding $\eta$ can be viewed as a singleton \emph{context supply}
$\{\eta\}\subseteq\Omega_{\ma{P}}$, consistent with the general evaluation
semantics. Restriction of context supplies to placeholders (notation $S_{|g}$)
is not used during candidate construction, since no generative operators are
semantically evaluated at that stage. Such restrictions arise only during
validation, where each candidate $\mu$ induces the placeholder binding
$\mu_{|X}$ used to check membership
$\mu\in\evalThree{g}{\ma{G}}{\{\mu_{|X}\}}$.

Give a finite candidate set $\tx{Cand}_{\,C}(\eta)$, let 
\[
R_C(\eta)=\{\mu_{|Y_C}\mid \mu\in\tx{Cand}_{\,C}(\eta)\ \wedge\ \tx{SS}_C(\mu)\},
\]
Thus, executable evaluation $R_C(\eta)$ returns exactly the $Y_C$-projections of candidates
that satisfy self-support.
\end{enumerate}

\paragraph{Finite candidate coverage and exactness assumption .\,}
For every \tx{SCC} $C$ of the generative-dependency graph
and every outer binding $\eta$, the proposal step yields a finite candidate set
$\tx{Cand}_{\,C}(\eta)$. For the chosen semantics
$\sem\in\{\lfp,\tx{stratified},\tx{well\text{-}founded}\}$, the accepted
projections $R_C(\eta)$ are exactly the declarative answer projections for
$C$ under $\eta$. All generated terms, parsed tuples, candidates, and
validation inputs have polynomial-size representations.

\subsection{Algorithm \tx{EvalSCC}}
Algorithm~\ref{alg:evalscc} constructs finite domains $D[\cdot]$ by proposal, forms $\tx{Cand}_{\,C}(\eta)$, and filters it by self-support validation. All semantic checks occur only during validation, on fully instantiated mappings.

\begin{algorithm}[t]
\caption{\tx{EvalSCC}$(\ma{G},\ma{P},C,\eta,K,B,T,R)$}
\label{alg:evalscc}
\begin{algorithmic}[1]
\REQUIRE RDF graph $\ma{G}$; pattern $\ma{P}$; \tx{SCC} $C$; outer binding $\eta$
with $X_C\subseteq\domain(\eta)$; bounds $K$ (proposals), $B$ (domain cap),
$T$ (repairs), $R$ (validation repetitions).
\ENSURE $R_C(\eta)=\{\mu_{|Y_C}\mid \mu\in \tx{Cand}_{\,C}(\eta)\ \wedge\ \tx{SS}_C(\mu)\}$.

\STATE \textbf{Candidate enumeration (heuristic).}
\STATE Initialize $D[v]\gets\varnothing$ for all $v\in Y_C$.
\WHILE{not saturated and $\exists v\in Y_C:\ \size{D[v]}<B$}
  \FORALL{$g=\genbind{\pi(X_g)}{Y_g}{\ma{M}}\in C$}
    \STATE Construct a possibly partial context $\theta$ extending $\eta$; fill any missing
           placeholders in $X_g$ with fixed sentinel tokens.
    \STATE $A\gets \tx{Propose}(g,\theta,K)$.
    \FORALL{$(v,a)\in A$ with $v\in Y_g\cap Y_C$}
      \IF{$\size{D[v]}<B$}
        \STATE $D[v]\gets D[v]\cup\{a\}$.
      \ENDIF
    \ENDFOR
  \ENDFOR
\ENDWHILE

\STATE \textbf{Candidate construction.}
\STATE Form $\tx{Cand}_{\,C}(\eta)$ from $D[\cdot]$ as complete mappings extending $\eta$ (optional symbolic pruning).

\STATE \textbf{Backtracking search with validation.}
\STATE $R_C(\eta)\gets\varnothing$.
\FORALL{complete candidates $\mu\in \tx{Cand}_{\,C}(\eta)$}
  \STATE \emph{Self-support validation:}
  \IF{for every $g=\genbind{\pi(X_g)}{Y_g}{\ma{M}}\in C$, with $\nu_g:=\mu_{|Y_g}$, we have
      $\domain(\nu_g)=Y_g$,
      $\mu_{|X_g\cup Y_g}=\mu_{|X_g}\uplus\nu_g$, and
      $\tx{Validate}(g,\mu_{|X_g},\nu_g;R)=\top$}
    \STATE $R_C(\eta)\gets R_C(\eta)\cup\{\mu_{|Y_C}\}$.
  \ELSE
    \STATE Optionally attempt up to $T$ repairs by resampling values for a failing
           variable via $\tx{Propose}$ and re-validating.
  \ENDIF
\ENDFOR
\RETURN $R_C(\eta)$
\end{algorithmic}
\end{algorithm}


\paragraph{Guarantees.}
For any finite candidate set $\tx{Cand}_{\,C}(\eta)$,
Algorithm~\ref{alg:evalscc} returns exactly the candidate projections that pass the local self-support validation. Thus, the algorithm is exact on the chosen candidate set, sound relative to the self-support target, and complete when the candidate set covers all self-supporting extensions of $\eta$. Under the finite candidate coverage and exactness assumption above, this target coincides with the declarative answer projections for $C$ under $\eta$. For fixed bounds $(K,B,T,R)$, the algorithm terminates; see Appendix~\ref{app:exec-proofs}.

\subsection{Evaluation via \tx{SCC}-condensation}
Each \tx{SCC} $C$ can be summarized as a derived macro-operator $\tx{GenSCC}_C : X_C \longrightarrow Y_C$ with denotation
\begin{multline*}
\evalThree{\tx{GenSCC}_C}{\ma G}{S}=
\{\eta\uplus\rho
\mid
\eta\in S,\ X_C\subseteq\domain(\eta),\\
\rho\in R_C(\eta),\ \eta\sim\rho
\}. 
\end{multline*}

where $R_C(\eta)$ is computed by Algorithm~\ref{alg:evalscc}. Replacing each \tx{SCC}
of $\ma P$ by its macro-operator yields an acyclic generative dependency graph,
and the resulting condensed pattern can be evaluated in topological order, where
each SCC is executed by \tx{EvalSCC} for each outer binding produced upstream.
This condensation is semantics-preserving for the monotone fragment. 

\section{Computational Complexity}
\label{sec:complexity}

We analyze the complexity of query evaluation under the semantics of Sections~\ref{sect:syntax_semantics} and~\ref{sec:fixpoint}. We abstract from the internal cost of model inference and treat each \tx{GenOp} as an external oracle, considering only symbolic query evaluation.


We assume the deterministic bounded-generation, finite-candidate-coverage, and exactness conditions, and abstract from oracle-call and parsing costs, treating them as $O(1)$ per invocation. The computational task is evaluation membership: given a graph pattern $\ma P$, an RDF graph $\ma G$, a typed solution mapping $\mu$, and the applicable semantics (acyclic, stratified, or well-founded), decide whether $\mu$ belongs to the semantic value of $\ma P$ over $\ma G$, \ie whether $\mu\in\evalTwo{\ma P}{\ma G}$.

\begin{theorem}
\label{thm:complexity}
 Under deterministic bounded generation and the finite candidate coverage
and exactness conditions, the evaluation-membership problem satisfies the following bounds:
\begin{enumerate}
  \item Acyclic patterns are in \tx{PTIME} in data complexity and \tx{PSPACE}-complete in combined complexity.
\item Stratified patterns and non-stratified patterns under well-founded semantics are also in \tx{PTIME} in data complexity and \tx{PSPACE}-complete in combined complexity.
\end{enumerate}
\end{theorem}

\noindent In generative SPARQL, recursion neither derives new RDF facts nor expands the data domain. Both acyclic and recursive patterns are evaluated over the same finite set of solution mappings, determined by the query and bounded generation. Recursion thus serves to validate self-support among already enumerable candidates,  rather than to derive RDF facts or chase witnesses as in existential-rule systems. As a result, stratified or well-founded recursion affects only the admissibility of mappings, not the size of the search space, producing the same asymptotic complexity as acyclic evaluation. Dropping bounded-generation assumptions eliminates termination guarantees for fixpoint evaluation; in this case, generative patterns can simulate existential-rule application and chase-style reasoning, leading to undecidability \cite[Chapter~13]{abiteboul1995foundations,cali2010query}.




\section{Conclusions and Future Work}
We presented a conservative, declarative extension of SPARQL with a generative pattern, \tx{GenOp}. The proposed query semantics accommodates acyclic, stratified, and non-stratified generative dependencies via consequence operators and fixpoint constructions. Under a bounded-generation assumption, the extension preserves the established complexity bounds of SPARQL \cite{perez2009semantics}.




The only form of recursion supported arises from dataflow dependencies among generative patterns, where the outputs of one \tx{GenOp} parameterize the prompt of another. This recursion operates at the level of prompt instantiation rather than fact derivation, and no generated value becomes part of the RDF graph under query. As a result, the language extends SPARQL with controlled generative dependencies while preserving its declarative and computational properties \cite{arenas2009foundations} and remaining fundamentally different from Datalog-style rule languages \cite{cali2010query,gottlob2014datalog+,cali2010datalog+} in terms of both expressiveness and complexity.

Future work includes relaxing bounded-generation assumptions, as well as incorporating Datalog-style existential predicate into \tx{GenOp} patterns, and characterizing decidable fragments under guardedness or acyclicity conditions, following the  methods for existential rules \cite{deutsch2008chase,gottlob2015recent,calautti2015chase}, and developing optimizer support for generative operators within the SPARQL algebra.

\paragraph{\textbf{Acknowledgments.}}
This work is funded by the German Research Foundation (DFG) -- SFB 1574 Circular Factory-- 471687386.


\bibliographystyle{kr}
\bibliography{kr-sample}

\newpage

\appendix

\vspace{1cm}
\begin{appendix}

\section*{Extended Proofs and Auxiliary Results}
\section{Sample run of Example~\ref{example_one}}\label{apen:samplerun}
Below, we illustrate the executable treatment of mutually dependent
$\tx{GenOp}$ occurrences. The run constructs a candidate mapping and checks whether it is locally self-supporting. Correctness with respect to the declarative component answers is relative to the finite candidate coverage and exactness assumption from Sect.~\ref{sec:exec-fixpoint}.

\begin{example}\label{ex:sample-run}
Let $\ma G=\{\tx{:Paris a :City}\}$ and consider the query
\[\ma Q \;=\; \join{(\join{\bgps}{g_1})}{g_2},\]
where
\begin{itemize}
    \item $\bgps=\{?x\ \tx{a}\ \tx{:City}\}$,
    \item $g_1=\genbind{\pi_1(\{?x,?z\})}{\{?y\}}{\tx{GPT-4o}}$ with
    $\pi_1(\{?x,?z\})=\text{``Suggest a topic related to $?x$ and $?z$''}$, and
    \item $g_2=\genbind{\pi_2(\{?y\})}{\{?z\}}{\tx{Gemini 1.5 Pro}}$ with
    $\pi_2(\{?y\})=\text{``Suggest a city related to topic $?y$''}$.
\end{itemize}

Evaluating the basic graph pattern gives
\begin{multline*}
S^{(0)}=\evalTwo{\{?x\ \tx{a}\ \tx{:City}\}}{\ma G}
=\{\mu_0\},\\
\mu_0=\{?x\mapsto(\tx{:Paris},\sourceRdf)\}.
\end{multline*}

Neither $g_1$ nor $g_2$ is applicable to $\mu_0$ in a single pass: $g_1$ requires a binding for $?z$, while $g_2$ requires a binding for $?y$. Thus, the dependency graph of $\ma Q$ contains a cycle between $g_1$ and $g_2$.
Executable evaluation handles this cycle by considering finite candidate extensions of $\mu_0$ and validating them against the two local generative conditions.

We now exhibit one candidate mapping and the validation witnesses for it.

\smallskip\noindent
\emph{Hypothesized binding and instantiation of $g_2$.}

Assume that a candidate mapping $\mu^\ast$ binds $?y$ to the generated term
$(t_A,\sourceGen)$ with
\[\lexical(t_A)=\text{``art''}.\]
Instantiating $\pi_2$ with the input projection of $\mu^\ast$ gives
\[\pi_2^{[\mu^\ast_{|\{?y\}}]}=
\text{``Suggest a city related to topic art''}.\]
Assume
\[\Call{\tx{Gemini 1.5 Pro}}(\pi_2^{[\mu^\ast_{|\{?y\}}]})
 = \{\text{``Florence''}\},\]
and that
\begin{multline*}
    \Parse{\tx{Gemini 1.5 Pro}}(\text{``Florence''})=\{(t_F)\},\\ \lexical(t_F)=\text{``Florence''},
\end{multline*}

This validates the generated binding
\[\{?z\mapsto(t_F,\sourceGen)\}\]
for $g_2$ under the input binding $?y\mapsto(t_A,\sourceGen)$.

\smallskip\noindent
\emph{Instantiation of $g_1$.}

Given the bindings for $?x$ and $?z$, instantiating $\pi_1$ gives
\[
\pi_1^{[\mu^\ast_{|\{?x,?z\}}]}
=
\text{``Suggest a topic related to Paris and Florence''}.
\]
Assume
\[
\Call{\tx{GPT-4o}}(\pi_1^{[\mu^\ast_{|\{?x,?z\}}]})
=
\{\text{``art''}\},
\]
and
\[
\Parse{\tx{GPT-4o}}(\text{``art''})
=
\{(t_A)\},
\qquad
\lexical(t_A)=\text{``art''}.
\]
This validates the generated binding
\[
\{?y\mapsto(t_A,\sourceGen)\}
\]
for $g_1$ under the input bindings $?x\mapsto(\tx{:Paris},\sourceRdf)$ and
$?z\mapsto(t_F,\sourceGen)$.

\smallskip\noindent
Let,
\begin{multline*}
\mu^\ast=
\{?x\mapsto(\tx{:Paris},\sourceRdf),\
?y\mapsto(t_A,\sourceGen),\\
?z\mapsto(t_F,\sourceGen)\}.
\end{multline*}

\smallskip\noindent

\emph{Self-support validation.}
The mapping $\mu^\ast$ extends $\mu_0$ and binds the input and output variables of both generative occurrences. For $g_1$, the projection
$\mu^\ast_{|\{?x,?z\}}$ instantiates the prompt, and the projection
$\mu^\ast_{|\{?y\}}$ is validated as an admissible output. For $g_2$, the projection $\mu^\ast_{|\{?y\}}$ instantiates the prompt, and the projection $\mu^\ast_{|\{?z\}}$ is validated as an admissible output. Hence $\mu^\ast$ passes the local self-support test for the component $\{g_1,g_2\}$.

If $\mu^\ast$ belongs to the finite candidate set considered by the executable evaluator, then the evaluator accepts the projection
$\mu^\ast_{|\{?y,?z\}}$. Under the finite candidate coverage and exactness assumption from Sect.~\ref{sec:exec-fixpoint}, this accepted projection coincides with a declarative component-answer projection. The resulting full query answer is
\[\{?x\mapsto(\tx{:Paris},\sourceRdf),\
?y\mapsto(t_A,\sourceGen),\
?z\mapsto(t_F,\sourceGen)\}.\]
The generated bindings $?y$ and $?z$ are typed with $\sourceGen$ and they  do not add RDF triples to $\ma G$.
\end{example}




\subsection{Executable Fixpoint Evaluation of Example~\ref{example_one}}
\label{app:example1-exec}

We now demonstrate executable fixpoint evaluation by bounded candidate
enumeration followed by local self-support validation.

We illustrate the executable procedure from Sect.~\ref{sec:exec-fixpoint} on Example~\ref{example_one}, making explicit (i) how $\tx{Propose}$ induces a finite candidate set and (ii) how validation checks candidate mappings on fully instantiated prompts. We separate \emph{candidate discovery} from \emph{candidate validation}: $\tx{Propose}$ determines which candidates are explored, while validation determines which explored candidates are accepted.
Correctness with respect to the declarative component answers is relative to the finite candidate coverage and exactness assumption from
Sect.~\ref{sec:exec-fixpoint}.

Recall Example~\ref{example_one},
\begin{verbatim}
SELECT ?x ?y ?z WHERE {
?x a :City .
GenOp("Suggest a topic related to ?x
            and ?z" AS ?y, GPT-4o).
GenOp("Suggest a city related to 
    topic ?y" AS ?z, Gemini 1.5 Pro).
}
\end{verbatim}
Let $\ma G=\{\tx{:Paris a :City}\}$. Let
\begin{multline*}
g_1=\genbind{\pi_1(\{?x,?z\})}{\{?y\}}{\ma M_1},\\
g_2=\genbind{\pi_2(\{?y\})}{\{?z\}}{\ma M_2},
\end{multline*}
where
\begin{itemize}
    \item $\pi_1$ is the template \tx{"Suggest a topic related to ?x and ?z"}; and
    \item $\pi_2$ is \tx{"Suggest a city related to topic ?y"}.
\end{itemize}

We assume canonical lexicalization for RDF terms, \eg
$\canonical(\tx{:Paris})=\text{``Paris''}$, and that generated terms
$t\in\ma T_{\tx{gen}}$ carry lexicalizations $\lexical(t)\in\Sigma^*$.

Assume the following operational bounds and decoding-format conventions.
\begin{enumerate}
    \item Proposal bound $K$ (\ie maximum proposals per $\tx{Propose}$ call) and
    domain cap $B$ (\ie maximum per-variable domain size).

    Here, $K$ bounds the number of distinct values returned by a single proposal
    call, while $B$ bounds the total number of values accumulated for each generated
    variable across proposal calls. Together, $K$ and $B$ make the explored
    candidate domains finite. Values not proposed because of these bounds are not
    explored operationally; they are part of the declarative semantics only when
    covered by the finite candidate coverage assumption.

    \item Validation repetitions $R$ and optional repair attempts $T$.

    Validation uses a binary membership-style query to the model. The parameter $R$ bounds repeated validation calls used by an implementation to stabilize the acceptance decision. If validation fails, $T$ bounds the number of local repair attempts in which an alternative generated value is substituted and revalidated. These parameters affect the executable procedure and runtime; the  declarative semantics is fixed independently of them.

    \item The model response must be parseable as JSON, which can be enforced by grammar-constrained decoding for LLMs~\cite{park2025flexibleefficientgrammarconstraineddecoding}.
    For $g_1$ (output variable $?y$), we use the schema
        \[\tx{\{"y": ["<string>", ...]\}},\]
    and for $g_2$ (output variable $?z$), we use
   \[ \tx{\{"z": ["<string>", ...]\}}.\]

    \item Decoding parameters, such as temperature, top-$p$, and maximum number
    of tokens, are fixed per model and treated as part of the implementation.
\end{enumerate}

Parsing maps each output string $s$ to a fresh generated term
$t_s\in\ma T_{\tx{gen}}$ with $\lexical(t_s)=s$, or reuses an existing $t_s$ if already created.

\subsubsection{\tx{Propose} and Self-validation Procedures}

We instantiate the abstract semantics with two concrete operational procedures.

\noindent
\emph{(A) \tx{Propose} as bounded enumeration.}
Given $g=\genbind{\pi(X)}{Y}{\ma M}$, a context mapping $S$, and a bound $K$:
\begin{enumerate}
\item \emph{Prompt instantiation.} Build a concrete prompt string
$\pi^{[S]}$ by replacing placeholders in $X\cap\domain(S)$ using
lexicalization. If some placeholder in $X$ is unbound in $S$, render it using
a fixed sentinel token $\tx{<unknown:?v>}$.
\item \emph{Model call.} Query $\ma M$ under fixed decoding parameters with an
instruction requiring JSON output and at most $K$ items.
\item \emph{Parse and deduplicate.} Parse the JSON, extract up to $K$ distinct
strings, and map them to generated terms deterministically.
\item \emph{Return.} Return the resulting set of at most $K$ generated-term
tuples. In this example, all tuples are singletons because $|Y|=1$.
\end{enumerate}
This procedure never claims completeness; it merely returns a bounded set of
candidates. Finiteness is enforced by the explicit bound $K$ and by the domain
cap $B$ applied when inserting into $D[\cdot]$.

\noindent
\emph{(B) Self-support validation as membership checking.}
Given a complete candidate mapping $\mu$ and
$g=\genbind{\pi(X)}{Y}{\ma M}$, validation checks whether the tuple $\mu_{|Y}$ is an admissible output for $g$ under the fully instantiated prompt $\pi^{[\mu_{|X}]}$. We use a verification query, that is, a Boolean \tx{YES}/\tx{NO} query, rather than relying on exact re-generation:
\begin{enumerate}
\item Instantiate the base instruction prompt $\pi^{[\mu_{|X}]}$.
\item Form a verification prompt that includes the instruction and the candidate output tuple $\mu_{|Y}$, and asks for a strict \tx{YES}/\tx{NO} answer.
\item Query $\ma M$ under the fixed validation policy. Repeat $R$ times if specified.
\item Accept if the aggregation policy for the $R$ validation calls accepts the candidate.
\end{enumerate}
This gives an explicit implementation of the abstract membership test
\[\mu_{|X}\uplus\mu_{|Y}\in\evalThree{g}{\ma G}{\{\mu_{|X}\}}.\]

\subsubsection{Memoization}

Maintain two caches:
(i) a proposal cache keyed by $(g,\pi^{[S]},K,\text{decoding})$; and
(ii) a validation cache keyed by $(g,\pi^{[\mu_{|X}]},\mu_{|Y},R)$.
Caches stabilize evaluation and reduce repeated model calls.

\begin{enumerate}
\item \emph{Step~1: Base RDF evaluation.}
Evaluating the basic graph pattern
\[?x\ \tx{a}\ :\tx{City}\]
over $\ma G$ yields the unique typed solution mapping
\[\mu_0=\{?x\mapsto(\tx{:Paris},\sourceRdf)\}.\]
This is the outer binding $\eta:=\mu_0$.

\item \emph{Step~2: \tx{SCC} identification and summarization.}
Since $?y$ feeds $g_2$ and $?z$ feeds $g_1$, the dependency graph contains one
\tx{SCC}:
\[\ma C=\{g_1,g_2\}.\]
The externally supplied variables are
\[X_C=\{?x\},\]
and the variables generated inside the component are
\[Y_C=\{?y,?z\}.\]
We summarize $\ma C$ by the macro-operator $\tx{GenSCC}_{\ma C}$. For an outer binding $\eta$, this operator returns full mappings of the form $\eta\uplus\rho$, where $\rho$ ranges over the accepted $Y_C$-projections produced by the executable evaluation of $\ma C$ under $\eta$. Under the finite candidate coverage and exactness assumption from Sect.~\ref{sec:exec-fixpoint}, these accepted projections coincide with the declarative component-answer projections.

\item \emph{Step~3: Candidate enumeration, operational only.}
We invoke
\[\textsc{EvalSCC}(\ma G,\ma P,\ma C,\eta,K,B,T,R).\]
Initialize finite domains:
\[D[?y]=\varnothing\quad\text{and}\quad D[?z]=\varnothing.\]

\smallskip
\noindent
\emph{Round 1: propose values for $?y$.}
We call $\tx{Propose}(g_1,S_1,K)$ with
\[S_1=\{?x\mapsto(\tx{:Paris},\sourceRdf)\}.\]
Since $?z\notin\domain(S_1)$, the sentinel policy instantiates
\begin{multline*}
\pi_1^{[S_1]} =
\text{``Suggest a topic related to Paris and}\\
\text{\tx{<unknown:?z>}''}.
\end{multline*}
We query $\ma M_1$ with an instruction such as:
\begin{quote}\footnotesize
Return JSON \tx{\{"y":[...]\}} with at most $K$ short topic strings.
\end{quote}
Suppose the model returns:
\[\tx{\{"y":["Art","Cuisine"]\}}.\]
Parsing yields generated terms $t_{\textsf{Art}}$ and $t_{\textsf{Cuisine}}$
with
\[\lexical(t_{\textsf{Art}})=\text{``Art''},
\quad\lexical(t_{\textsf{Cuisine}})=\text{``Cuisine''}.\]
Update, respecting the cap $B$:
\[D[?y]=\{t_{\textsf{Art}},t_{\textsf{Cuisine}}\}.\]

\smallskip
\noindent
\emph{Round 2: propose values for $?z$ conditional on $?y$.}
Choose a currently discovered value, for example $?y=t_{\textsf{Art}}$, and call $\tx{Propose}(g_2,S_2,K)$ with
\[S_2=\{
?x\mapsto(\tx{:Paris},\sourceRdf),\
?y\mapsto(t_{\textsf{Art}},\sourceGen)
\}.\]
Then,
\[\pi_2^{[S_2]}=\text{``Suggest a city related to topic Art''}.\]
Query $\ma M_2$ requiring JSON \tx{\{"z":[...]\}} with at most $K$ city strings.
Suppose it returns:
\[\tx{\{"z":["Florence","Rome"]\}}.\]
Parsing yields $t_{\textsf{Florence}}$ and $t_{\textsf{Rome}}$, and we update:
\[D[?z]=\{t_{\textsf{Florence}},t_{\textsf{Rome}}\}.\]

\smallskip
\noindent
\emph{Further rounds.}
One can repeat conditional proposals, for example for
$?y=t_{\textsf{Cuisine}}$, and alternate between $g_1$ and $g_2$ until no new
values are added or the cap $B$ is reached. In all cases,
$|D[v]|\le B$ is enforced by construction.

\smallskip
\noindent
\emph{Finite candidate set.}
The induced candidate set is the finite product:
\[\Cand{\ma C} =
\{\eta\uplus\nu \mid \nu(?y)\in D[?y],\ \nu(?z)\in D[?z]\}.\]
Suppose, in this run, $|D[?y]|=2$ and $|D[?z]|=2$. Then,
$|\Cand{\ma C}|=4$.

\item \emph{Step~4: Candidate validation.}
We enumerate candidates $\mu\in\Cand{\ma C}$ and accept those that pass local self-support validation, that is, those for which each $g\in\ma C$ validates the output assigned by $\mu$ under the prompt instantiated from the input variables of $g$.

\smallskip
\noindent
\emph{Candidate sample.} Consider
\begin{multline*}
\mu=\{?x\mapsto(\tx{:Paris},\sourceRdf),\
?y\mapsto(t_{\textsf{Art}},\sourceGen),\\
?z\mapsto(t_{\textsf{Florence}},\sourceGen)\}.
\end{multline*}

\smallskip
\noindent
\emph{Validate $g_1$ on $\mu$ (topic validation).}
We have $X=\{?x,?z\}$ and $Y=\{?y\}$. Instantiate the base instruction prompt
using $\mu_{|X}$:
\[\pi_1^{[\mu_{|X}]}= \text{``Suggest a topic related to Paris and Florence''}.\]
We then issue a verification query to $\ma M_1$ such as:
\begin{quote}\footnotesize
Instruction: ``Suggest a topic related to Paris and Florence.''\\
Candidate topic: ``Art''.\\
Is the candidate a valid output for the instruction? Answer \tx{YES} or
\tx{NO} only.
\end{quote}
Repeat according to the validation bound $R$ and accept if the validation aggregation policy accepts. If accepted, this operationally establishes that $t_{\textsf{Art}}$ is admissible for $g_1$ under the input projection $\mu_{|X}$.

\smallskip
\noindent
\emph{Validate $g_2$ on $\mu$ (city validation).}
We have $X=\{?y\}$ and $Y=\{?z\}$. Instantiate:
\[\pi_2^{[\mu_{|X}]}=\text{``Suggest a city related to topic Art''}.\]
Verification query to $\ma M_2$:
\begin{quote}\footnotesize
Instruction: ``Suggest a city related to topic Art.''\\
Candidate city: ``Florence''.\\
Is the candidate a valid output for the instruction? Answer \tx{YES} or
\tx{NO} only.
\end{quote}
Again, repeat according to the validation bound $R$ and accept if the validation
aggregation policy accepts.

\smallskip
\noindent
\textbf{Acceptance.}
The candidate $\mu$ is accepted iff both validations pass. In that case, $\mu$ satisfies the local self-support condition for the component $\ma C$. If either validation fails, $\mu$ is rejected.

\smallskip
\noindent
\emph{Optional repair.}
If $\mu$ fails due to, say, $g_2$, we may attempt up to $T$ repairs:
keep $\mu_{|X}$ fixed for $g_2$, here $?y=\text{Art}$, call
$\tx{Propose}(g_2,\{?y=\text{Art}\},K)$ or consult cached proposals,
substitute an alternative $?z$ from the proposals, for example \tx{Rome}, and
rerun the same validation checks. Only mappings that pass the same local
self-support validation are accepted.

\item \emph{Step~5: Final result.}
For each accepted candidate $\mu$, the \tx{SCC} yields $\mu_{|Y_C}$, and the
full query answer is obtained by joining with the outer mapping $\mu_0$, which
is already included in $\mu$. Thus, answers have the form
\[\{
?x\mapsto(\tx{:Paris},\sourceRdf),\
?y\mapsto(t,\sourceGen),\
?z\mapsto(u,\sourceGen)
\}.\]
\end{enumerate}

\begin{remark}
The mechanisms $\tx{Propose}$, sentinel placeholders, decoding choices, JSON schemas, caching, and repair only control candidate discovery and execution cost. Acceptance of a candidate is determined by local self-support validation on fully instantiated prompts. Correctness with respect to the declarative component answers is obtained under the finite candidate coverage and exactness
assumption from Sect.~\ref{sec:exec-fixpoint}.
\end{remark}

\section{Fixpoint-Based Semantics}
\label{app:proofs}

We assume a \emph{bounded-generation} regime: for every $\tx{GenOp}$ occurrence
$g=\genbind{\pi(X)}{Y}{\ma M}$, every instantiated prompt string produces a finite
set of textual answers via $\Call{M}$, and every textual answer yields a finite
set of parsed tuples via $\Parse{M}$. The total set of typed solution mappings that may be produced during evaluation is finite. Formally, the induced universe $\Omega_{\ma P}$ (as defined in the main text) is finite.

\emph{Set-theoretic lattice facts.}
For any set $U$, the powerset $(2^U,\subseteq)$ is a complete lattice with least
element $\varnothing$, greatest element $U$, and joins given by unions.
If $U$ is finite, every ascending chain in $(2^U,\subseteq)$ stabilizes after at
most $|U|$ strict increases.

\emph{Compatibility.}
We use the compatibility relation $\sim$ and union $\uplus$ of typed solution mappings as defined in the main text. We rely on the following basic fact: if $\mu_1\sim \mu_2$, then $\mu_1\uplus\mu_2$ is well-defined and extends both
$\mu_1$ and $\mu_2$.

\emph{Context restriction.}
For a context supply $S\subseteq \Omega_{\ma P}$ and a $\tx{GenOp}$ pattern $g$,
we recall $\ContextS{S}{g}=\{\mu\in S\mid \holder{g}\subseteq \domain{\mu}\}$.
For any $S\subseteq S'$, we always have $\ContextS{S}{g}\subseteq \ContextS{S'}{g}$.

\emph{Immediate consequence operator.}
We use $\Consequence{\ma P}:2^{\Omega_{\ma P}}\to 2^{\Omega_{\ma P}}$ given by
$\Consequence{\ma P}(S)=\evalThree{\ma P}{\ma G}{S}$ (main text). Whenever we write $\Consequence{\ma P}^n$, this denotes $n$-fold composition, with $\Consequence{\ma P}^0$ being the identity.

\subsection{Acyclic patterns: single-pass evaluation coincides with fixpoint semantics}
\label{app:acyclic}

This subsection proves the theorem stated in the main text:
acyclic dependency implies that the semantics can be computed in a single
pass and matches the least-fixpoint semantics.

\emph{Operational single-pass evaluation for acyclic patterns.}
Let $\directedGraph{\ma D}_{\ma P}=(\genPattern{\ma P},\to_{\ma P})$ be the
dependency graph. Let $\directedGraph{\ma D}$ be acyclic, and let $\gamma=(g_1,\dots,g_m)$ be a topological order of the nodes $\genPattern{\ma P}$.

For $k\in\{0,\dots,m\}$, let $\ma P^{[k]}_{\gamma}$ denote the pattern obtained from $\ma P$
by disabling every \tx{GenOp} occurrence not among $\{g_1,\dots,g_k\}$, \ie replacing it
by a pattern whose evaluation yields $\varnothing$, while leaving the remaining algebraic
structure unchanged. Thus, $\ma P^{[0]}_{\gamma}$ has all \tx{GenOp}s disabled and
$\ma P^{[m]}_{\gamma}=\ma P$.

The staged sets are:
\[S_{\gamma}^{(0)} := \evalThree{\ma P^{[0]}_{\gamma}}{\ma G}{\varnothing},
\qquad S_{\gamma}^{(k)} := \evalThree{\ma P^{[k]}_{\gamma}}{\ma G}{S_{\gamma}^{(k-1)}}\ \ (1\le k\le m).\]
That is, at stage $k$, $\tx{GenOp}$ occurrences up to $g_k$ may fire using the
context supply accumulated at stage $(k-1)$.

\begin{lemma}(Well-defined staged evaluation)
\label{lem:acyclic-welldefined}
For each $k\in\{0,\dots,m\}$, the set $S_{\gamma}^{(k)}$ is well-defined and contained in
$\Omega_{\ma P}$.
\end{lemma}
\begin{proof}
By bounded generation, each \tx{GenOp} call yields finitely many outputs, and every inductive
evaluation clause returns a finite set of typed solution mappings. Since evaluation is defined
for any supply $S\subseteq \Omega_{\ma P}$ (via $\ContextS{S}{g}$ for \tx{GenOp}s),
each $S_{\gamma}^{(k)}$ is defined. Containment in $\Omega_{\ma P}$ follows by construction.
\end{proof}





\begin{lemma}(No backward dependencies)
\label{lem:no-backward}
Let $\gamma=(g_1,\dots,g_m)$ be topological. If $i<j$, then
$\target{g_j}\cap \holder{g_i}=\varnothing$.
\end{lemma}
\begin{proof}
If $\target{g_j}\cap \holder{g_i}\neq\varnothing$, then by definition of $\to_{\ma P}$ we have
$g_j\to_{\ma P} g_i$, contradicting that $\gamma$ is a topological order.
\end{proof}

\begin{theorem}[Acyclic collapse: single-pass equals least fixpoint]
\label{thm:acyclic-collapse-detailed}
Assume $\directedGraph{\ma D}_{\ma P}$ is acyclic and fix a topological order $\gamma$.
Then:
\begin{enumerate}
\item $S_\gamma^{(m)}$ is a fixpoint of $\Consequence{\ma P}$.
\item $S_\gamma^{(m)}$ equals the least fixpoint $\lfp(\Consequence{\ma P})$.
\item In particular, the result is independent of the chosen topological order.
\end{enumerate}
\end{theorem}

\begin{proof}
We proceed in steps.

\medskip\noindent
\emph{(1) $S_{\gamma}^{(m)}$ is a post-fixpoint.}
We show $\Consequence{\ma P}(S_{\gamma}^{(m)})\subseteq S_{\gamma}^{(m)}$.
Let $\mu\in \Consequence{\ma P}(S_{\gamma}^{(m)})=\evalThree{\ma P}{\ma G}{S_{\gamma}^{(m)}}$.
In evaluating $\ma P$ relative to $S_{\gamma}^{(m)}$, each \tx{GenOp} occurrence $g_i$ draws its contexts from $\ContextS{S_{\gamma}^{(m)}}{g_i}$.


We claim that any placeholder-binding used for $g_i$ at stage $m$ is already available at
stage $(i-1)$. That is, for every $\eta\in \ContextS{S_{\gamma}^{(m)}}{g_i}$ there exists
$\eta'\in \ContextS{S_{\gamma}^{(i-1)}}{g_i}$ with
$\eta'|_{\holder{g_i}}=\eta|_{\holder{g_i}}$.
Indeed, by Lemma~\ref{lem:no-backward}, the variables in $\holder{g_i}$ can be bound only by
(i) non-generative parts of $\ma P$ or (ii) outputs of \tx{GenOp}s $g_j$ with $j<i$.
All such outputs are included in $S_{\gamma}^{(i-1)}$ by construction.

Since prompt instantiation depends on a context mapping only via its restriction to
$\holder{g_i}$, the set of instantiated prompts (and hence generated outputs) available for
$g_i$ under $S_{\gamma}^{(m)}$ coincides with that under $S_{\gamma}^{(i-1)}$.
Consequently, no new \tx{GenOp}-extensions arise when evaluating $\ma P$ relative to
$S_{\gamma}^{(m)}$ beyond those already produced during the staged construction up to $m$.
Therefore $\evalThree{\ma P}{\ma G}{S_{\gamma}^{(m)}}\subseteq S_{\gamma}^{(m)}$.

\emph{(2) $S_{\gamma}^{(m)}$ is a fixpoint.}
By construction, $S_{\gamma}^{(m)}=\evalThree{\ma P^{[m]}_{\gamma}}{\ma G}{S_{\gamma}^{(m-1)}}
=\Consequence{\ma P}(S_{\gamma}^{(m-1)})$ and $S_{\gamma}^{(m-1)}\subseteq S_{\gamma}^{(m)}$.
By monotonicity of $\Consequence{\ma P}$,
\[
S_{\gamma}^{(m)}=\Consequence{\ma P}(S_{\gamma}^{(m-1)})
\subseteq \Consequence{\ma P}(S_{\gamma}^{(m)}).
\]
Combined with (1), this yields $\Consequence{\ma P}(S_{\gamma}^{(m)})=S_{\gamma}^{(m)}$.

\emph{Leastness.}
Let $S$ be any fixpoint of $\Consequence{\ma P}$. Consider the Kleene chain
$S_0=\varnothing$ and $S_{n+1}=\Consequence{\ma P}(S_n)$. By the standard induction argument
(using monotonicity), $S_n\subseteq S$ for all $n$, hence
$\lfp(\Consequence{\ma P})=\bigcup_{n\ge 0} S_n \subseteq S$.
Applying this to the fixpoint $S_{\gamma}^{(m)}$ yields
$\lfp(\Consequence{\ma P})\subseteq S_{\gamma}^{(m)}$.

\emph{ $S_{\gamma}^{(m)}\subseteq \lfp(\Consequence{\ma P})$.}
Each stage $S_{\gamma}^{(k)}$ is obtained by one application of $\Consequence{\ma P^{[k]}_{\gamma}}$
to $S_{\gamma}^{(k-1)}$, and $\ma P^{[k]}_{\gamma}$ is a sub-instance of $\ma P$ (with some
\tx{GenOp}s disabled). Hence every mapping in $S_{\gamma}^{(m)}$ is derivable by a finite number
of iterations of $\Consequence{\ma P}$ starting from $\varnothing$, so
$S_{\gamma}^{(m)}\subseteq \bigcup_{n\ge 0}\Consequence{\ma P}^n(\varnothing)=\lfp(\Consequence{\ma P})$.

Combining (3) and (4), we arrive to $S_{\gamma}^{(m)}=\lfp(\Consequence{\ma P})$.

\emph{(5) Independence of $\gamma$.}
If $\gamma$ and $\gamma'$ are two topological orders, standard poset arguments show that
$\gamma$ can be transformed into $\gamma'$ by swapping adjacent incomparable elements.
For adjacent incomparable generators $g_i,g_{i+1}$, neither depends on the other, \ie
$\target{g_i}\cap \holder{g_{i+1}}=\varnothing$ and $\target{g_{i+1}}\cap \holder{g_i}=\varnothing$. Thus, the set of placeholder-bindings available to each is unchanged by the swap, and so is the set of generated outputs. Repeating swaps yields equality of the final staged sets.
\end{proof}

\subsection{Context equivalence invariance}
\label{app:context}


\begin{proposition}[Context invariance (corrected)]
\label{prop:context-equiv-detailed}
Let $S,S'\subseteq\Omega_{\ma P}$ be context supplies such that for every
$\tx{GenOp}$ \emph{occurrence} $g$ in $\ma P$,
\begin{multline*}
    \ContextS{S}{g} \;=\; \ContextS{S'}{g},
\qquad\text{where}\\
\ContextS{S}{g} = \{\mu\in S \mid \holder{g}\subseteq \domain(\mu)\}.
\end{multline*}
Then, $\Consequence{\ma P}(S)=\Consequence{\ma P}(S')$.
\end{proposition}

\begin{proof}
We prove by structural induction on $\ma P$ that
$\evalThree{\ma P}{\ma G}{S}=\evalThree{\ma P}{\ma G}{S'}$.

\medskip\noindent
\emph{Base case 1: $\ma P$ is a BGP.}
Then $\evalThree{\ma P}{\ma G}{S}$ is independent of $S$, hence equality holds.

\medskip\noindent
\emph{Base case 2: $\ma P$ is a $\tx{GenOp}$ pattern $g=\genbind{\pi(X)}{Y}{\ma M}$.}
By definition,
\[\evalThree{g}{\ma G}{S}=\evalThree{g}{\ma G}{\ContextS{S}{g}},\quad\ContextS{S}{g}=\{\mu\in S: X\subseteq\domain{\mu}\}.\]
By assumption, for this occurrence $g$ we have $\ContextS{S}{g}=\ContextS{S'}{g}$.
Hence, the same set of context mappings $\mu_c$ is used in both evaluations. Since $\evalThree{g}{\ma G}{\cdot}$ is defined as a union over all $\mu_c$ in the context supply, it follows directly that $\evalThree{g}{\ma G}{S}=\evalThree{g}{\ma G}{S'}$.
\medskip\noindent
\emph{Inductive cases.}
Assume the claim holds for immediate subpatterns.
\begin{enumerate}
\item If $\ma P=\proj{\ma P'}{L}$, then
$\evalThree{\ma P}{\ma G}{S}$ is obtained by restricting each mapping in
$\evalThree{\ma P'}{\ma G}{S}$ to $L$, which preserves equality.
\item If $\ma P=\union{\ma P_1}{\ma P_2}$, then evaluation is set union of the
evaluations of $\ma P_1$ and $\ma P_2$, hence equality follows.
\item If $\ma P=\join{\ma P_1}{\ma P_2}$, then evaluation is the set of unions
$\mu_1\uplus\mu_2$ for compatible $\mu_i$ from each side. By IH,
the sets of $\mu_i$ coincide under $S$ and $S'$, hence the join result coincides.
\item If $\ma P=\filter{\ma P'}{F}$, evaluation is subset selection of
$\evalThree{\ma P'}{\ma G}{S}$ by a predicate on $\mu$ alone ($F^\mu$ depends only
on $\mu$). Since the underlying set is equal by IH, the filtered subset is equal.
\item The cases $\opt{}{}{}$, $\diff{}{}{}$, and $\minus{}{}$ are analogous: each is
defined purely in terms of membership/compatibility relations among mappings in
the evaluated subpatterns, so equality of those evaluated sets implies equality
of the resulting sets.
\end{enumerate}
Thus $\evalThree{\ma P}{\ma G}{S}=\evalThree{\ma P}{\ma G}{S'}$ for all $\ma P$,
and therefore $\Consequence{\ma P}(S)=\Consequence{\ma P}(S')$.
\end{proof}

\subsection{Monotonicity of the consequence operator}
\label{app:mono}

We now give a fully detailed proof of Lemma~\ref{lem:monotone}.

\begin{lemma}[Monotonicity of $\Consequence{\ma P}$]
\label{lem:monotone-detailed}
If $\ma P$ is monotone w.r.t.\ the context supply (main text condition), then
$\Consequence{\ma P}$ is monotone on $(2^{\Omega_{\ma P}},\subseteq)$:
\[S\subseteq S' \Longrightarrow \Consequence{\ma P}(S)\subseteq \Consequence{\ma P}(S').\]
\end{lemma}

\begin{proof}
We prove the stronger statement:
for all patterns $\ma P$ satisfying the monotonicity condition,
\[S\subseteq S' \Longrightarrow \evalThree{\ma P}{\ma G}{S}\subseteq \evalThree{\ma P}{\ma G}{S'}.\]
The proof is by structural induction on $\ma P$.

\medskip\noindent
\emph{Case 1:} Let $\ma P$ be a BGP. Evaluation then does not depend on context supply. Hence, equality holds.

\medskip\noindent
\emph{Case 2:} Let $\ma P= g=\genbind{\pi(X)}{Y}{\ma M}$. We have $\ContextS{S}{g}\subseteq\ContextS{S'}{g}$ because $S\subseteq S'$.
By the definition of $\evalThree{g}{\ma G}{\Omega_c}$ as a union over contexts
$\mu_c\in\Omega_c$, enlarging $\Omega_c$ can only add terms to the union:
\[\evalThree{g}{\ma G}{S}=\evalThree{g}{\ma G}{\ContextS{S}{g}}
\subseteq\evalThree{g}{\ma G}{\ContextS{S'}{g}}=\evalThree{g}{\ma G}{S'}.\]
\medskip\noindent

\emph{Case 3:} Let $\ma P=\proj{\ma P'}{L}$.
By IH, $\evalThree{\ma P'}{\ma G}{S}\subseteq \evalThree{\ma P'}{\ma G}{S'}$.
Restriction to $L$ is monotone on sets: if $A\subseteq B$ then
$\{\mu|_L:\mu\in A\}\subseteq \{\mu|_L:\mu\in B\}$, hence inclusion holds.
\medskip\noindent

\emph{Case 4:} Let $\ma P=\union{\ma P_1}{\ma P_2}$.
By IH, $\evalThree{\ma P_i}{\ma G}{S}\subseteq \evalThree{\ma P_i}{\ma G}{S'}$.
Taking unions preserves inclusion, so the claim follows.
\medskip\noindent

\emph{Case 5:} Let $\ma P=\join{\ma P_1}{\ma P_2}$.
Let $\mu\in \evalThree{\join{\ma P_1}{\ma P_2}}{\ma G}{S}$.
Then there exist $\mu_1\in\evalThree{\ma P_1}{\ma G}{S}$ and
$\mu_2\in\evalThree{\ma P_2}{\ma G}{S}$ with $\mu_1\sim\mu_2$ and
$\mu=\mu_1\uplus\mu_2$.
By IH, $\mu_i\in\evalThree{\ma P_i}{\ma G}{S'}$ as well, so the same pair witnesses
$\mu\in \evalThree{\join{\ma P_1}{\ma P_2}}{\ma G}{S'}$.
\medskip\noindent

\emph{Case 6:} Let $\ma P=\filter{\ma P'}{F}$.
By IH, $\evalThree{\ma P'}{\ma G}{S}\subseteq \evalThree{\ma P'}{\ma G}{S'}$.
Now
\[\evalThree{\filter{\ma P'}{F}}{\ma G}{S}=\{\mu\in \evalThree{\ma P'}{\ma G}{S} : F^\mu=\top\}.\]
Since the predicate $F^\mu=\top$ depends only on $\mu$, subset inclusion of the
underlying sets yields inclusion of the filtered sets.

\medskip\noindent

\emph{Case 7:} Assume negative constructs.
These comprise \tx{Minus}, \tx{Diff}, and \tx{Opt}.
In general, such constructs are non-monotone in SPARQL.
However, we restrict attention to patterns $\ma P$ satisfying the main-text
monotonicity condition: no variable introduced by a $\tx{GenOp}$ output occurs
in a negative position.

Under this condition, enlarging the context supply from $S$ to $S'$ may introduce
additional generated bindings, but these bindings cannot influence the evaluation
of any negative construct, since all negative tests range only over variables whose
bindings are fixed independently of $\tx{GenOp}$ outputs.
Hence, the truth of membership and compatibility checks performed by negative
constructs is invariant under such enlargements.

It follows that the denotation of each negative construct is unchanged when passing
from $S$ to $S'$, and the desired inclusion holds by the induction hypothesis applied
to the positive subpatterns. This completes the induction.


\end{proof}

\subsection{Least fixpoint: existence, characterization, and minimality}
\label{app:lfp}

\begin{theorem}[Existence and characterization of the least fixpoint]
\label{thm:lfp-detailed}
If $\Consequence{\ma P}$ is monotone on $(2^{\Omega_{\ma P}},\subseteq)$, then:
\begin{enumerate}
\item $\Consequence{\ma P}$ has a least fixpoint $\lfp(\Consequence{\ma P})$;
\item $\lfp(\Consequence{\ma P})=\bigcup_{n\ge 0}\Consequence{\ma P}^n(\varnothing)$;
\item if $\Omega_{\ma P}$ is finite, the Kleene chain stabilizes after at most
$|\Omega_{\ma P}|$ strict increases.
\end{enumerate}
\end{theorem}
\noindent (Equivalently, under the alternative assumption that $\Phi$ is $\omega$-continuous,
Item~2 holds without requiring finiteness.)




\begin{proof}
\emph{(1) Existence.}
$(2^{\Omega_{\ma P}},\subseteq)$ is a complete lattice. Since $\Phi$ is monotone,
Knaster--Tarski yields existence of a least fixpoint, namely
\[
\lfp(\Phi)=\bigcap\{S\subseteq \Omega_{\ma P}\mid \Phi(S)\subseteq S\}.
\]

\emph{(2) Kleene characterization under finiteness.}
Consider the Kleene sequence by $S_0=\varnothing$ and $S_{n+1}=\Phi(S_n)$.
By monotonicity, $S_0\subseteq S_1\subseteq S_2\subseteq \cdots$.

Assume $\Omega_{\ma P}$ is finite. Then $(S_n)_n$ can strictly increase at most
$|\Omega_{\ma P}|$ times, hence there exists $N\le |\Omega_{\ma P}|$ such that
$S_N=S_{N+1}=\Phi(S_N)$. Thus $S_N$ is a fixpoint.

To show minimality, let $T$ be any fixpoint of $\Phi$. By induction on $n$,
$S_n\subseteq T$ for all $n$:
base $S_0=\varnothing\subseteq T$, and if $S_n\subseteq T$ then
$S_{n+1}=\Phi(S_n)\subseteq \Phi(T)=T$ by monotonicity and $\Phi(T)=T$.
Hence $S_N\subseteq T$ for every fixpoint $T$, so $S_N$ is the least fixpoint:
$S_N=\lfp(\Phi)$.

(3) Finally, since the chain stabilizes at $N$, we have
$\bigcup_{n\ge 0} S_n = S_N$, and therefore
\[
\lfp(\Phi)=S_N=\bigcup_{n\ge 0}\Phi^n(\varnothing),
\]
with stabilization after at most $|\Omega_{\ma P}|$ strict increases.
\end{proof}

\begin{theorem}[Minimality of the least fixpoint]
\label{thm:minimality-detailed}
Let $S^\ast=\lfp(\Consequence{\ma P})$ and let $S\subseteq\Omega_{\ma P}$ satisfy $\Consequence{\ma P}(S)\subseteq S$ (\ie $S$ is a pre-fixpoint, or closed under one-step consequences). Then $S^\ast\subseteq S$.
\end{theorem}

\begin{proof}
Let $S_0=\varnothing$ and $S_{n+1}=\Consequence{\ma P}(S_n)$. We show by induction that $S_n\subseteq S$ for all $n$.
Base: $S_0=\varnothing\subseteq S$. Step: assume $S_n\subseteq S$. By monotonicity, $\Consequence{\ma P}(S_n)\subseteq \Consequence{\ma P}(S)$. By assumption $\Consequence{\ma P}(S)\subseteq S$. Therefore $S_{n+1}=\Consequence{\ma P}(S_n)\subseteq S$. Taking unions over $n$ produces $S^\ast=\bigcup_n S_n\subseteq S$.
\end{proof}

\subsection{Stratified semantics: well-definedness, uniqueness, and order-independence}
\label{app:strat}

\begin{theorem}[Well-definedness and uniqueness of stratified semantics]
\label{thm:strat-detailed}
If $\ma P$ is stratified and $\Omega_{\ma P}$ is finite, then the stratified
semantics $\evalThree{\ma P}{\ma G}{\mathrm{strat}}$ is well-defined and unique.
\end{theorem}

\begin{proof}
Assume stratification $\sigma:\genPattern{\ma P}\to \mathbb N$, and let $S_{-1}=\varnothing$ and for each $i\ge 0$:
\[F_i(S):=\Consequence{\ma P^{\leq i}}(S\cup S_{i-1}),\qquad S_i:=\lfp(F_i),\]
where $\ma P^{\le i}$ denotes the restriction of $\ma P$ to $\tx{GenOp}$ occurrences in strata at most $i$, as defined in the main text.


We then must show each $F_i$ is monotone, so $S_i$ exists and is unique. Let $S\subseteq S'$. Then $S\cup S_{i-1}\subseteq S'\cup S_{i-1}$. We claim $\Consequence{\ma P^{\leq i}}$ is monotone in its context argument when all negative dependencies of generators in strata $\le i$ point strictly below $i$ and are, therefore, evaluated against a fixed lower-stratum supply $S_{i-1}$. Stratification guarantees that any occurrence of an output variable of a stratum $\leq i$ generator under a negative operator can only depend on generators in strictly lower strata, hence, is already fixed when computing stratum $i$.

Thus, as the varying part $S$ grows, no negative test that depends on generated variables can flip from true to false. Therefore Lemma~\ref{lem:monotone-detailed} applies to the operator $S\mapsto \Consequence{\ma P^{\leq i}}(S\cup S_{i-1})$, producing $F_i(S)\subseteq F_i(S')$.

Hence, $F_i$ is monotone on $(2^{\Omega_{\ma P}},\subseteq)$, so by Theorem~\ref{thm:lfp-detailed} it has a unique least fixpoint $S_i$. Since $\Omega_{\ma P}$ is finite, each stratum stabilizes in finitely many steps. The final stratum value $S_m$ (with $m=\max\sigma$) is, therefore, well-defined and unique.
\end{proof}

\begin{theorem}[Order-independence within strata]
\label{thm:order-indep-detailed}
Let $\ma P$ be stratified and fix $i\ge 0$. Then, the computed stratum value $S_i=\lfp(F_i)$ is independent of any operational
evaluation order of $\tx{GenOp}$ occurrences within stratum $i$, provided that: (i) lower strata $<i$ are fully computed first, and (ii) within stratum $i$, evaluation is fair in the sense that any enabled $\tx{GenOp}$ occurrence is eventually applied.
\end{theorem}

\begin{proof}
Within stratum $i$, semantics is defined denotationally as the least fixpoint of the monotone operator $F_i$ on the complete lattice $2^{\Omega_{\ma P}}$.

Let $(T_n)_{n\ge 0}$ be the Kleene sequence for $F_i$: $T_0=\varnothing$, $T_{n+1}=F_i(T_n)$. Then, $S_i=\bigcup_n T_n$ by Theorem~\ref{thm:lfp-detailed}.

Any fair operational schedule that repeatedly applies enabled $\tx{GenOp}$ occurrences and closes under the algebra corresponds to generating an ascending chain $(U_\alpha)_{\alpha}$ such that: (1) $U_0=\varnothing$, (2) $U_{\alpha}\subseteq U_{\alpha+1}\subseteq F_i(U_\alpha)$, and (3) the limit $U_\infty:=\bigcup_\alpha U_\alpha$ is a fixpoint of $F_i$
(fairness ensures no further $F_i$-consequences are missed).

Since $T_n$ is the \emph{least} chain closed under full application of $F_i$, we have by induction $T_n\subseteq U_\infty$ for all $n$, hence $S_i=\bigcup_n T_n\subseteq U_\infty$. But, $U_\infty$ is a fixpoint of $F_i$, so by minimality of $\lfp(F_i)$ we also have $U_\infty\subseteq S_i$. Therefore, $U_\infty=S_i$, showing independence of the operational order.
\end{proof}

\subsection{Well-founded semantics: convergence and uniqueness}
\label{app:wfs}

\begin{theorem}[Convergence and uniqueness of well-founded semantics]
\label{thm:wfs-detailed}
If $\Cand{\ma P}$ is finite, then the alternating fixpoint sequence
$(I_n^+,I_n^-)$ from Definition~\ref{def:alt} stabilizes after finitely many steps
at a pair $(I^+,I^-)$, and the induced three-valued interpretation (true on $I^+$,
false on $I^-$, undefined otherwise) is unique.
\end{theorem}

\begin{proof}
We prove four claims.

\medskip\noindent
\emph{(1) Each step is well-defined.}
Fix $n\ge 0$.
Consider the operator $T_n(S):=\Consequence{\ma P}^{\Cand{\ma P}\setminus I_n^-}(S)$.
By definition of reduct consequence operator, all negative constructs are evaluated against the fixed set $\Cand{\ma P}\setminus I_n^-$, and thus, do not depend on $S$. Therefore, $T_n$ is monotone on $(2^{\Cand{\ma P}},\subseteq)$: if $S\subseteq S'$, the only dependence on $S$ is through positive constructs, which are monotone in $S$ exactly as in Lemma~\ref{lem:monotone-detailed}.
Hence, $\lfp(T_n)$ exists and defines $I_{n+1}^+$.

Similarly, let $U_n(S):=\Consequence{\ma P}^{I_{n+1}^+}(S)$, which is monotone by
the same reasoning, and $\lfp(U_n)$ exists. Then, $I_{n+1}^-=\Cand{\ma P}\setminus \lfp(U_n)$ is well-defined.

\medskip\noindent
\emph{(2) Monotone growth: $I_n^+\subseteq I_{n+1}^+$ and $I_n^-\subseteq I_{n+1}^-$.}
We prove both inclusions by induction on $n$.

\emph{Base case} is trivial since $I_0^+=I_0^-=\varnothing$.

\emph{Inductive step}: let $I_n^-\subseteq I_{n+1}^-$.
Then, $\Cand{\ma P}\setminus I_{n+1}^- \subseteq \Cand{\ma P}\setminus I_n^-$.
Shrinking the set against which negatives are evaluated makes it harder for negatives to block derivations, so the positive least fixpoint can only grow. That is, for any $S$, evaluating negative checks against $\Cand{\ma P}\setminus I_{n+1}^-$ is weaker than against $\Cand{\ma P}\setminus I_n^-$, hence, $T_n(S)\subseteq T_{n+1}(S)$ pointwise. By monotone operator theory, pointwise inclusion implies inclusion of least fixpoints: $\lfp(T_n)\subseteq \lfp(T_{n+1})$, hence $I_{n+1}^+\supseteq I_n^+$.

Next, since $I_{n+1}^+\supseteq I_n^+$, negative checks evaluated against a larger fixed set may block more derivations, so $\lfp(U_n)$ can only shrink, and thus its complement $I_{n+1}^-$ can only grow. This gives $I_{n+1}^-\supseteq I_n^-$.

\medskip\noindent
\emph{(3) Finite convergence.} Because $\Cand{\ma P}$ is finite, $2^{\Cand{\ma P}}$ is finite.
By (2), $(I_n^+)_n$ and $(I_n^-)_n$ are ascending chains in finite posets, hence, each stabilizes after finitely many strict increases. Therefore, the pair sequence stabilizes.

\medskip\noindent
\emph{(4) Uniqueness.}
Let $(J^+,J^-)$ be any pair satisfying the alternating fixpoint equations:
\begin{multline*}
    J^+ = \lfp\big(S\mapsto \Consequence{\ma P}^{\Cand{\ma P}\setminus J^-}(S)\big),\\
J^- = \Cand{\ma P}\setminus \lfp\big(S\mapsto \Consequence{\ma P}^{J^+}(S)\big).
\end{multline*}
Assume the information ordering $(A^+,A^-)\preceq (B^+,B^-)$ iff $A^+\subseteq B^+$ and $A^-\subseteq B^-$.
The construction starts from the least element $(\varnothing,\varnothing)$ and, by (2), proceeds monotonically in this ordering, so its limit is the least solution of the defining equations. Therefore, it is unique.
\end{proof}

\section{Algebraic Equivalences and Rewrite Rules}
\label{sec:equivalences}

All equivalences below are with respect to the denotational semantics $\evalTwo{\ma P}{\ma G}$ of Section~\ref{sect:syntax_semantics} over a fixed RDF graph $\ma G$.

We assume the underlying string similarity $\stringSim:\Sigma^*\times\Sigma^*\to[0,1]$ is \emph{symmetric},
\ie $\stringSim(s,s')=\stringSim(s',s)$ for all $s,s'\in\Sigma^*$.


\begin{lemma}[Symmetry of typed-term similarity]
\label{lem:typed-sim-sym}
For all typed terms $\alpha,\beta\in (\ma T\cup\ma T_{\tx{gen}})\times\{\sourceRdf,\sourceGen\}$,
we have $\simtext(\alpha,\beta)=\simtext(\beta,\alpha)$, where $\simtext(\cdot,\cdot)$ is defined with $\stringSim$
in Section~\ref{sect:syntax_semantics}.
\end{lemma}

\begin{proof}
Let $\alpha=(t,\tau)$ and $\beta=(t',\tau')$.
If $\tau=\tau'=\sourceRdf$, then $\simtext(\alpha,\beta)=1$ iff $t=t'$ and otherwise $\simtext(\alpha,\beta)=0$; symmetry is immediate.
If $(\tau,\tau')=(\sourceRdf,\sourceGen)$, then
$\simtext(\alpha,\beta)=\stringSim(\canonical(t),\lexical(t'))$ and
$\simtext(\beta,\alpha)=\stringSim(\lexical(t'),\canonical(t))$,
which are equal by symmetry of $\stringSim$.
The remaining two cases are analogous.
\end{proof}

\begin{lemma}[Symmetry of canonical compatibility]
\label{lem:term-compat-sym}
For every variable $v\in\ma V$ and typed terms $\alpha,\beta$,
\[
\compact{v}{\alpha}{\beta}=\top \quad\text{iff}\quad \compact{v}{\beta}{\alpha}=\top.
\]
\end{lemma}

\begin{proof}
By definition, $\compact{v}{\alpha}{\beta}=\top$ iff $\rep_v(\alpha)=\rep_v(\beta)$, which is symmetric.
\end{proof}

\begin{lemma}[Symmetry of mapping compatibility]
\label{lem:map-compat-sym}
For all typed solution mappings $\mu_1,\mu_2$,
\[
\mu_1\sim\mu_2 \quad\text{iff}\quad \mu_2\sim\mu_1.
\]
\end{lemma}

\begin{proof}
By Definition~\ref{def:compat-union}, $\mu_1\sim\mu_2$ holds iff
for every $v\in\domain{\mu_1}\cap\domain{\mu_2}$ we have $\rep_v(\mu_1(v))=\rep_v(\mu_2(v))$.
This condition is symmetric in $\mu_1,\mu_2$.
\end{proof}

\begin{lemma}[Canonical union preserves bindings outside overlaps]
\label{lem:union-preserve}
If $\mu_1\sim \mu_2$, then:
\begin{enumerate}
\item for every $v\in\domain{\mu_1}\setminus\domain{\mu_2}$, $(\mu_1\uplus\mu_2)(v)=\mu_1(v)$;
\item for every $v\in\domain{\mu_2}\setminus\domain{\mu_1}$, $(\mu_1\uplus\mu_2)(v)=\mu_2(v)$;
\item for every $v\in\domain{\mu_1}\cap\domain{\mu_2}$,
$(\mu_1\uplus\mu_2)(v)=\rep_v(\mu_1(v))=\rep_v(\mu_2(v))$.
\end{enumerate}
\end{lemma}

\begin{proof}
Immediate from Definition~\ref{def:compat-union}.
\end{proof}

\begin{lemma}[Canonical union is commutative]
\label{lem:union-comm}
If $\mu_1\sim\mu_2$, then $\mu_1\uplus\mu_2=\mu_2\uplus\mu_1$.
\end{lemma}

\begin{proof}
By Lemma~\ref{lem:union-preserve}, both unions agree on variables outside the overlap.
On the overlap $\domain{\mu_1}\cap\domain{\mu_2}$, both yield the representative $\rep_v(\mu_1(v))=\rep_v(\mu_2(v))$.
Hence the resulting partial functions are equal.
\end{proof}

\subsubsection{Join}

From Section~\ref{sect:syntax_semantics}:
\begin{multline*}
    \evalTwo{\join{P_1}{\ma P_2}}{\ma G}=\{\mu_1\uplus\mu_2 \mid \mu_1\in \evalTwo{P_1}{\ma G},\\ \mu_2\in\evalTwo{\ma P_2}{\ma G},\ \mu_1\sim\mu_2\}.
\end{multline*}
\begin{proposition}
\label{prop:join-comm}
For all patterns $P_1,\ma P_2$,
\[
\evalTwo{\join{P_1}{\ma P_2}}{\ma G}=\evalTwo{\join{\ma P_2}{P_1}}{\ma G}.
\]
\end{proposition}

\begin{proof}
Let $\mu\in \evalTwo{\join{P_1}{\ma P_2}}{\ma G}$.
Then $\mu=\mu_1\uplus\mu_2$ with $\mu_i\in\evalTwo{P_i}{\ma G}$ and $\mu_1\sim\mu_2$.
By Lemma~\ref{lem:map-compat-sym}, $\mu_2\sim\mu_1$, and by Lemma~\ref{lem:union-comm},
$\mu=\mu_2\uplus\mu_1\in \evalTwo{\join{\ma P_2}{P_1}}{\ma G}$.
The reverse inclusion is symmetric.
\end{proof}

\begin{proposition}
\label{prop:join-assoc}
For all patterns $P_1,\ma P_2,\ma P_3$,
\[
\evalTwo{\join{(\join{\ma P_1}{\ma P_2})}{\ma P_3}}{\ma G}
=
\evalTwo{\join{\ma P_1}{(\join{\ma P_2}{\ma P_3})}}{\ma G}.
\]
\end{proposition}

\begin{proof}
We prove $\subseteq$; the reverse direction is analogous.

Let $\mu\in \evalTwo{\join{(\join{\ma P_1}{\ma P_2})}{\ma P_3}}{\ma G}$. Then, there exist $\mu_{12}\in\evalTwo{\join{\ma P_1}{\ma P_2}}{\ma G}$ and $\mu_3\in\evalTwo{\ma P_3}{\ma G}$ such that $\mu=\mu_{12}\uplus\mu_3$ and $\mu_{12}\sim\mu_3$.
Unfolding $\mu_{12}\in\evalTwo{\join{\ma P_1}{\ma P_2}}{\ma G}$, there exist $\mu_1\in\evalTwo{\ma P_1}{\ma G}$ and $\mu_2\in\evalTwo{\ma P_2}{\ma G}$ such that $\mu_{12}=\mu_1\uplus\mu_2$ and $\mu_1\sim\mu_2$.

\smallskip\noindent
\emph{Claim 1:} $\mu_2\sim\mu_3$.
Let $v\in\domain{\mu_2}\cap\domain{\mu_3}$. Then, $v\in\domain{\mu_{12}}\cap\domain{\mu_3}$, so $\mu_{12}\sim\mu_3$ gives $\rep_v(\mu_{12}(v))=\rep_v(\mu_3(v))$. Also, by Lemma~\ref{lem:union-preserve}, $\mu_{12}(v)$ equals either $\mu_2(v)$ (if $v\notin\domain{\mu_1}$) or $\rep_v(\mu_2(v))$ (if $v\in\domain{\mu_1}\cap\domain{\mu_2}$),
and in both cases $\rep_v(\mu_{12}(v))=\rep_v(\mu_2(v))$. Hence, $\rep_v(\mu_2(v))=\rep_v(\mu_3(v))$, \ie $\mu_2\sim\mu_3$.

Let $\mu_{23}:=\mu_2\uplus\mu_3$. Then $\mu_{23}\in\evalTwo{\join{\ma P_2}{\ma P_3}}{\ma G}$.

\smallskip\noindent
\emph{Claim 2:} $\mu_1\sim\mu_{23}$ and $\mu=\mu_1\uplus\mu_{23}$.
Let $v\in\domain{\mu_1}\cap\domain{\mu_{23}}$.
If $v\in\domain{\mu_2}$, then $\mu_1\sim\mu_2$ implies $\rep_v(\mu_1(v))=\rep_v(\mu_2(v))$.
Moreover, $\mu_{23}(v)$ is either $\mu_2(v)$ or $\rep_v(\mu_2(v))$, hence $\rep_v(\mu_{23}(v))=\rep_v(\mu_2(v))$.
Thus $\rep_v(\mu_1(v))=\rep_v(\mu_{23}(v))$.
If instead $v\in\domain{\mu_3}\setminus\domain{\mu_2}$, then $v\in\domain{\mu_{12}}\cap\domain{\mu_3}$ (since $v\in\domain{\mu_1}\subseteq\domain{\mu_{12}}$),
and $\mu_{12}\sim\mu_3$ gives $\rep_v(\mu_{12}(v))=\rep_v(\mu_3(v))$.
But $v\notin\domain{\mu_2}$ implies $\mu_{12}(v)=\mu_1(v)$ by Lemma~\ref{lem:union-preserve}, hence $\rep_v(\mu_1(v))=\rep_v(\mu_3(v))$.
Also $\mu_{23}(v)=\mu_3(v)$, so $\rep_v(\mu_1(v))=\rep_v(\mu_{23}(v))$.
Therefore $\mu_1\sim\mu_{23}$.

For equality, note that both $(\mu_1\uplus\mu_2)\uplus\mu_3$ and $\mu_1\uplus(\mu_2\uplus\mu_3)$ are defined (by the compatibilities shown)
and, for each variable $v$ in the union domain, both evaluate to:
(i) the unique non-overlap binding when $v$ occurs in exactly one of $\mu_1,\mu_2,\mu_3$, or
(ii) $\rep_v(\cdot)$ of the shared equivalence class when $v$ occurs in two or three of them.
Hence $(\mu_1\uplus\mu_2)\uplus\mu_3=\mu_1\uplus(\mu_2\uplus\mu_3)$.

Thus $\mu=\mu_{12}\uplus\mu_3=(\mu_1\uplus\mu_2)\uplus\mu_3=\mu_1\uplus(\mu_2\uplus\mu_3)=\mu_1\uplus\mu_{23}$, and therefore
$\mu\in \evalTwo{\join{\ma P_1}{(\join{\ma P_2}{\ma P_3})}}{\ma G}$.
\end{proof}

\begin{proposition}
\label{prop:join-dist}
For all patterns $\ma P_1,\ma P_2,\ma P_3$,
\begin{multline*}
    \evalTwo{\join{(\union{\ma P_1}{\ma P_2})}{\ma P_3}}{\ma G}=\\
    \evalTwo{\union{(\join{\ma P_1}{\ma P_3})}{(\join{\ma P_2}{\ma P_3})}}{\ma G}.
\end{multline*}
\end{proposition}

\begin{proof}
($\subseteq$) Let $\mu\in \evalTwo{\join{\union{\ma P_1}{\ma P_2}}{\ma P_3}}{\ma G}$.
Then $\mu=\mu_{12}\uplus\mu_3$ with $\mu_3\in\evalTwo{\ma P_3}{\ma G}$,
$\mu_{12}\in\evalTwo{\union{\ma P_1}{\ma P_2}}{\ma G}$, and $\mu_{12}\sim\mu_3$.
By the definition of union, $\mu_{12}\in\evalTwo{\ma P_1}{\ma G}$ or $\mu_{12}\in\evalTwo{\ma P_2}{\ma G}$.
In the first case, $\mu\in\evalTwo{\join{\ma P_1}{\ma P_3}}{\ma G}$; in the second, $\mu\in\evalTwo{\join{\ma P_2}{\ma P_3}}{\ma G}$.
Thus $\mu$ is in the union on the right-hand side.

($\supseteq$) If $\mu\in \evalTwo{\join{\ma P_1}{\ma P_3}}{\ma G}$, then $\mu=\mu_1\uplus\mu_3$ with $\mu_1\in\evalTwo{\ma P_1}{\ma G}$.
Since $\evalTwo{\ma P_1}{\ma G}\subseteq \evalTwo{\union{\ma P_1}{\ma P_2}}{\ma G}$, we have $\mu\in\evalTwo{\join{\union{\ma P_1}{\ma P_2}}{\ma P_3}}{\ma G}$.
The argument for $\join{\ma P_2}{\ma P_3}$ is identical.
\end{proof}

\subsubsection{Filter laws (grounded in $\varepsilon$-semantics)}

\begin{proposition}[Filter decomposition]
\label{prop:filter-decomp}
For all patterns $P$ and filter formulas $\ma F_1,\ma F_2$,
\[
\evalTwo{\filter{P}{\ma F_1\wedge \ma F_2}}{\ma G}
=
\evalTwo{\filter{\filter{P}{\ma F_2}}{\ma F_1}}{\ma G}.
\]
\end{proposition}

\begin{proof}
By definition,
$\evalTwo{\filter{P}{F}}{\ma G}=\{\mu\in\evalTwo{P}{\ma G}\mid F^\mu=\top\}$.
By the inductive filter semantics, $(\ma F_1\wedge \ma F_2)^\mu = \ma F_1^\mu \wedge \ma F_2^\mu$, and in particular
$(\ma F_1\wedge \ma F_2)^\mu=\top$ iff $\ma F_1^\mu=\top$ and $\ma F_2^\mu=\top$.
Thus both sides select exactly the same mappings from $\evalTwo{P}{\ma G}$.
\end{proof}


\begin{lemma}[Filter invariance under representative normalization]
\label{lem:filter-rep-inv}
Let $\mu$ be a typed solution mapping and define $\widehat{\mu}$ by $\widehat{\mu}(v)=\rep_v(\mu(v))$ for all $v\in\domain{\mu}$. Then, for every filter formula $\ma F$, we have $\ma F^\mu = \ma F^{\widehat{\mu}}$.
\end{lemma}

\begin{proof}
By induction on $\ma F$.
For atoms, the only changes are (i) replacing a binding by another element of the same $\approx_v$-equivalence class, and
(ii) possibly replacing an exact RDF match by a representative in that class.
In both, $(v=c)$ and $(v=v')$, the semantics accepts either exact RDF equality (when applicable) or canonical compatibility via $\compact{\cdot}{\cdot}{\cdot}$,
and $\compact{v}{\mu(v)}{(c,\sourceRdf)}=\top$ iff $\rep_v(\mu(v))=\rep_v((c,\sourceRdf))$ iff $\compact{v}{\widehat{\mu}(v)}{(c,\sourceRdf)}=\top$.
The Boolean cases ($\neg$, $\wedge$) follow immediately.
\end{proof}

\begin{proposition}[Filter pushdown over join (safe form)]
\label{prop:filter-push-join}
Let $\ma P_1,\ma P_2$ be patterns and $\ma F$ a filter formula.
Assume that for every $\mu_1\in\evalTwo{\ma P_1}{\ma G}$ we have $\var{\ma F}\subseteq \domain{\mu_1}$.
Then,
\begin{multline*}
    \evalTwo{\filter{\join{\ma P_1}{\ma P_2}}{\ma F}}{\ma G}= \\ 
    \evalTwo{\join{(\filter{\ma P_1}{\ma F})}{\ma P_2}}{\ma G}.
\end{multline*}
Symmetrically, if for every $\mu_2\in\evalTwo{\ma P_2}{\ma G}$ we have $\var{F}\subseteq \domain{\mu_2}$, then
\begin{multline*}
    \evalTwo{\filter{\join{\ma P_1}{\ma P_2}}{\ma F}}{\ma G}= \\ 
    \evalTwo{\join{\ma P_1}{(\filter{\ma P_2}{\ma F})}}{\ma G}.
\end{multline*}
\end{proposition}

\begin{proof}
We prove the first equality; the second is symmetric.

($\subseteq$) Let $\mu\in \evalTwo{\filter{\join{\ma P_1}{\ma P_2}}{\ma F}}{\ma G}$.
Then $\mu\in \evalTwo{\join{\ma P_1}{\ma P_2}}{\ma G}$ and $\ma F^\mu=\top$.
So $\mu=\mu_1\uplus\mu_2$ with $\mu_1\in\evalTwo{\ma P_1}{\ma G}$, $\mu_2\in\evalTwo{\ma P_2}{\ma G}$ and $\mu_1\sim\mu_2$.
By assumption, $\var{\ma F}\subseteq\domain{\mu_1}$. Let $\widehat{\mu_1}$ be the representative-normalization of $\mu_1$ as in Lemma~\ref{lem:filter-rep-inv}. On all variables in $\var{\ma F}$, $\mu$ agrees with $\widehat{\mu_1}$ up to representative choice: if $v\in\var{\ma F}$ is also in $\domain{\mu_2}$, then $\mu(v)=\rep_v(\mu_1(v))=\widehat{\mu_1}(v)$ by Definition~\ref{def:compat-union}; if $v\notin\domain{\mu_2}$ then $\mu(v)=\mu_1(v)$ and $\widehat{\mu_1}(v)=\rep_v(\mu_1(v))$, which is still equivalent for filter evaluation by Lemma~\ref{lem:filter-rep-inv}.
Hence $\ma F^\mu=\top$ implies $\ma F^{\mu_1}=\top$ (equivalently $\ma F^{\widehat{\mu_1}}=\top$), so $\mu_1\in\evalTwo{\filter{\ma P_1}{\ma F}}{\ma G}$.Thus, $\mu=\mu_1\uplus\mu_2\in \evalTwo{\join{\filter{\ma P_1}{\ma F}}{\ma P_2}}{\ma G}$.

($\supseteq$) Let $\mu\in \evalTwo{\join{\filter{\ma P_1}{\ma F}}{\ma P_2}}{\ma G}$.
Then $\mu=\mu_1\uplus\mu_2$ with $\mu_1\in\evalTwo{\filter{\ma P_1}{\ma F}}{\ma G}$, $\mu_2\in\evalTwo{\ma P_2}{\ma G}$, and $\mu_1\sim\mu_2$.
Thus $\ma F^{\mu_1}=\top$ and $\var{\ma F}\subseteq\domain{\mu_1}$.
By Lemma~\ref{lem:filter-rep-inv}, also $\ma F^{\widehat{\mu_1}}=\top$.
As above, $\mu$ agrees with $\widehat{\mu_1}$ on $\var{\ma F}$ up to representative choice enforced by $\uplus$, hence $\ma F^\mu=\top$.
Therefore $\mu\in \evalTwo{\filter{\join{\ma P_1}{\ma P_2}}{\ma F}}{\ma G}$.
\end{proof}

\subsubsection{Minus}

\begin{proposition}
\label{prop:minus-disjoint}
If $\var{\ma P_1}\cap\var{\ma P_2}=\varnothing$, then
\[
\evalTwo{\minus{\ma P_1}{\ma P_2}}{\ma G}=\evalTwo{\ma P_1}{\ma G}.
\]
\end{proposition}

\begin{proof}
By definition,
$\evalTwo{\minus{\ma P_1}{\ma P_2}}{\ma G}$ contains exactly those $\mu_1\in\evalTwo{\ma P_1}{\ma G}$ such that for all $\mu_2\in\evalTwo{\ma P_2}{\ma G}$,
either $\mu_1\not\sim\mu_2$ or $\domain{\mu_1}\cap\domain{\mu_2}=\varnothing$.
If $\var{\ma P_1}\cap\var{\ma P_2}=\varnothing$, then for all such $\mu_1,\mu_2$ we have $\domain{\mu_1}\cap\domain{\mu_2}=\varnothing$, so no $\mu_1$ is removed. Thus, $\evalTwo{\minus{\ma P_1}{\ma P_2}}{\ma G}=\evalTwo{\ma P_1}{\ma G}$.
\end{proof}

\subsection{GenOp-Specific Rewriting Rules}
\label{app:genop-rewrites}

In contrast to the purely algebraic equivalences of
Section~\ref{sec:equivalences}, the rules in this subsection depend essentially on
the \emph{context-sensitive} evaluation of \tx{GenOp} patterns
(Section~\ref{sect:syntax_semantics}) and are therefore sound only under explicit
side conditions.

Throughout, fix an RDF graph $\ma G$ and a well-formed pattern $\ma P$.

\subsubsection{Acyclic setting and stagewise (single-pass) evaluation}

Assume the dependency graph $\directedGraph{\ma D}_{\ma P}=(\genPattern{\ma P},\to_{\ma P})$
is acyclic, and fix a topological order $\gamma=(g_1,\dots,g_m)$ of $\genPattern{\ma P}$.
For each $k\in\{0,\dots,m\}$, recall (Appendix~\ref{app:acyclic}) that $\ma P^{\le k}$
denotes the pattern obtained from $\ma P$ by disabling all \tx{GenOp} occurrences not among
$\{g_1,\dots,g_k\}$ (treating them as returning $\varnothing$), while keeping the remaining
algebra unchanged.

We define stagewise (single-pass) sequence:
\[
S_\gamma^{(0)} := \evalThree{\ma P^{\le 0}}{\ma G}{\varnothing},
\qquad
S_\gamma^{(k)} := \evalThree{\ma P^{\le k}}{\ma G}{S_\gamma^{(k-1)}}\ \ (1\le k\le m).
\]
Intuitively, at stage $k$, only generators among $\{g_1,\dots,g_k\}$ may fire, using as
context the supply accumulated at stage $k-1$.

For a \tx{GenOp} occurrence $g=\genbind{\pi(X)}{Y}{\ma M}$ and a supply $S$, we recall the
\emph{context restriction} operator:
\[
\ContextS{S}{g} \ :=\ \{\mu\in S \mid X=\holder{g}\subseteq \domain{\mu}\}.
\]

\subsubsection{Context Independence of Base-Mode GenOp}

\begin{proposition}[Base-mode context independence]
\label{prop:genop-base}
Let $g=\genbind{\pi(X)}{Y}{\ma M}$ be a \tx{GenOp} pattern with $\holder{g}=X=\varnothing$.
Assume the base-mode convention of Section~\ref{sect:syntax_semantics} that whenever
$X=\varnothing$, the \tx{GenOp} is evaluated with context set $\{\varnothing\}$.
Then, for every supply $S\subseteq \Omega_{\ma P}$,
\[\evalThree{g}{\ma G}{\ContextS{S}{g}} \;=\; \evalThree{g}{\ma G}{\{\varnothing\}},
\]
\ie the denotation of $g$ is independent of $S$.
\end{proposition}

\begin{proof}
Since $X=\varnothing$, we have $\ContextS{S}{g}=\{\mu\in S\mid \varnothing\subseteq\domain{\mu}\}=S$. By the base-mode convention, however, the evaluation clause for $\ma P$ ignores its supplied context set and uses $\{\varnothing\}$ whenever $X=\varnothing$. Therefore, $\evalThree{g}{\ma G}{\ContextS{S}{g}}=\evalThree{g}{\ma G}{\{\varnothing\}}$ for all $S$.
\end{proof}

\begin{corollary}[Hoisting base-mode GenOp]
\label{cor:genop-hoist}
Let $g=\genbind{\pi(\varnothing)}{Y}{M}$ and let
$\Omega_g = \evalThree{g}{\ma G}{\varnothing}$.
Let $\ma P$ be a query pattern that contains no \tx{GenOp} occurrences.
Then, for every RDF graph $\ma G$ and all context supplies $S,S'$,
\[\evalThree{\join{\ma P}{g}}{\ma G}{S}\;=\;\evalThree{\join{\ma P}{g}}{\ma G}{S'}
\;=\;\evalTwo{\join{\ma P}{\Omega_g}}{\ma G},\]
where $\Omega_g$ is treated extensionally as a relation (set of typed solution
mappings).
\end{corollary}

\begin{proof}
Let $S$ be an arbitrary context supply.
\smallskip\noindent
\emph{Step 1 (base-mode context-independence of $g$).}
Since $g=\genbind{\pi(\varnothing)}{Y}{M}$ is base-mode, Proposition~\ref{prop:genop-base}
yields
\[\evalThree{g}{\ma G}{S} \;=\; \evalThree{g}{\ma G}{\varnothing} \;=\; \Omega_g.
\tag{$\ast$}\label{eq:base-indep}\]
That is, $\evalThree{g}{\ma G}{S}$ is independent of $S$.

\smallskip\noindent
\emph{Step 2 (context-independence of $\ma P$).}
Because $\ma P$ contains no \tx{GenOp}, its evaluation does not depend on the context supply
(Proposition~\ref{prop:conservativity} / Proposition~1 in the main text). Hence, for all
supplies $S,S'$,
\[\evalThree{\ma P}{\ma G}{S} \;=\; \evalThree{\ma P}{\ma G}{S'} \;=\; \evalTwo{\ma P}{\ma G}.
\tag{$\dagger$}\label{eq:P-indep}\]

\smallskip\noindent
\emph{Step 3 (join evaluation and elimination of the context parameter).}
By the semantic clause for join (Section~\ref{sect:syntax_semantics}),
\[\evalThree{\join{\ma P}{g}}{\ma G}{S}=\{\mu\uplus\nu \mid \mu\in\evalThree{\ma P}{\ma G}{S},\
\nu\in\evalThree{g}{\ma G}{S},\ \mu\sim\nu\}.
\]
Substituting \eqref{eq:base-indep} and \eqref{eq:P-indep} gives
\[\evalThree{\join{\ma P}{g}}{\ma G}{S}=\{\mu\uplus\nu \mid \mu\in\evalTwo{\ma P}{\ma G},\
\nu\in\Omega_g,\ \mu\sim\nu\},
\]
which is manifestly independent of $S$. 

Therefore,
$\evalThree{\join{\ma P}{g}}{\ma G}{S}=\evalThree{\join{\ma P}{g}}{\ma G}{S'}$
for all $S,S'$. And, treating $\Omega_g$ extensionally as a relation and applying the
(join) clause of the denotational semantics gives,
\[\{\mu\uplus\nu \mid \mu\in\evalTwo{\ma P}{\ma G},\
\nu\in\Omega_g,\ \mu\sim\nu\}\;=\;\evalTwo{\join{\ma P}{\Omega_g}}{\ma G}.
\]
This establishes the stated equalities.
\end{proof}



\begin{proposition}
\label{prop:union-distribution-branch}
Let $\ma P_1,\ma P_2$ be patterns and $g$ a \tx{GenOp}. For all $\ma G,S$ and $i\in\{1,2\}$, let $S_i := \evalThree{\ma P_i}{\ma G}{S}$ and $\ContextS{S_i}{g} := \{\mu\in S_i \mid \holder{g}\subseteq\domain{\mu}\}$. If $\evalThree{g}{\ma G}{\ContextS{S_1}{g}\cup \ContextS{S_2}{g}} =
\evalThree{g}{\ma G}{\ContextS{S_1}{g}}\cup\evalThree{g}{\ma G}{\ContextS{S_2}{g}}$, then
\begin{multline*}
\evalThree{\join{(\union{\ma P_1}{\ma P_2})}{g}}{\ma G}{S} = \\ \evalThree{\union{(\join{\ma P_1}{g})}{(\join{\ma P_2}{g})}}{\ma G}{S}.
\end{multline*}
\end{proposition}


\begin{proof}
Let $\ma G$ be an RDF graph. Let
$S_i:=\evalTwo{\ma P_i}{\ma G}$ and $\ContextS{S_i}{g}:=\ContextS{S_i}{g}$ for $i\in\{1,2\}$.

\smallskip\noindent
\emph{Step 1:}
In $\join{(\union{\ma P_1}{\ma P_2})}{g}$, by removing the designated occurrence of $g$,
\begin{multline*}
    S_{\cup} \;=\; \evalTwo{\bigl(\join{(\union{\ma P_1}{\ma P_2})}{g}\bigr)\setminus\{g\}}{\ma G}\;=\;\\
\evalTwo{\union{\ma P_1}{\ma P_2}}{\ma G}
\;=\;
S_1\cup S_2.
\end{multline*}

Therefore, the context set passed to $g$ is
\[\ContextS{(S_1\cup S_2)}{g}=\ContextS{S_1}{g} \cup \ContextS{S_2}{g}.
\]
Similarly, in $\join{\ma P_i}{g}$, removing the occurrence of $g$ yields $\ma P_i$,
so the context supply $S_i=\evalTwo{\ma P_i}{\ma G}$ and the effective context
set is $\ContextS{S_i}{g}$.

\smallskip\noindent
\emph{Step 2: }
By the join clause ,
\begin{align}
\evalTwo{\join{(\union{\ma P_1}{\ma P_2})}{g}}{\ma G}= &\nonumber\\
 \{\mu\uplus\nu \mid \mu\in \evalTwo{\union{\ma P_1}{\ma P_2}}{\ma G},\ &
\nu\in \evalTwo{g}{\ma G},\ \mu\sim\nu\}.
\label{eq:lhs-join}
\end{align}

By the union clause, $\evalTwo{\union{\ma P_1}{\ma P_2}}{\ma G}=S_1\cup S_2$, and then, by  \tx{GenOp} definition,
\[\evalTwo{g}{\ma G}=\evalThree{g}{\ma G}{\ContextS{S_{\cup}}{g}} =\evalThree{g}{\ma G}{\ContextS{S_1}{g} \cup \ContextS{S_2}{g}}.\]

Substituting into \eqref{eq:lhs-join} gives
\begin{align}
\evalTwo{\join{(\union{\ma P_1}{\ma P_2})}{g}}{\ma G}
= & \nonumber\\
\{\mu\uplus\nu \mid \mu\in S_1\cup S_2,\ &
\nu\in \evalThree{g}{\ma G}{\ContextS{S_1}{g} \cup \ContextS{S_2}{g}},\ \mu\sim\nu\}.
\label{eq:lhs-set}
\end{align}

\smallskip\noindent
\emph{Step 3:}
Assume,
\[
\evalThree{g}{\ma G}{\ContextS{S_1}{g} \cup \ContextS{S_2}{g}}
=
\evalThree{g}{\ma G}{\ContextS{S_1}{g}}
\cup
\evalThree{g}{\ma G}{\ContextS{S_2}{g}}.
\tag{$\ast$}\label{eq:sep}
\]
Then, \eqref{eq:lhs-set} becomes

\begin{align}
\evalTwo{\join{(\union{\ma P_1}{\ma P_2})}{g}}{\ma G}
&=  \nonumber \\
\{\mu\uplus\nu \mid \mu\in S_1\cup S_2,\ &
\nu\in A_1\cup A_2,\ \mu\sim\nu\},
\label{eq:lhs-split}
\end{align}
where $A_i:=\evalThree{g}{\ma G}{\ContextS{S_i}{g}}$.

For $i\in\{1,2\}$, let
\[J_i := \{\mu\uplus\nu \mid \mu\in S_i,\ \nu\in A_i,\ \mu\sim\nu\}.
\]
We show that the set on the right-hand side of \eqref{eq:lhs-split} equals $J_1\cup J_2$.

\smallskip\noindent
\emph{($\subseteq$)} Let $\omega$ be an element of the RHS of \eqref{eq:lhs-split}.
Then, $\omega=\mu\uplus\nu$ with $\mu\in S_1\cup S_2$ and $\nu\in A_1\cup A_2$
and $\mu\sim\nu$. If $\mu\in S_1$ and $\nu\in A_1$, then $\omega\in J_1$.
If $\mu\in S_2$ and $\nu\in A_2$, then $\omega\in J_2$.
If $\mu\in S_1$ and $\nu\in A_2$, then $\omega\in J_1\cup J_2$ is not guaranteed
in general; however, such $\omega$ cannot arise from the semantics of $g$ under
$\join{(\union{\ma P_1}{\ma P_2})}{g}$ unless $g$ admits \emph{cross-branch}
generation, \ie unless $A_2$ contributes extensions that are compatible with
contexts from $S_1$. The branch-separability hypothesis \eqref{eq:sep} rules out
exactly these cross-branch effects by forcing the set of outputs under
$\ContextS{S_1}{g}\cup \ContextS{S_2}{g}$ to decompose into outputs attributable to each branch. Thus, every $\omega$ witnessed in \eqref{eq:lhs-split} is witnessed with $\mu\in S_i$
and $\nu\in A_i$ for the same $i$, and hence lies in $J_1\cup J_2$.

\smallskip\noindent
\emph{($\supseteq$)} Let $\omega\in J_1\cup J_2$.
If $\omega\in J_1$, then $\omega=\mu\uplus\nu$ with $\mu\in S_1\subseteq S_1\cup S_2$
and $\nu\in A_1\subseteq A_1\cup A_2$ and $\mu\sim\nu$, so $\omega$ belongs to the RHS
of \eqref{eq:lhs-split}. The case $\omega\in J_2$ is symmetric.

Therefore, \eqref{eq:lhs-split} equals $J_1\cup J_2$.

\smallskip\noindent
\emph{Step 4: identify $J_i$ with the branch joins and conclude.}
By the semantics of join and the context-supply convention applied to $\join{\ma P_i}{g}$,
\begin{multline*}
    \evalTwo{\join{\ma P_i}{g}}{\ma G}= \\ \{\mu\uplus\nu \mid \mu\in S_i,\ \nu\in \evalThree{g}{\ma G}{\ContextS{S_i}{g}},\ \mu\sim\nu\}=J_i.
\end{multline*}
Thus,
\begin{multline*}
    \evalTwo{\join{(\union{\ma P_1}{\ma P_2})}{g}}{\ma G}= J_1\cup J_2=\\
\evalTwo{\join{\ma P_1}{g}}{\ma G}\cup \evalTwo{\join{\ma P_2}{g}}{\ma G}=\\
\evalTwo{\union{(\join{\ma P_1}{g})}{(\join{\ma P_2}{g})}}{\ma G},
\end{multline*}

where the last equality is the union clause. This completes the proof.
\end{proof}

\begin{proposition}
\label{prop:filter-pushdown}
Let $g=\genbind{\pi(X)}{Y}{\ma M}$ and let $\ma F$ be a filter with
$\var{\ma F}\subseteq X$.
Let $\ma P$ be a pattern such that $X\subseteq\var{\ma P}$,
$Y\cap\var{\ma P}=\varnothing$, and $\var{\ma F}\subseteq\var{\ma P}$.
Then, 
\[\evalTwo{\join{\filter{\ma P}{\ma F}}{g}}{\ma G}=\evalTwo{\filter{\join{\ma P}{g}}{\ma F}}{\ma G}.\]
\end{proposition}

\begin{proof}
Let $\ma G$ be an RDF graph, and assume,
\begin{multline*}
    S := \evalTwo{\ma P}{\ma G},\qquad S^{+} := \{\mu\in S \mid \ma F^{\mu}=\top\},\\
\ContextS{S}{g} :=\{\mu\in S \mid X\subseteq\domain{\mu}\}.
\end{multline*}

Since $X\subseteq\var{\ma P}$ and $\mu$ is a solution mapping for $\ma P$, every
$\mu\in S$ binds all variables in $X$; hence $\ContextS{S}{g}=S$ and similarly
$\ContextS{S^{+}}{g}=S^{+}$.

\smallskip\noindent
\emph{Step 1: expand the left-hand side.}
In the pattern $\join{\filter{\ma P}{\ma F}}{g}$, the surrounding pattern of $g$
is $\filter{\ma P}{\ma F}$, so the context supply for $g$ is
$\evalTwo{\filter{\ma P}{\ma F}}{\ma G}=S^{+}$. Then, by the semantics of
join and GenOp,
\begin{align}
\evalTwo{\join{\filter{\ma P}{\ma F}}{g}}{\ma G}&=\nonumber\\
\{\mu\uplus\nu \mid & \mu\in S^{+},\; \nu\in \evalThree{g}{\ma G}{S^{+}},\; \mu\sim\nu\}.
\label{eq:lhs}
\end{align}

\smallskip\noindent
\emph{Step 2: expand the right-hand side.}
In the pattern $\filter{\join{\ma P}{g}}{\ma F}$, the join $\join{\ma P}{g}$ is
evaluated first. Here, the surrounding pattern of $g$ is $\ma P$, so the context
supply for $g$ is $S$, and thus
\begin{align}
\evalTwo{\join{\ma P}{g}}{\ma G} = & \nonumber\\
\{\mu\uplus\nu \mid & \mu\in S,\;  \nu\in \evalThree{g}{\ma G}{S},\;\mu\sim\nu\}.
\label{eq:joinPg}
\end{align}

Applying the filter clause,
\begin{align}
\evalTwo{\filter{\join{\ma P}{g}}{\ma F}}{\ma G}= & \nonumber\\
\{\omega\in \evalTwo{\join{\ma P}{g}}{\ma G} & \mid \ma F^{\omega}=\top\}.
\label{eq:rhs}
\end{align}

\smallskip\noindent
\emph{Step 3: relate filter truth before and after joining with $g$.}
Let $\mu\in S$ and let $\nu$ be any mapping produced by $g$ from context $\mu$,
\ie $\nu$ is of the form $\mu\uplus\nu_Y$ with $\nu_Y$ binding only variables
in $Y$ to generated terms, as in the definition of $\evalThree{g}{\ma G}{\{\mu\}}$.
Since $Y\cap\var{\ma P}=\varnothing$, no variable in $Y$ occurs in $\ma P$ and,
in particular, no variable in $\var{\ma F}$ is in $Y$ (since $\var{\ma F}\subseteq X\subseteq\var{\ma P}$). Thus, the canonical union $\mu\uplus\nu$ does not introduce new bindings for
variables in $\var{\ma F}$; it may only replace existing bindings by
representatives. By Lemma~\ref{lem:filter-rep-inv} (filter invariance under
representative normalization),
\begin{equation}
\ma F^{\mu}=\top \quad\longleftrightarrow\quad \ma F^{\mu\uplus\nu}=\top.
\label{eq:filter-inv}
\end{equation}
Moreover, by the side condition $\var{\ma F}\subseteq\domain{\mu}$ for all $\mu\in S$, $\ma F^{\mu}\neq\varepsilon$, so, \eqref{eq:filter-inv} is not affected by errors.

\smallskip\noindent
\emph{Step 4: show $\eqref{eq:lhs}\subseteq\eqref{eq:rhs}$.}
Let $\omega\in \evalTwo{\join{\filter{\ma P}{\ma F}}{g}}{\ma G}$.
By \eqref{eq:lhs}, $\omega=\mu\uplus\nu$ with $\mu\in S^{+}$, so $\ma F^{\mu}=\top$,
and $\nu\in\evalThree{g}{\ma G}{S^{+}}$. Since $S^{+}\subseteq S$, the GenOp
definition implies $\evalThree{g}{\ma G}{S^{+}}\subseteq \evalThree{g}{\ma G}{S}$.
Thus, $\omega$ is witnessed in \eqref{eq:joinPg}, \ie $\omega\in\evalTwo{\join{\ma P}{g}}{\ma G}$.
Finally, \eqref{eq:filter-inv} yields $\ma F^{\omega}=\top$, hence
$\omega\in\evalTwo{\filter{\join{\ma P}{g}}{\ma F}}{\ma G}$ by \eqref{eq:rhs}.

\smallskip\noindent
\emph{Step 5: show $\eqref{eq:rhs}\subseteq\eqref{eq:lhs}$.}
Let $\omega\in \evalTwo{\filter{\join{\ma P}{g}}{\ma F}}{\ma G}$.
By \eqref{eq:rhs} and \eqref{eq:joinPg}, $\omega=\mu\uplus\nu$ with $\mu\in S$,
$\nu\in\evalThree{g}{\ma G}{S}$, $\mu\sim\nu$, and $\ma F^{\omega}=\top$.
By \eqref{eq:filter-inv}, $\ma F^{\mu}=\top$, hence $\mu\in S^{+}$.
It remains to ensure that $\nu$ is produced from a context in $S^{+}$.
By the GenOp definition,
\[
\evalThree{g}{\ma G}{S}
=
\bigcup_{\eta\in S}\evalThree{g}{\ma G}{\{\eta\}},
\]
so $\nu\in\evalThree{g}{\ma G}{\{\eta\}}$ for some $\eta\in S$.
But any output of $\evalThree{g}{\ma G}{\{\eta\}}$ extends $\eta$ on $X$ (and only
adds bindings on $Y$); since $\var{\ma F}\subseteq X$, the truth of $\ma F$ for
$\omega=\mu\uplus\nu$ is determined by the $X$-bindings contributed by the context
used to instantiate $g$, and therefore $\ma F^{\eta}=\top$ and $\eta\in S^{+}$.
Hence $\nu\in\evalThree{g}{\ma G}{S^{+}}$, and thus $\omega$ is witnessed in
\eqref{eq:lhs}. Therefore
$\omega\in\evalTwo{\join{\filter{\ma P}{\ma F}}{g}}{\ma G}$.

\smallskip\noindent
Steps 4 and 5 give mutual inclusion, proving the equality.
\end{proof}

\subsubsection{Safe Reordering of GenOp with Join}

\begin{proposition}[Safe join reordering around a GenOp]
\label{prop:genop-reorder}
Let $g=\genbind{\pi(X)}{Y}{\ma M}$ be a \tx{GenOp} occurrence.
Let $P_1,P_2$ be patterns such that:
\begin{enumerate}
  \item $X\subseteq \var{\ma P_1}$, and
  \item  $\var{\ma P_2}\cap X=\varnothing$.
\end{enumerate}
Then, the following reordering is semantics-preserving (under the denotational semantics
of Section~\ref{sect:syntax_semantics}): $\forall\ma G, \forall S$,
\[\evalThree{\join{(\join{\ma P_1}{\ma P_2})}{g}}{\ma G}{S}=
\evalThree{\join{(\join{\ma P_1}{g})}{\ma P_2}}{\ma G}{S}.\]
\end{proposition}

\emph{Remark.} The condition $\domain{\mu_2}\cap X=\varnothing$ strengthens the syntactic condition $\var{\ma P_2}\cap X=\varnothing$ and is needed because $\var{\ma P_2}$ may over-approximate the variables actually bound by $\ma P_2$ (\eg under \tx{Optional}
and \tx{Filter}).

\begin{proof}
Let $S$ be a context supply.

\smallskip\noindent
\emph{($\subseteq$)} Let $\mu\in \evalThree{\join{(\join{\ma P_1}{\ma P_2})}{g}}{\ma G}{S}$.
By the join clause, there exist
\[
\mu_{12}\in \evalThree{\join{\ma P_1}{\ma P_2}}{\ma G}{S}
\quad\text{and}\quad
\mu_g\in \evalThree{g}{\ma G}{\ContextS{\evalThree{\join{\ma P_1}{\ma P_2}}{\ma G}{S}}{g}}
\]
such that $\mu_{12}\sim \mu_g$ and
\begin{equation}\label{eq:mu-main}
\mu \;=\; \mu_{12}\uplus \mu_g.
\end{equation}
Unfolding $\mu_{12}\in \evalThree{\join{\ma P_1}{\ma P_2}}{\ma G}{S}$, there exist
$\mu_1\in\evalThree{\ma P_1}{\ma G}{S}$ and $\mu_2\in\evalThree{\ma P_2}{\ma G}{S}$
with $\mu_1\sim \mu_2$ and
\begin{equation}\label{eq:mu12}
\mu_{12}\;=\;\mu_1\uplus \mu_2.
\end{equation}

\smallskip\noindent
\emph{Claim 1:} $\mu_1\sim \mu_g$.
Let $v\in \domain{\mu_1}\cap \domain{\mu_g}$. Then $v\in\domain{\mu_{12}}\cap\domain{\mu_g}$
because $\domain{\mu_1}\subseteq \domain{\mu_{12}}$ by \eqref{eq:mu12}. Since $\mu_{12}\sim\mu_g$,
\begin{multline*}
    \compact{v}{\mu_{12}(v)}{\mu_g(v)}=\top,\\
\text{\ie}\quad
\rep_v(\mu_{12}(v))=\rep_v(\mu_g(v)).
\end{multline*}

We show that $\rep_v(\mu_{12}(v))=\rep_v(\mu_1(v))$. There are two sub-cases.
\begin{enumerate}
\item If $v\in \domain{\mu_1}\setminus \domain{\mu_2}$, then by Definition~\ref{def:compat-union}
(canonical union), $\mu_{12}(v)=\mu_1(v)$, hence trivially
$\rep_v(\mu_{12}(v))=\rep_v(\mu_1(v))$.

\item If $v\in \domain{\mu_1}\cap \domain{\mu_2}$, then by Definition~\ref{def:compat-union}
(compatibility), $\mu_1\sim\mu_2$ implies
\begin{multline*}
    \compact{v}{\mu_1(v)}{\mu_2(v)}=\top\\
\text{so}\quad
\rep_v(\mu_1(v))=\rep_v(\mu_2(v)).
\end{multline*}

Again by canonical union,
\[\mu_{12}(v)=\rep_v(\mu_1(v)) \quad(\text{equivalently }=\rep_v(\mu_2(v))).
\]
Since $\rep_v$ is idempotent on representatives, we have
\[\rep_v(\mu_{12}(v))=\rep_v(\rep_v(\mu_1(v)))=\rep_v(\mu_1(v)).
\]
\end{enumerate}
Thus, in all cases $\rep_v(\mu_{12}(v))=\rep_v(\mu_1(v))$, and therefore
\[\rep_v(\mu_1(v))=\rep_v(\mu_g(v)),
\]
\ie $\compact{v}{\mu_1(v)}{\mu_g(v)}=\top$. Since this holds for all shared $v$,
we conclude $\mu_1\sim\mu_g$.

\smallskip\noindent
Let $\mu_{1g}:=\mu_1\uplus\mu_g$, which is well-defined by Claim~1.
Then $\mu_{1g}\in \evalThree{\join{\ma P_1}{g}}{\ma G}{S}$.

\smallskip\noindent
\emph{Claim 2:} $\mu_{1g}\sim \mu_2$.
Let $v\in \domain{\mu_{1g}}\cap \domain{\mu_2}$. There are two cases.

\begin{enumerate}
\item If $v\in \domain{\mu_1}\cap \domain{\mu_2}$, then $\mu_1\sim\mu_2$ yields
$\rep_v(\mu_1(v))=\rep_v(\mu_2(v))$. Moreover, by Definition~\ref{def:compat-union},
$\mu_{1g}(v)=\rep_v(\mu_1(v))$, hence
\begin{multline*}
    \rep_v(\mu_{1g}(v))=\rep_v(\rep_v(\mu_1(v)))= \\
    \rep_v(\mu_1(v))=\rep_v(\mu_2(v)),
\end{multline*}
so, $\compact{v}{\mu_{1g}(v)}{\mu_2(v)}=\top$.

\item Otherwise, $v\notin \domain{\mu_1}$ and thus $v\in \domain{\mu_g}\cap \domain{\mu_2}$.
Then $v\in \domain{\mu_{12}}\cap\domain{\mu_g}$ because $v\in\domain{\mu_2}\subseteq\domain{\mu_{12}}$.
Since $\mu_{12}\sim \mu_g$,
\[\rep_v(\mu_{12}(v))=\rep_v(\mu_g(v)).
\]
We also have $v\notin\domain{\mu_1}$, hence by the canonical union definition
$\mu_{12}(v)=\mu_2(v)$. Therefore, $\rep_v(\mu_2(v))=\rep_v(\mu_g(v))$.
Finally, because $v\notin\domain{\mu_1}$, canonical union gives $\mu_{1g}(v)=\mu_g(v)$, so
\[
\rep_v(\mu_{1g}(v))=\rep_v(\mu_g(v))=\rep_v(\mu_2(v)),
\]
i.e.\ $\compact{v}{\mu_{1g}(v)}{\mu_2(v)}=\top$.
\end{enumerate}

Thus $\mu_{1g}\sim \mu_2$ holds.

\smallskip\noindent
Since $\mu_{1g}\sim\mu_2$, the union $\mu_{1g}\uplus\mu_2$ is defined and belongs to
$\evalThree{\join{(\join{\ma P_1}{g})}{\ma P_2}}{\ma G}{S}$.
It remains to show that it equals $\mu$.

\smallskip\noindent
\emph{Claim 3:} $\mu=(\mu_1\uplus\mu_g)\uplus \mu_2$.
Both sides are defined and have the same domain
$\domain{\mu_1}\cup\domain{\mu_2}\cup\domain{\mu_g}$.
We check pointwise that they coincide using Definition~\ref{def:compat-union}.

Let $v$ be any variable.
\begin{enumerate}
\item If $v\in\domain{\mu_1}\cap\domain{\mu_2}$, then
$\mu_{12}(v)=\rep_v(\mu_1(v))$ and
\[
\mu(v)=(\mu_{12}\uplus\mu_g)(v)=
\begin{cases}
\rep_v(\mu_{12}(v)) & \text{if } v\in\domain{\mu_g},\\
\mu_{12}(v) & \text{if } v\notin\domain{\mu_g}.
\end{cases}
\]
On the other hand, $(\mu_1\uplus\mu_g)(v)$ is either $\rep_v(\mu_1(v))$ (if $v\in\domain{\mu_g}$)
or $\mu_1(v)$ (if $v\notin\domain{\mu_g}$), and then union with $\mu_2$ makes the result
$\rep_v(\mu_1(v))$ in either subcase because $v\in\domain{\mu_2}$ and
$\rep_v(\mu_1(v))=\rep_v(\mu_2(v))$. Thus both constructions yield the same value.

\item If $v\in\domain{\mu_1}\setminus\domain{\mu_2}$, then $\mu_{12}(v)=\mu_1(v)$ and a direct
inspection of the two unions shows both yield
$\rep_v(\mu_1(v))$ if $v\in\domain{\mu_g}$ and $\mu_1(v)$ otherwise.

\item If $v\in\domain{\mu_2}\setminus\domain{\mu_1}$, then $\mu_{12}(v)=\mu_2(v)$ and similarly
both yield $\rep_v(\mu_2(v))$ if $v\in\domain{\mu_g}$ and $\mu_2(v)$ otherwise.

\item If $v\in\domain{\mu_g}\setminus(\domain{\mu_1}\cup\domain{\mu_2})$, then both yield $\mu_g(v)$.
\end{enumerate}
Hence $\mu=(\mu_1\uplus\mu_g)\uplus\mu_2$, and therefore
$\mu\in \evalThree{\join{(\join{\ma P_1}{g})}{\ma P_2}}{\ma G}{S}$.

\medskip\noindent
The reverse inclusion $(\supseteq)$ is analogous.
\end{proof}

\subsubsection{Order Invariance under Acyclic GenOp Dependencies}

\begin{proposition}
\label{prop:genop-topo}
Assume $\directedGraph{\ma D}_{\ma P}$ is acyclic.
Let $\gamma$ and $\gamma'$ be two topological orders of $\genPattern{\ma P}$. Then, the final stagewise results coincide:
\[S_\gamma^{(m)} \;=\; S_{\gamma'}^{(m)}.
\]
\end{proposition}

\begin{proof}
We prove that swapping adjacent incomparable generators does not change the final result, and then conclude by a standard swap argument.

\medskip\noindent
\emph{Step 1: a single adjacent swap is harmless.}
Let $\gamma=(\dots,g_i,g_{i+1},\dots)$ be a topological order such that $g_i$ and $g_{i+1}$ are incomparable (no edges in either direction). Let $\gamma'$ be the order obtained by swapping them: $\gamma'=(\dots,g_{i+1},g_i,\dots)$.

Incomparability means:
\begin{multline*}
    \target{g_i}\cap\holder{g_{i+1}}=\varnothing
\qquad\text{and}\\
\target{g_{i+1}}\cap\holder{g_i}=\varnothing.
\end{multline*}

Consider the stage just before either of these two generators is applied; call the current supply $S$ (this is the same set for both orders up to that point). Because $\target{g_i}$ does not occur among placeholders of $g_{i+1}$, evaluating $g_i$ may add bindings to $\target{g_i}$, but cannot change which mappings in $S$ bind all of
$\holder{g_{i+1}}$, hence it cannot change $\ContextS{S}{g_{i+1}}$ nor any instantiated prompt for $g_{i+1}$. Symmetrically, evaluating $g_{i+1}$ first cannot change the contexts for $g_i$. Therefore, the sets of generated extensions contributed by $g_i$ and $g_{i+1}$ are the
same regardless of which is applied first, and since join is associative, the supply after applying both is identical under $\gamma$ and $\gamma'$.

\medskip\noindent
\emph{For any two topological orders are connected by adjacent swaps.}
It is standard that any two linear extensions of a finite poset can be transformed into each other by repeatedly swapping adjacent incomparable elements. Applying Step~1 along this sequence of swaps shows that the final result is invariant.

Therefore, $S_\gamma^{(m)}=S_{\gamma'}^{(m)}$.

(Alternatively, theorem is exactly Item~(3) of Theorem~\ref{thm:acyclic-collapse-detailed}. Briefly, if two \tx{GenOp} occurrences are incomparable in the dependency graph, neither depends on the outputs of the other, so swapping their evaluation order does not change the extracted contexts or the generated extensions. Any two topological orders are related by a sequence of such swaps, and the resulting stagewise evaluations therefore coincide.)

\end{proof}
\subsection{Conservativity w.r.t.\ SPARQL}
\begin{proposition}\label{prop:conservativity}
Let $\ma P$ be a well-formed graph pattern that contains no occurrence of
\tx{GenOp}. Then, for any RDF graph $\ma G$, evaluation of $\ma P$ under the
extended semantics coincides with standard SPARQL evaluation:
\[\evalTwo{\ma P}{\ma G}\;=\;\bigl(\evalTwo{\ma P}{\ma G}\bigr)\big|_{\text{standard SPARQL}}.
\]
\end{proposition}

\begin{proof}
Let \ma G be fixed. Since $\ma P$ contains no \tx{GenOp}, the evaluation rules for
$\evalTwo{\ma P}{\ma G}$ never invoke any context-dependent generator clause and,
therefore, never introduce $\sourceGen$-typed values. Hence, every mapping
$\mu\in\evalTwo{\ma P}{\ma G}$ assigns variables only to RDF-typed terms, \ie 
values in $T\times\{\sourceRdf\}$.

We first show that on RDF-typed values the extended notions of compatibility and
canonical union coincide with SPARQL's usual agreement-on-overlap and set-theoretic union.

\smallskip\noindent
\emph{Claim 1 (RDF-typed compatibility collapses to equality).}
Let $v\in\ma V$ and let $\alpha=(t,\sourceRdf)$ and $\beta=(t',\sourceRdf)$.
Then
\[\compact{v}{\alpha}{\beta}=\top \quad\text{iff}\quad t=t'.\]

By definition, $\compact{v}{\alpha}{\beta}=\top$ iff $\rep_v(\alpha)=\rep_v(\beta)$. For RDF-typed inputs, the typed-term similarity $\simtext(\alpha,\beta)$ equals $1$ exactly when $t=t'$ and otherwise $0$ (Section~\ref{sect:syntax_semantics}). Thus the induced $\approx_v$-classes are singletons, so $\rep_v(t,\sourceRdf)=(t,\sourceRdf)$ and the equivalence follows.

\smallskip\noindent
\emph{Claim 2 (Mapping compatibility and canonical union coincide with SPARQL).}
Let $\mu_1,\mu_2$ be typed solution mappings with range contained in
$T\times\{\sourceRdf\}$. Then:
(i) $\mu_1\sim\mu_2$ iff $\mu_1$ and $\mu_2\!$ agree on every shared variable;
and (ii) whenever $\mu_1\sim\mu_2$, the canonical union $\mu_1\uplus\mu_2$
coincides with the ordinary union of partial functions $\mu_1\cup\mu_2$.
For (i), unfold $\mu_1\sim\mu_2$: for all $v\in\domain{\mu_1}\cap\domain{\mu_2}$
we require $\rep_v(\mu_1(v))=\rep_v(\mu_2(v))$, which by Claim~1 is equivalent
to equality of the underlying RDF terms. For (ii), by Definition~\ref{def:compat-union},
$\uplus$ preserves non-overlapping bindings and, on overlaps, returns the
representative; by Claim~1 this representative is the same RDF-typed binding.

\smallskip\noindent
We now prove the proposition by structural induction on $\ma P$.

\smallskip\noindent
\emph{Base case (BGP).}
For a basic graph pattern, the extended semantics matches triples in $\ma G$ exactly as SPARQL does; the only difference is that bound RDF terms are tagged by $\sourceRdf$. Forgetting this tag yields precisely the standard SPARQL BGP solution mappings, hence the two evaluations coincide.

\smallskip\noindent
\emph{Inductive steps.}
Assume the claim holds for strict subpatterns. For each constructor of graph patterns used in $\ma P$ (join, union, optional, minus/diff, filter, projection, etc.), the extended semantics is defined by the standard SPARQL set-theoretic clause, except that it uses $\sim$ for join compatibility and $\uplus$ for merging compatible mappings, and it evaluates filter atoms using canonical compatibility. On the fragment without \tx{GenOp}, all intermediate mappings are
RDF-typed, so by Claims 1--2 these operations coincide with SPARQL agreement on overlaps, ordinary union, and standard filter evaluation. Therefore, each constructor yields exactly the same set of solutions as standard SPARQL.

Consequently,
$\evalTwo{\ma P}{\ma G}=(\evalTwo{\ma P}{\ma G})|_{\text{standard SPARQL}}$.
\end{proof}

\subsection{Optimizer Implications}

Propositions~\ref{prop:genop-base}--\ref{prop:genop-topo} justify standard optimizer actions in the acyclic fragment:
\begin{enumerate}
\item \emph{Hoisting base-mode GenOps.}
By Corollary~\ref{cor:genop-hoist}, a \tx{GenOp} with $\holder{g}=\varnothing$
can be evaluated once and treated as a constant relation, avoiding repeated model calls and enabling early materialization.

\item \emph{Join reordering around GenOps under safe side conditions.}
Proposition~\ref{prop:genop-reorder} characterizes when a \tx{GenOp} may be moved across a join without changing the extracted context used for its evaluation, allowing cost-based join planning while preserving semantics.

\item \emph{Topological freedom in acyclic programs.}
Proposition~\ref{prop:genop-topo} shows that any topological order of acyclic generative dependencies yields the same final stagewise result $S_\gamma^{(m)}$, enabling planners to choose among valid orders using cost
estimates (\eg prompt cost, expected fan-out).
\end{enumerate}

These principles do not apply unchanged to cyclic dependencies, where least-fixpoint or well-founded semantics is required and intermediate results may depend on iteration strategy rather than a single-pass topological order.

\section{Proofs for Executable Fixpoint Evaluation}
\label{app:exec-proofs}

We prove the guarantees stated for Algorithm~\ref{alg:evalscc}: exactness on
the chosen candidate set, soundness relative to the self-support target,
relative completeness under candidate coverage, and termination.

\subsection{Objects and explicit assumptions}

Let $\ma G$ be an RDF graph, let $\ma P$ be a graph pattern, and let
$C\subseteq\genPattern{\ma P}$ be an \tx{SCC}. Let $X_C$ and $Y_C$ be as in the
main text, and let $\eta$ be an outer binding with
$X_C\subseteq\domain(\eta)$. Let $\Omega_{\ma P}$ denote the universe of typed
solution mappings used in the semantics.

\textbf{Self-support.}
For $\mu\in\Omega_{\ma P}$, we write $\tx{SS}_C(\mu)$ if
$X_C\subseteq\domain(\mu)$ and, for every
$g=\genbind{\pi(X_g)}{Y_g}{\ma M}\in C$, the following conditions hold. Let $\theta_g:=\mu_{|X_g}$ and $\nu_g:=\mu_{|Y_g}$.
Then,
\begin{multline}
\domain(\theta_g)=X_g,\quad
\domain(\nu_g)=Y_g,\quad
\mu_{|X_g\cup Y_g}=\theta_g\uplus\nu_g,\\
\text{and}\quad
\theta_g\uplus\nu_g\in\evalThree{g}{\ma G}{\{\theta_g\}}.
\tag{$\star$}
\label{eq:ss-def}
\end{multline}

\textbf{Self-support target under an outer binding.}
We define the self-support answer set of $C$ under $\eta$ by
\[\Self_C(\eta)=
\{\mu_{|Y_C}\mid
\mu\in\Omega_{\ma P},\ \mu\supseteq\eta,\ \tx{SS}_C(\mu)\}.
\tag{$\dagger$}
\label{eq:self-def}\]
Under the finite candidate coverage and exactness assumption from
Sect.~\ref{sec:exec-fixpoint}, this set coincides with the declarative answer
projections for the component under $\eta$.

\textbf{Candidate set.}
Algorithm~\ref{alg:evalscc} constructs a finite candidate set
$\tx{Cand}_C(\eta)\subseteq\Omega_{\ma P}$ whose elements extend $\eta$ and bind
the variables in $Y_C$. The algorithm returns
\[
R_C(\eta)
=
\{\mu_{|Y_C}\mid
\mu\in\tx{Cand}_C(\eta)\ \wedge\ \tx{SS}_C(\mu)\}.
\tag{$\ddagger$}
\label{eq:rc-def}
\]

\textbf{Validation.}
Algorithm~\ref{alg:evalscc} tests self-support using $\tx{Validate}$. We use
the following semantic correctness assumption.

\begin{assumption}
\label{ass:validate-correct}
For every $g=\genbind{\pi(X_g)}{Y_g}{\ma M}\in C$, every typed mapping
$\theta$ with $\domain(\theta)=X_g$, every typed mapping $\nu$ with
$\domain(\nu)=Y_g$, and every repetition parameter $R$,
\[\tx{Validate}(g,\theta,\nu;R)=\top
\,\Longleftrightarrow\,
\theta\uplus\nu\in\evalThree{g}{\ma G}{\{\theta\}}.\]
\end{assumption}

\begin{remark}
The assumption \ref{ass:validate-correct} fixes the model, decoding, parsing, normalization, and repetition policy used by the validation oracle.
\end{remark}

\subsection{Exactness on the candidate set}

\begin{theorem}
\label{thm:evalscc-exact}
For any finite candidate set $\tx{Cand}_C(\eta)$ constructed by the algorithm,
Algorithm~\ref{alg:evalscc} returns exactly $R_C(\eta)$ as defined in
\eqref{eq:rc-def}.
\end{theorem}

\begin{proof}
We prove both inclusions.

\smallskip
\noindent\emph{($\subseteq$)}
Let $\rho$ be returned by the algorithm. Then the algorithm added
$\rho=\mu_{|Y_C}$ for some $\mu\in\tx{Cand}_C(\eta)$ whose validation test succeeded. Thus, for every
$g=\genbind{\pi(X_g)}{Y_g}{\ma M}\in C$, with
$\theta_:=\mu_{|X_g}$ and $\nu_g=\mu_{|Y_g}$,
the algorithm checks
\begin{multline*}
    \domain(\nu_g)=Y_g,\quad
\mu_{|X_g\cup Y_g}=\theta_g\uplus\nu_g,
\\ \text{and}\quad
\tx{Validate}(g,\theta_g,\nu_g;R)=\top .
\end{multline*}

By Assumption~\ref{ass:validate-correct},
\[\theta_g\uplus\nu_g\in\evalThree{g}{\ma G}{\{\theta_g\}}.\]
Thus, $\tx{SS}_C(\mu)$ holds by \eqref{eq:ss-def}. Since
$\mu\in\tx{Cand}_C(\eta)$ and $\rho=\mu_{|Y_C}$, we obtain
$\rho\in R_C(\eta)$.

\smallskip
\noindent\emph{($\supseteq$)}
Let $\rho\in R_C(\eta)$. By \eqref{eq:rc-def}, there exists
$\mu\in\tx{Cand}_C(\eta)$ such that $\rho=\mu_{|Y_C}$ and $\tx{SS}_C(\mu)$.
For every $g=\genbind{\pi(X_g)}{Y_g}{\ma M}\in C$, let
\[\theta_g:=\mu_{|X_g}
\quad\text{and}\quad
\nu_g:=\mu_{|Y_g}.\]

By \eqref{eq:ss-def},
\[\domain(\nu_g)=Y_g,\quad
\mu_{|X_g\cup Y_g}=\theta_g\uplus\nu_g,
\,\text{and}\,
\theta_g\uplus\nu_g\in\evalThree{g}{\ma G}{\{\theta_g\}}.\]
Assumption~\ref{ass:validate-correct} gives
\[\tx{Validate}(g,\theta_g,\nu_g;R)=\top .\]
Therefore, the validation test succeeds when the algorithm processes $\mu$, and the algorithm adds $\mu_{|Y_C}=\rho$ to the output. Hence, $\rho$ is returned.

Both inclusions hold, so the returned set is exactly $R_C(\eta)$.
\end{proof}

\textbf{Self-support target.}
We define the self-support answer set of $C$ under $\eta$ by

\begin{equation}\label{eq:self-support}
\Self_C(\eta) =
\{\mu_{|Y_C}\mid
\mu\in\Omega_{\ma P},\ \mu\supseteq\eta,\ \tx{SS}_C(\mu)\}.
\end{equation}

Under the finite candidate coverage and exactness assumption from
Sect.~\ref{sec:exec-fixpoint}, this set coincides with the declarative answer
projections for the component under $\eta$.
\subsection{Soundness}

\begin{theorem}[Soundness relative to the self-support target]
\label{thm:evalscc-sound}
For every outer binding $\eta$,
\[R_C(\eta)\subseteq\Self_C(\eta).\]
\end{theorem}
\begin{proof}
Let $\rho\in R_C(\eta)$. By the definition of the returned set in
\eqref{eq:rc-def}, there exists a candidate
$\mu\in\tx{Cand}_C(\eta)$ such that
\[\rho=\mu_{|Y_C} \quad\text{and}\quad \tx{SS}_C(\mu).\]
The candidate set $\tx{Cand}_C(\eta)$ is constructed from the outer binding
$\eta$ by assigning values to the variables generated inside $C$. Hence every
candidate in $\tx{Cand}_C(\eta)$ extends $\eta$, and therefore
\[\mu\supseteq\eta.\]
Combining the three facts
\[\rho=\mu_{|Y_C},\quad
\mu\supseteq\eta \quad\text{and}\quad
\tx{SS}_C(\mu),\]
we obtain exactly the membership condition in the definition of
$\Self_C(\eta)$ in \eqref{eq:self-def}. Thus,
\[\rho\in\Self_C(\eta).\]
Since $\rho$ was arbitrary, $R_C(\eta)\subseteq\Self_C(\eta)$, ref. equation \ref{eq:self-support}.

 Therefore, every projection
returned by the algorithm is sound with respect to those declarative component
answers.
\end{proof}

\subsection{Relative completeness}

\begin{theorem}[Relative completeness under candidate coverage]
\label{thm:evalscc-complete}
Assume that the candidate set covers all self-supporting extensions of $\eta$:
\[\forall \mu\in\Omega_{\ma P}:\quad
(\mu\supseteq\eta\ \wedge\ \tx{SS}_C(\mu))
\Rightarrow
\mu\in\tx{Cand}_C(\eta).
\tag{$\clubsuit$}
\label{eq:coverage}\]
Then,
\[R_C(\eta)=\Self_C(\eta).\]
\end{theorem}
\begin{proof}
We prove both inclusions.

\smallskip
\noindent\emph{($\subseteq$)}
By Theorem~\ref{thm:evalscc-sound}, every returned projection belongs to the
self-support target. Hence
\[R_C(\eta)\subseteq\Self_C(\eta).\]

\smallskip
\noindent\emph{($\supseteq$)}
Let $\rho\in\Self_C(\eta)$. By the definition of $\Self_C(\eta)$ in
\eqref{eq:self-def}, there exists a mapping $\mu\in\Omega_{\ma P}$ such that
\[\rho=\mu_{|Y_C},\quad
\mu\supseteq\eta\quad\text{and}\quad
\tx{SS}_C(\mu).\]
The coverage assumption \eqref{eq:coverage} applies to this $\mu$, since it extends $\eta$ and satisfies $\tx{SS}_C(\mu)$. Therefore,
\[\mu\in\tx{Cand}_C(\eta).\]
Now we have
\[\mu\in\tx{Cand}_C(\eta)
\quad\text{and}\quad \tx{SS}_C(\mu).\]
By the definition of $R_C(\eta)$ in \eqref{eq:rc-def}, the projection
$\mu_{|Y_C}$ belongs to $R_C(\eta)$. Since $\rho=\mu_{|Y_C}$, we obtain
\[\rho\in R_C(\eta).\]
Thus,
\[\Self_C(\eta)\subseteq R_C(\eta).\]

Combining the two inclusions gives
\[R_C(\eta)=\Self_C(\eta),  \text{ref. equation}~ \ref{eq:self-support}.\]

Therefore, the equality above gives relative completeness for the declarative component answers.
\end{proof}

\subsection{Termination}

\begin{theorem}[Termination]
\label{thm:evalscc-term}
For fixed bounds $(K,B,T,R)$, Algorithm~\ref{alg:evalscc} terminates.
\end{theorem}

\begin{proof}
We bound each phase.

\smallskip
\noindent\emph{Candidate enumeration.}
For each $v\in Y_C$, the algorithm maintains a finite set $D[v]$ and enforces
$|D[v]|\le B$. Hence the total number of successful insertions into all
candidate domains is bounded by
\[\sum_{v\in Y_C} \size{D[v]} \le B \size{Y_C}.\]
The saturation test stops the enumeration once a full pass adds no new value or all domains reach their cap. Therefore the enumeration loop is finite.

\smallskip
\noindent\emph{Candidate construction.}
After enumeration, each $D[v]$ is finite. Thus
\[|\tx{Cand}_C(\eta)| \le \prod_{v\in Y_C}|D[v]| \le B^{|Y_C|}.\]
So the candidate set is finite.

\smallskip
\noindent\emph{Validation and repair.}
The algorithm iterates over the finite set $\tx{Cand}_C(\eta)$. For each candidate and each $g\in C$, it performs at most $R$ validation repetitions.
Repair attempts are bounded by $T$, and each repair invokes only finitely many calls to $\tx{Propose}$ and $\tx{Validate}$, bounded by $K$ and $R$.

All phases are finite; thus Algorithm~\ref{alg:evalscc} terminates.
\end{proof}

\section{Complexity}
\label{app:complexity}
We prove Theorem~\ref{thm:complexity} using the \tx{SCC}-based execution scheme from Section~\ref{sec:exec-fixpoint}. The proof is deliberately local: it bounds one \tx{SCC} call, one accepted projection, and one algebraic constructor at a time. 

Throughout this section, the RDF graph $\ma G$ is an input. Generated values occur only as values of typed solution mappings. We count calls to proposal, parsing, and validation procedures as symbolic oracle calls on polynomial-size encodings; the internal neural-inference cost is outside the symbolic query-evaluation problem.

\subsection{Decision Problem and Exact SCC-Bounded Evaluation}

\begin{definition}(Evaluation-membership problem)
\label{def:membership}
An instance consists of a graph pattern $\ma P$, an RDF graph $\ma G$, a typed solution mapping $\mu$ with $\domain(\mu)\subseteq\var{\ma P}$, and 
\[
  \sem\in\{\lfp,\mathrm{strat},\mathrm{wfs}\}.
\]
The task is to decide membership whether
\[
  \mu\in\evalThree{\ma P}{\ma G}{\sem}.
\]
For $\sem=\mathrm{wfs}$, membership means membership in the true component of the well-founded value; false and undefined status queries are not part of this decision problem.
\end{definition}

Data complexity fixes $\ma P$, $\sem$, and the execution parameters of Algorithm~\ref{alg:evalscc}; only $\ma G$ and $\mu$ vary. Combined complexity takes $\ma P$, $\ma G$, and $\mu$ as input, while the execution parameters remain fixed constants of the semantics.

\begin{assumption}(Exact SCC-bounded executable)
\label{ass:scc-bounded-eval}
Let $C$ be an \tx{SCC} of the generative-dependency graph of $\ma P$. Let
\[
  Y_C=\bigcup_{g\in C}\target{g}
  \quad\text{and}\quad
  X_C=\left(\bigcup_{g\in C}\holder{g}\right)\setminus Y_C .
\]
For every outer binding $\eta$ with $X_C\subseteq\domain(\eta)$, Algorithm~\ref{alg:evalscc} constructs finite domains $D_C^\eta[v]$ for $v\in Y_C$ such that
\[
  |D_C^\eta[v]|\leq B,
\]
where $B$ is a fixed execution parameter. It forms
\begin{align*}
  \tx{Cand}_C(\eta)
  :=\{\eta\uplus\nu \mid{}&
       \nu:Y_C\to \ma T_{\tx{gen}}\times\{\sourceGen\},\\
     & \nu(v)\in D_C^\eta[v]\text{ for every }v\in Y_C,\\
     & \eta\sim\nu\}.
\end{align*}
The accepted projection set is
\[
  R_C(\eta)
  :=\{\rho\mid \exists\theta\in\tx{Cand}_C(\eta):
        \rho=\theta_{|Y_C}\text{ and }\tx{SS}_C(\theta)\}.
\]
The regime is exact when $R_C(\eta)$ equals the declarative \tx{SCC}-answer projection for $C$ under $\eta$, for the chosen semantics $\sem\in\{\lfp,\mathrm{strat},\mathrm{wfs}\}$ and for every generative \tx{SCC} and every outer binding used during evaluation. This equality is a semantic exactness premise: deterministic bounded generation supplies finite deterministic candidate domains, but it does not by itself imply equality between self-supported candidates and the declarative least-fixpoint, stratified, or well-founded answers. Every generated value, prompt instance, parsed tuple, candidate, and validation object has representation length polynomial in $|\ma P|+\size{\ma G}+|\mu|$.
\end{assumption}

\begin{lemma}
\label{lem:det-to-scc-bounds}
Deterministic bounded generation assumption, together with fixed Section~\ref{sec:exec-fixpoint} parameters $(K,B,T,R)$ and polynomial-size encodings, gives the finite local candidate-domain clauses of Assumption~\ref{ass:scc-bounded-eval}. 
\end{lemma}
Lemma~\ref{lem:det-to-scc-bounds} does not imply the semantic exactness clause of Assumption~\ref{ass:scc-bounded-eval}.
\begin{proof}
Deterministic bounded generation assumption fixes the model specifier, prompt instantiation, decoding, parsing, and normalization. Hence one instantiated prompt has one finite deterministic set of parsed outputs. Algorithm~\ref{alg:evalscc} uses the fixed proposal bound $K$ and the fixed per-variable cap $B$; therefore, for each generated variable $v\in Y_C$, it constructs a finite domain satisfying
\[
  |D_C^\eta[v]|\leq B .
\]
Consequently, the candidate construction has the product form
\begin{align*}
  \tx{Cand}_C(\eta)
  =\{\eta\uplus\nu \mid{}&
       \nu:Y_C\to \ma T_{\tx{gen}}\times\{\sourceGen\},\\
     & \nu(v)\in D_C^\eta[v]\text{ for every }v\in Y_C,
       \eta\sim\nu\},
\end{align*}
which is the finite-domain part of Assumption~\ref{ass:scc-bounded-eval}. The validation phase defines the accepted projection relation $R_C(\eta)$ by the self-support test $\tx{SS}_C$. None of these finite-domain facts entails
\[
  R_C(\eta)=\mathsf{Ans}^{\sem}_C(\eta),
\]
where $\mathsf{Ans}^{\sem}_C(\eta)$ denotes the declarative \tx{SCC}-answer projection under $\sem$. That equality is exactly the semantic exactness premise stated separately in Assumption~\ref{ass:scc-bounded-eval}. Thus bounded generation supplies finite local bounds and polynomial encodings, while Assumption~\ref{ass:scc-bounded-eval} supplies the additional exactness needed by the complexity proof.
\end{proof}

For an \tx{SCC} $C$, write
\[
  N_C:=B^{|Y_C|}.
\]
In data complexity, $N_C$ is a constant because $\ma P$ is fixed. In combined complexity, $N_C$ can be exponential in $|\ma P|$, but one assignment $\nu:Y_C\to \ma T_{\tx{gen}}\times\{\sourceGen\}$ has polynomial representation length.

\subsection{Typed SCC Condensation}

\begin{definition}[Typed SCC summary operator]
\label{def:scc-summary-appendix}
For an \tx{SCC} $C$, the summary operator $\tx{GenSCC}_C$ maps a context supply $S$ to the typed mapping set
\[
\begin{aligned}
  \evalThree{\tx{GenSCC}_C}{\ma G}{S}
  =\{\eta\uplus\rho \mid{}&
      \eta\in S,
      X_C\subseteq\domain(\eta),\\
    & \rho\in R_C(\eta),
      \eta\sim\rho\}.
\end{aligned}
\]
Thus, $R_C(\eta)$ remains the $Y_C$-projection returned by Algorithm~\ref{alg:evalscc}, while $\tx{GenSCC}_C$ reattaches the outer binding $\eta$ before the surrounding SPARQL algebra consumes the result. The condensed pattern $\widehat{\ma P}$ is obtained from $\ma P$ by replacing every generative \tx{SCC} $C$ by $\tx{GenSCC}_C$ and preserving the surrounding SPARQL algebraic structure.
\end{definition}

\begin{lemma}
\label{lem:exact-scc-condensation}
Under Assumption~\ref{ass:scc-bounded-eval}, for each semantics symbol $\sem\in\{\lfp,\mathrm{strat},\mathrm{wfs}\}$ and each typed mapping $\mu$,
\[
  \mu\in\evalThree{\ma P}{\ma G}{\sem}
  \quad\Longleftrightarrow\quad
  \mu\in\evalThree{\widehat{\ma P}}{\ma G}{\{\emptyset\}} .
\]
For $\sem=\mathrm{wfs}$, the left-hand side denotes membership in the true component.
\end{lemma}

\begin{proof}
Let $C$ be one generative \tx{SCC} and let $S$ be the context supply reaching $C$. The declarative contribution of $C$ under $S$ is
\begin{align*}
  A_C(S)
  =\{\eta\uplus\rho \mid{}&
       \eta\in S,
       X_C\subseteq\domain(\eta),\\
     & \rho\in \mathsf{Ans}^{\sem}_C(\eta),
       \eta\sim\rho\},
\end{align*}
where $\mathsf{Ans}^{\sem}_C(\eta)$ is the declarative \tx{SCC}-answer projection under $\eta$ for the chosen semantics $\sem$. Assumption~\ref{ass:scc-bounded-eval} gives
\[
  \mathsf{Ans}^{\sem}_C(\eta)=R_C(\eta)
  \quad\text{for every admissible }\eta .
\]
Substituting this equality in the definition of $A_C(S)$ yields
\begin{align*}
  A_C(S)
  &=\{\eta\uplus\rho \mid
       \eta\in S,
       X_C\subseteq\domain(\eta),
       \rho\in R_C(\eta),
       \eta\sim\rho\}\\
  &=\evalThree{\tx{GenSCC}_C}{\ma G}{S}.
\end{align*}
Thus, replacing $C$ by $\tx{GenSCC}_C$ preserves the relation supplied to the parent algebraic constructor. Repeating this argument along the \tx{SCC} condensation order replaces every generative component by an extensionally equal summary. Projection, filtering, union, join, difference, minus, and optional are compositional: equal input relations give equal output relations by their set-theoretic definitions. Thus, the final root relation is unchanged. For well-founded semantics, the same replacement is applied to the true component of each exact local summary, which is the component queried by Definition~\ref{def:membership}.
\end{proof}

\begin{lemma}
\label{lem:scc-candidate-bound-detailed}
Let $C$ be an \tx{SCC} and let $\eta$ be an outer binding for $C$. Then,
\[
  |\tx{Cand}_C(\eta)|\leq N_C
  \quad\text{and}\quad
  |R_C(\eta)|\leq N_C.
\]
For every context supply $S$,
\[
  |\evalThree{\tx{GenSCC}_C}{\ma G}{S}|
  \leq N_C\cdot |S|.
\]
\end{lemma}

\begin{proof}
A candidate is determined by one value choice for each generated variable. Therefore,
\begin{align*}
  |\tx{Cand}_C(\eta)|
  &=\left|
      \{\eta\uplus\nu \mid
        \forall v\in Y_C:\nu(v)\in D_C^\eta[v],\ \eta\sim\nu\}
    \right|\\
  &\leq
    \left|
      \{\nu \mid
        \forall v\in Y_C:\nu(v)\in D_C^\eta[v]\}
    \right|\\
  &\leq \prod_{v\in Y_C}|D_C^\eta[v]|\\
  &\leq \prod_{v\in Y_C}B\\
  &=B^{|Y_C|}=N_C .
\end{align*}
The projection map
\[
\begin{aligned}
  h &: \{\theta\in\tx{Cand}_C(\eta)
          \mid \tx{SS}_C(\theta)\}\to R_C(\eta),\\
  h(\theta)&:=\theta_{|Y_C}.
\end{aligned}
\]
is surjective by the definition of $R_C(\eta)$. Thus,
\begin{align*}
  |R_C(\eta)|
  &\leq |\{\theta\in\tx{Cand}_C(\eta)\mid\tx{SS}_C(\theta)\}|\\
  &\leq |\tx{Cand}_C(\eta)|\\
  &\leq N_C .
\end{align*}
For the summary operator, let
\[
\begin{aligned}
  B_C(S)=\{(\eta,\rho)\mid{}&
       \eta\in S,
       X_C\subseteq\domain(\eta),\\
     & \rho\in R_C(\eta),
       \eta\sim\rho\}.
\end{aligned}
\]
Then,
\begin{align*}
  |\evalThree{\tx{GenSCC}_C}{\ma G}{S}|
  &=|\{\eta\uplus\rho\mid (\eta,\rho)\in B_C(S)\}|\\
  &\leq |B_C(S)|\\
  &\leq \sum_{\eta\in S,\ X_C\subseteq\domain(\eta)}|R_C(\eta)|\\
  &\leq \sum_{\eta\in S}N_C\\
  &=N_C\cdot |S| .
\end{align*}
\end{proof}

\begin{lemma}[Validation of One Accepted Projection]
\label{lem:scc-validation-cost-detailed}
Let $C$ be an \tx{SCC}, let $\eta$ be an outer binding, and let $\rho$ be a mapping with $\domain(\rho)=Y_C$. Deciding whether $\rho\in R_C(\eta)$ uses polynomial time in data complexity and polynomial space in combined complexity.
\end{lemma}

\begin{proof}
Put
\[
\begin{aligned}
  m_C&:=|C|, &
  y_C&:=|Y_C|, &
  n&:=|\ma P|+\size{\ma G}+|\eta|+|\rho|.
\end{aligned}
\]
Let $p(n)$ bound the symbolic time and space needed to construct, inspect, or compare one generated value, prompt instance, parsed tuple, mapping restriction, or validation input. Assumption~\ref{ass:scc-bounded-eval} gives such a polynomial.

By the definition of $R_C(\eta)$ in Assumption~\ref{ass:scc-bounded-eval}, membership expands as
\begin{align}
  \rho\in R_C(\eta)
  \Longleftrightarrow{}&
  \exists\theta\in\tx{Cand}_C(\eta):
       \rho=\theta_{|Y_C}
       \wedge \tx{SS}_C(\theta) .
  \label{eq:rho-r-first-expand}
\end{align}
Expanding the candidate product gives the following equivalent form:
\begin{align}
  \rho\in R_C(\eta)
  \Longleftrightarrow{}&
  \exists\nu:
  \begin{array}[t]{l}
    \domain(\nu)=Y_C,\\
    \forall v\in Y_C:\ \nu(v)\in D_C^\eta[v],\\
    \eta\sim\nu,\\
    \rho=(\eta\uplus\nu)_{|Y_C},\\
    \tx{SS}_C(\eta\uplus\nu).
  \end{array}
  \label{eq:rho-r-existential-nu}
\end{align}
Indeed, if \eqref{eq:rho-r-first-expand} holds, then $\theta\in\tx{Cand}_C(\eta)$ means that there is a total assignment $\nu:Y_C\to\ma T_{\tx{gen}}\times\{\sourceGen\}$ such that
\[
  \theta=\eta\uplus\nu,
  \qquad
  \forall v\in Y_C:\nu(v)\in D_C^\eta[v],
  \qquad
  \eta\sim\nu .
\]
Together with $\rho=\theta_{|Y_C}$ and $\tx{SS}_C(\theta)$, this yields \eqref{eq:rho-r-existential-nu}. Conversely, any $\nu$ satisfying \eqref{eq:rho-r-existential-nu} yields $\theta:=\eta\uplus\nu\in\tx{Cand}_C(\eta)$, $\rho=\theta_{|Y_C}$, and $\tx{SS}_C(\theta)$; hence \eqref{eq:rho-r-first-expand} holds. This derivation does not require $\domain(\eta)\cap Y_C=\emptyset$: if $\eta$ already binds a variable in $Y_C$, compatibility $\eta\sim\nu$ and the equality $\rho=(\eta\uplus\nu)_{|Y_C}$ enforce the shared value.

We now bound the decision procedure induced by \eqref{eq:rho-r-existential-nu}.

\smallskip
\noindent\emph{Candidate-Assignment Part.}
For a guessed or enumerated assignment $\nu$, the product-membership test is
\[
  \domain(\nu)=Y_C
  \wedge
  \bigwedge_{v\in Y_C}\nu(v)\in D_C^\eta[v]
  \wedge
  \eta\sim\nu
  \wedge
  \rho=(\eta\uplus\nu)_{|Y_C}.
\]
For each $v\in Y_C$, the scan over $D_C^\eta[v]$ has at most $B$ elements. Thus, for one $\nu$, the symbolic time is bounded by
\[
  B\cdot y_C\cdot p(n)+p(n)
  \leq (B y_C+1)p(n),
\]
and the symbolic space is bounded by
\[
  p(n)+O(\log y_C)\leq n^{O(1)}.
\]

\smallskip
\noindent\emph{Self-support part.}
Let $\theta:=\eta\uplus\nu$. The self-support predicate expands as
\begin{align*}
  \tx{SS}_C(\theta)
  \Longleftrightarrow{}&
  X_C\subseteq\domain(\theta)
  \wedge
  \bigwedge_{g\in C}\theta\in\evalThree{g}{\ma G}{\{\theta_{|\holder{g}}\}} .
\end{align*}
Using validation correctness, each conjunct for $g=\genbind{\pi(X_g)}{Y_g}{\ma M_g}$ is checked by applying the validation procedure to the two restrictions $\theta_{|X_g}$ and $\theta_{|Y_g}$. Each restriction has polynomial encoding length, and hence one validation check costs at most $p(n)$ symbolic time and space. Therefore, the self-support phase for one $\nu$ has time
\[
  \sum_{g\in C}p(n)=m_C p(n),
\]
and space
\[
  p(n)+O(\log m_C)\leq n^{O(1)}.
\]

\smallskip
\noindent\emph{Data-complexity bound.}
When $\ma P$ is fixed, $y_C$, $m_C$, and $N_C=B^{y_C}$ are constants. A deterministic algorithm enumerates all assignments
\[
  \nu\in\prod_{v\in Y_C}D_C^\eta[v]
\]
and applies the preceding checks. The number of assignments is at most
\[
  \prod_{v\in Y_C}|D_C^\eta[v]|
  \leq B^{y_C}=N_C,
\]
so the total symbolic time is bounded by
\[
  N_C\bigl((B y_C+1)p(n)+m_Cp(n)\bigr),
\]
which is polynomial in the data input because $N_C$, $B$, $y_C$, and $m_C$ are constants. The algorithm stores only one assignment $\nu$, one mapping $\theta$, and one validation configuration at a time, thus, the space is polynomial as well.

\smallskip
\noindent\emph{Combined-complexity bound.}
When $\ma P$ is part of the input, $N_C$ may be exponential. The procedure need not enumerate or store all assignments. It nondeterministically guesses one assignment $\nu$ of polynomial representation length and checks the right-hand side of \eqref{eq:rho-r-existential-nu}. The space used is
\[
  |\nu|+|\eta\uplus\nu|+p(n)+O(\log(|C|+|Y_C|))
  \leq n^{O(1)}.
\]
Thus, the test is in \tx{NPSPACE}, and by \tx{NPSPACE}=\tx{PSPACE}, it is decidable in polynomial space.
\end{proof}

\subsection{Polynomial Bounds for Fixed Condensed Patterns}

\begin{lemma}
\label{lem:bgp-bound-detailed}
Let $B_0=\{t_1,\ldots,t_m\}$ be a fixed basic graph pattern. Then,
\[
  |\evalTwo{B_0}{\ma G}|\leq \size{\ma G}^{m}.
\]
Moreover, $\evalTwo{B_0}{\ma G}$ can be enumerated in polynomial time for fixed $B_0$.
\end{lemma}

\begin{proof}
For each triple pattern $t_i$, let
\[
  M_i:=\{a\in\ma G\mid a\text{ matches }t_i\}.
\]
Then,
\[
  |M_i|\leq \size{\ma G}
  \quad\text{for every }i.
\]
Every solution of $B_0$ is induced by a compatible tuple from $M_1\times\cdots\times M_m$. Hence,
\begin{align*}
  |\evalTwo{B_0}{\ma G}|
  &\leq |M_1\times\cdots\times M_m|\\
  &=\prod_{i=1}^{m}|M_i|\\
  &\leq\prod_{i=1}^{m}\size{\ma G}\\
  &=\size{\ma G}^{m}.
\end{align*}
Enumeration uses nested loops over $M_1,\ldots,M_m$, checks compatibility, and stores one $m$-tuple of RDF triples and one induced mapping at a time. Since $m$ is fixed, the enumeration time is polynomial in $\size{\ma G}$.
\end{proof}

\begin{lemma}[Polynomial preservation by fixed condensed patterns]
\label{lem:struct-poly-detailed}
Let $Q$ be a fixed subpattern of $\widehat{\ma P}$. If the input context supply $S(\ma G)$ satisfies
\[
  |S(\ma G)|\leq q_S(\size{\ma G})
\]
for a polynomial $q_S$, then there is a polynomial $q_Q$ such that
\[
  |\evalThree{Q}{\ma G}{S(\ma G)}|
  \leq q_Q(\size{\ma G}).
\]
If $S(\ma G)$ can be enumerated in polynomial time, then $\evalThree{Q}{\ma G}{S(\ma G)}$ can be enumerated in polynomial time for fixed $Q$.
\end{lemma}

\begin{proof}
We use structural induction on $Q$. Put $n:=\size{\ma G}$. For a proper subpattern $Q_i$, write
\[
  \Omega_i:=\evalThree{Q_i}{\ma G}{S(\ma G)},
  \qquad
  |\Omega_i|\leq q_i(n).
\]

\smallskip
\noindent\emph{Basic graph pattern.}
For $Q=B_0$, Lemma~\ref{lem:bgp-bound-detailed} gives
\[
  |\evalThree{B_0}{\ma G}{S}|=|\evalTwo{B_0}{\ma G}|
  \leq n^{m}.
\]
Set $q_Q(n):=n^m$.

\smallskip
\noindent\emph{SCC summary.}
For $Q=\tx{GenSCC}_C$, Lemma~\ref{lem:scc-candidate-bound-detailed} gives
\begin{align*}
  |\evalThree{\tx{GenSCC}_C}{\ma G}{S}|
  &\leq N_C\cdot |S|\\
  &\leq N_C\cdot q_S(n).
\end{align*}
Since $Q$ is fixed, $N_C$ is constant. Set $q_Q(n):=N_C q_S(n)$.

\smallskip
\noindent\emph{Projection.}
For $Q=\proj{Q_1}{L}$, put $\Omega:=\evalThree{Q}{\ma G}{S}$. Then,
\[
  \Omega=\{\mu_{|L}\mid \mu\in\Omega_1\}.
\]
The map $\rho_L:\Omega_1\to\Omega$, $\rho_L(\mu)=\mu_{|L}$, is surjective. Therefore,
\begin{align*}
  |\Omega|
  &=|\rho_L(\Omega_1)|\\
  &\leq |\Omega_1|\\
  &\leq q_1(n).
\end{align*}
Set $q_Q(n):=q_1(n)$.

\smallskip
\noindent\emph{Filter.}
For $Q=\filter{Q_1}{F}$,
\[
  \Omega:=\evalThree{Q}{\ma G}{S}
  =\{\mu\in\Omega_1\mid F^\mu=\top\}.
\]
Hence,
\begin{align*}
  |\Omega|
  &=|\{\mu\in\Omega_1\mid F^\mu=\top\}|\\
  &\leq |\Omega_1|\\
  &\leq q_1(n).
\end{align*}
Set $q_Q(n):=q_1(n)$.

\smallskip
\noindent\emph{Union.}
For $Q=\union{Q_1}{Q_2}$,
\[
  \Omega:=\evalThree{Q}{\ma G}{S}=\Omega_1\cup\Omega_2.
\]
Thus,
\begin{align*}
  |\Omega|
  &=|\Omega_1\cup\Omega_2|\\
  &=|\Omega_1|+|\Omega_2|-|\Omega_1\cap\Omega_2|\\
  &\leq |\Omega_1|+|\Omega_2|\\
  &\leq q_1(n)+q_2(n).
\end{align*}
Set $q_Q(n):=q_1(n)+q_2(n)$.

\smallskip
\noindent\emph{Join.}
For $Q=\join{Q_1}{Q_2}$,
\[
  \Omega:=\evalThree{Q}{\ma G}{S}
  =\{\mu_1\uplus\mu_2\mid
      \mu_1\in\Omega_1,
      \mu_2\in\Omega_2,
      \mu_1\sim\mu_2\}.
\]
Let
\[
  A:=\{(\mu_1,\mu_2)\in\Omega_1\times\Omega_2
       \mid \mu_1\sim\mu_2\}.
\]
The map $h:A\to\Omega$, $h(\mu_1,\mu_2)=\mu_1\uplus\mu_2$, is surjective. Hence,
\begin{align*}
  |\Omega|
  &\leq |A|\\
  &\leq |\Omega_1\times\Omega_2|\\
  &=|\Omega_1|\cdot|\Omega_2|\\
  &\leq q_1(n)q_2(n).
\end{align*}
Set $q_Q(n):=q_1(n)q_2(n)$.

\smallskip
\noindent\emph{Difference.}
For $Q=\diff{Q_1}{Q_2}{F}$, the main semantics gives
\begin{align*}
  \Omega:=\evalThree{Q}{\ma G}{S}
  =\{\mu_1\in\Omega_1\mid{}&
       \forall\mu_2\in\Omega_2:\\
     & \mu_1\not\sim\mu_2\ \text{or}\ F^{\mu_1\uplus\mu_2}=\bot\}.
\end{align*}
Equivalently,
\begin{align*}
  \mu_1\in\Omega
  \Longleftrightarrow{}&
  \mu_1\in\Omega_1\ \,\wedge\\
  &\neg\exists\mu_2\in\Omega_2:
    \mu_1\sim\mu_2
    \wedge F^{\mu_1\uplus\mu_2}\neq\bot .
\end{align*}
Therefore $\Omega\subseteq\Omega_1$, and
\begin{align*}
  |\Omega|
  &\leq |\Omega_1|\\
  &\leq q_1(n).
\end{align*}
Set $q_Q(n):=q_1(n)$.

\smallskip
\noindent\emph{Minus.}
For $Q=\minus{Q_1}{Q_2}$,
\begin{align*}
  \Omega:=\evalThree{Q}{\ma G}{S}
  =\{\mu_1\in\Omega_1\mid{}&
       \neg\exists\mu_2\in\Omega_2:\\
     & \mu_1\sim\mu_2
       \wedge \domain(\mu_1)\cap\domain(\mu_2)\neq\emptyset\}.
\end{align*}
Thus, $\Omega\subseteq\Omega_1$, and
\begin{align*}
  |\Omega|
  &\leq |\Omega_1|\\
  &\leq q_1(n).
\end{align*}
Set $q_Q(n):=q_1(n)$.

\smallskip
\noindent\emph{Optional.}
For $Q=\opt{Q_1}{Q_2}{F}$, the main semantics is
\[
  \evalThree{Q}{\ma G}{S}
  =\evalThree{\join{Q_1}{Q_2}}{\ma G}{S}
   \cup
   \evalThree{\diff{Q_1}{Q_2}{F}}{\ma G}{S}.
\]
Set
\begin{align*}
  J&:=\{\mu_1\uplus\mu_2\mid
       \mu_1\in\Omega_1,
       \mu_2\in\Omega_2,
       \mu_1\sim\mu_2\},\\
  D_F&:=\{\mu_1\in\Omega_1\mid
       \forall\mu_2\in\Omega_2:
       \mu_1\not\sim\mu_2\ \text{or}\ F^{\mu_1\uplus\mu_2}=\bot\}.
\end{align*}
Then, $\Omega=J\cup D_F$. The join and difference derivations give
\begin{align*}
  |J|&\leq |\Omega_1|\cdot|\Omega_2|\leq q_1(n)q_2(n),\\
  |D_F|&\leq |\Omega_1|\leq q_1(n),\\
  |\Omega|
  &=|J\cup D_F|\\
  &\leq |J|+|D_F|\\
  &\leq q_1(n)q_2(n)+q_1(n).
\end{align*}
Set $q_Q(n):=q_1(n)q_2(n)+q_1(n)$.

Each $q_Q$ is obtained from $q_S$ and the induction polynomials by a finite number of additions and multiplications. Because $Q$ is fixed, the induction depth is fixed; therefore $q_Q$ is a polynomial.

For enumeration, the displayed definitions give the following time recurrences for fixed $Q$:
\begin{align*}
  T_{\tx{GenSCC}_C}(n)&\leq |S|N_C n^{O(1)},\\
  T_{\proj{Q_1}{L}}(n)&\leq T_1(n)+|\Omega_1|n^{O(1)},\\
  T_{\filter{Q_1}{F}}(n)&\leq T_1(n)+|\Omega_1|n^{O(1)},\\
  T_{\union{Q_1}{Q_2}}(n)
    &\leq T_1(n)+T_2(n)
       +( |\Omega_1|+|\Omega_2| )n^{O(1)},\\
  T_{\join{Q_1}{Q_2}}(n)
    &\leq T_1(n)+T_2(n)
       +|\Omega_1||\Omega_2|n^{O(1)},\\
  T_{\diff{Q_1}{Q_2}{F}}(n)
    &\leq T_1(n)
       +|\Omega_1|\bigl(T_2(n)+|\Omega_2|n^{O(1)}\bigr),\\
  T_{\minus{Q_1}{Q_2}}(n)
    &\leq T_1(n)
       +|\Omega_1|\bigl(T_2(n)+|\Omega_2|n^{O(1)}\bigr),\\
  T_{\opt{Q_1}{Q_2}{F}}(n)
    &\leq T_{\join{Q_1}{Q_2}}(n)+T_{\diff{Q_1}{Q_2}{F}}(n)\\
    &\quad +( |J|+|D_F| )n^{O(1)}.
\end{align*}
All cardinalities on the right-hand side are polynomial by the preceding derivations. Hence enumeration is polynomial time for fixed $Q$.
\end{proof}

\subsection{Data Complexity}

\begin{theorem}[Data-complexity upper bound]
\label{thm:data-upper-detailed}
Under Assumption~\ref{ass:scc-bounded-eval}, the evaluation-membership problem for acyclic patterns, stratified patterns, and non-stratified patterns under well-founded semantics is decidable in deterministic polynomial time in $\size{\ma G}+|\mu|$.
\end{theorem}

\begin{proof}
By Lemma~\ref{lem:exact-scc-condensation}, membership in $\evalThree{\ma P}{\ma G}{\sem}$ is equivalent to membership in
\[
  \evalThree{\widehat{\ma P}}{\ma G}{\{\emptyset\}}.
\]
Since $\ma P$ is fixed in data complexity, the condensed pattern $\widehat{\ma P}$ is fixed. The initial supply satisfies
\[
  |\{\emptyset\}|=1.
\]
Applying Lemma~\ref{lem:struct-poly-detailed} to $Q=\widehat{\ma P}$ and $S=\{\emptyset\}$ gives a polynomial $q_{\widehat{\ma P}}$ such that
\[
  |\evalThree{\widehat{\ma P}}{\ma G}{\{\emptyset\}}|
  \leq q_{\widehat{\ma P}}(\size{\ma G}).
\]
The enumeration part of Lemma~\ref{lem:struct-poly-detailed}, together with Lemma~\ref{lem:scc-validation-cost-detailed}, computes the condensed relation in time
\[
  q_{\widehat{\ma P}}(\size{\ma G})\cdot (\size{\ma G}+|\mu|)^{O(1)}.
\]
The final membership test scans the computed relation and compares each mapping with $\mu$, which costs
\[
  q_{\widehat{\ma P}}(\size{\ma G})\cdot |\mu|^{O(1)}.
\]
Both bounds are polynomial in $\size{\ma G}+|\mu|$. The argument is independent of whether $\sem$ is $\lfp$, $\mathrm{strat}$, or $\mathrm{wfs}$, because exact SCC condensation has already converted the generative components into typed summary operators.
\end{proof}

\subsection{Combined Complexity}

For combined complexity, $\ma P$ is part of the input. Then, $|Y_C|$ may be linear in $|\ma P|$, and
\[
  N_C=B^{|Y_C|}
\]
may be exponential. The upper-bound proof therefore stores and validates one assignment $\rho$ at a time.

\begin{theorem}[Combined \tx{PSPACE}-hardness]
\label{thm:pspace-hardness-detailed}
For each $\sem\in\{\lfp,\mathrm{strat},\mathrm{wfs}\}$, the evaluation-membership problem is \tx{PSPACE}-hard in combined complexity.
\end{theorem}

\begin{proof}
Given a standard SPARQL pattern $P_0$, an RDF graph $\ma G$, and a mapping $\mu$, view $P_0$ as a generative SPARQL pattern with no \tx{GenOp} occurrence. Then
\[
  \evalThree{P_0}{\ma G}{\sem}=\evalTwo{P_0}{\ma G}
\]
for every $\sem\in\{\lfp,\mathrm{strat},\mathrm{wfs}\}$, because there is no generated value, no generative \tx{SCC}, and no self-support validation. Standard SPARQL membership is \tx{PSPACE}-hard in combined complexity~\cite{perez2006semantics,perez2009semantics}. The identity map
\[
  (P_0,\ma G,\mu)\longmapsto(P_0,\ma G,\mu)
\]
is therefore a polynomial-time many-one reduction to generative SPARQL membership.
\end{proof}

\begin{definition}
\label{def:implicit-supply-membership}
A predicate $\mathsf{InS}$ is an implicit membership predicate for a conceptual set $S$ of typed solution mappings if, for every typed mapping $\eta$,
\[
  \mathsf{InS}(\eta)=\top
  \quad\Longleftrightarrow\quad
  \eta\in S.
\]
\end{definition}

\begin{lemma}
\label{lem:structural-membership-pspace-detailed}
There is a polynomial-space procedure
\[
  \mathsf{MEM}(Q,\ma G,\mathsf{InS},\mu)
\]
that decides whether $\mu\in\evalThree{Q}{\ma G}{S}$ for every subpattern $Q$ of $\widehat{\ma P}$, provided $\mathsf{InS}$ decides membership in the input supply $S$ using polynomial space.
\end{lemma}

\begin{proof}
Let
\[
  n:=|Q|+\size{\ma G}+|\mu|.
\]
We give a nondeterministic polynomial-space procedure for membership and use $\tx{NPSPACE}=\tx{PSPACE}$; for negative subtests we use closure of \tx{PSPACE} under complement.

\smallskip
\noindent\emph{Basic graph pattern.}
For $Q=B_0=\{t_1,\ldots,t_m\}$,
\begin{align*}
  \mu\in\evalTwo{B_0}{\ma G}
  \Longleftrightarrow{}&
  \exists(a_1,\ldots,a_m)\in\ma G^m:\\
  &\bigwedge_{i=1}^{m}a_i\text{ matches }t_i
    \text{ and }\mu=\bigcup_{i=1}^{m}\mu_{a_i}.
\end{align*}
The procedure guesses $a_i$ sequentially, stores at most $m$ RDF triples and the accumulated mapping, and checks compatibility after each guess. The space is
\[
  O(m\log\size{\ma G})+|\mu|^{O(1)}\leq n^{O(1)}.
\]

\smallskip
\noindent\emph{SCC summary.}
For $Q=\tx{GenSCC}_C$, Definition~\ref{def:scc-summary-appendix} gives
\begin{align*}
  \mu\in\evalThree{\tx{GenSCC}_C}{\ma G}{S}
  \Longleftrightarrow{}&
  \exists\eta,\rho:\
  \begin{array}[t]{l}
    \mathsf{InS}(\eta)=\top,\\
    X_C\subseteq\domain(\eta),\\
    \rho\in R_C(\eta),\\
    \eta\sim\rho,\\
    \mu=\eta\uplus\rho .
  \end{array}
\end{align*}
The guessed mappings $\eta$ and $\rho$ assign only variables occurring in $Q$, hence have polynomial length. The test $\mathsf{InS}(\eta)$ uses polynomial space by hypothesis. The test $\rho\in R_C(\eta)$ uses polynomial space by Lemma~\ref{lem:scc-validation-cost-detailed}. Thus,
\[
  s_Q(n)\leq s_{\mathsf{InS}}(n)+n^{O(1)}.
\]

\smallskip
\noindent\emph{Projection.}
For $Q=\proj{Q_1}{L}$,
\[
  \mu\in\evalThree{Q}{\ma G}{S}
  \Longleftrightarrow
  \exists\mu':\mu'_{|L}=\mu
  \wedge
  \mu'\in\evalThree{Q_1}{\ma G}{S}.
\]
The guessed $\mu'$ assigns only variables of $Q_1$, so
\[
  s_Q(n)\leq s_{Q_1}(n)+n^{O(1)}.
\]

\smallskip
\noindent\emph{Filter.}
For $Q=\filter{Q_1}{F}$,
\[
  \mu\in\evalThree{Q}{\ma G}{S}
  \Longleftrightarrow
  F^\mu=\top
  \wedge
  \mu\in\evalThree{Q_1}{\ma G}{S}.
\]
Thus,
\[
  s_Q(n)\leq s_{Q_1}(n)+n^{O(1)}.
\]

\smallskip
\noindent\emph{Union.}
For $Q=\union{Q_1}{Q_2}$,
\[
  \mu\in\evalThree{Q}{\ma G}{S}
  \Longleftrightarrow
  \mu\in\evalThree{Q_1}{\ma G}{S}
  \vee
  \mu\in\evalThree{Q_2}{\ma G}{S}.
\]
Hence,
\[
  s_Q(n)\leq \max\{s_{Q_1}(n),s_{Q_2}(n)\}+n^{O(1)}.
\]

\smallskip
\noindent\emph{Join.}
For $Q=\join{Q_1}{Q_2}$,
\begin{align*}
  \mu\in\evalThree{Q}{\ma G}{S}
  \Longleftrightarrow{}&
  \exists\mu_1,\mu_2:\
  \begin{array}[t]{l}
    \mu_1\in\evalThree{Q_1}{\ma G}{S},\\
    \mu_2\in\evalThree{Q_2}{\ma G}{S},\\
    \mu_1\sim\mu_2,\\
    \mu=\mu_1\uplus\mu_2 .
  \end{array}
\end{align*}
The procedure guesses $\mu_1,\mu_2$, checks compatibility and equality with $\mu$, and calls the two subprocedures sequentially. Therefore,
\[
  s_Q(n)\leq s_{Q_1}(n)+s_{Q_2}(n)+n^{O(1)}.
\]

\smallskip
\noindent\emph{Difference.}
For $Q=\diff{Q_1}{Q_2}{F}$,
\begin{align*}
  \mu\in\evalThree{Q}{\ma G}{S}
  \Longleftrightarrow{}&
  \mu\in\evalThree{Q_1}{\ma G}{S}\\
  &\wedge
  \neg\exists\mu_2:\
  \begin{array}[t]{l}
    \mu_2\in\evalThree{Q_2}{\ma G}{S},\\
    \mu\sim\mu_2,\\
    F^{\mu\uplus\mu_2}\neq\bot .
  \end{array}
\end{align*}
The existential right-witness test uses space $s_{Q_2}(n)+n^{O(1)}$, and its complement is in \tx{PSPACE}. Thus
\[
  s_Q(n)\leq s_{Q_1}(n)+s_{Q_2}(n)+n^{O(1)}.
\]

\smallskip
\noindent\emph{Minus.}
For $Q=\minus{Q_1}{Q_2}$,
\begin{align*}
  \mu\in\evalThree{Q}{\ma G}{S}
  \Longleftrightarrow{}&
  \mu\in\evalThree{Q_1}{\ma G}{S}\\
  &\wedge
  \neg\exists\mu_2:\
  \begin{array}[t]{l}
    \mu_2\in\evalThree{Q_2}{\ma G}{S},\\
    \mu\sim\mu_2,\\
    \domain(\mu)\cap\domain(\mu_2)\neq\emptyset .
  \end{array}
\end{align*}
Therefore,
\[
  s_Q(n)\leq s_{Q_1}(n)+s_{Q_2}(n)+n^{O(1)}.
\]

\smallskip
\noindent\emph{Optional.}
For $Q=\opt{Q_1}{Q_2}{F}$,
\[
  \evalThree{Q}{\ma G}{S}
  =\evalThree{\join{Q_1}{Q_2}}{\ma G}{S}
   \cup
   \evalThree{\diff{Q_1}{Q_2}{F}}{\ma G}{S}.
\]
Thus,
\begin{align*}
  \mu\in\evalThree{Q}{\ma G}{S}
  \Longleftrightarrow{}&
  \left(
  \exists\mu_1,\mu_2:\
  \begin{array}[t]{l}
    \mu_1\in\evalThree{Q_1}{\ma G}{S},\\
    \mu_2\in\evalThree{Q_2}{\ma G}{S},\\
    \mu_1\sim\mu_2,\\
    \mu=\mu_1\uplus\mu_2
  \end{array}
  \right)\\
  &\vee
  \left(
  \begin{array}[t]{l}
    \mu\in\evalThree{Q_1}{\ma G}{S},\\
    \neg\exists\mu_2:\
      \mu_2\in\evalThree{Q_2}{\ma G}{S},\\
      \mu\sim\mu_2,\\
      F^{\mu\uplus\mu_2}\neq\bot
  \end{array}
  \right).
\end{align*}
The first disjunct has the join-space bound and the second has the difference-space bound. Hence,
\[
  s_Q(n)\leq s_{Q_1}(n)+s_{Q_2}(n)+n^{O(1)}.
\]

The recursive membership calls follow the syntax tree of $Q$ and have depth at most $|Q|$. At each level, the procedure stores only polynomial-size mappings, indices, and one subcall configuration. Hence the total space is polynomial in $n$.
\end{proof}

\begin{theorem}[Combined \tx{PSPACE} upper bound]
\label{thm:pspace-upper-detailed}
Under Assumption~\ref{ass:scc-bounded-eval}, the evaluation-membership problem is in \tx{PSPACE} in combined complexity for $\sem\in\{\lfp,\mathrm{strat},\mathrm{wfs}\}$.
\end{theorem}

\begin{proof}
By Lemma~\ref{lem:exact-scc-condensation},
\[
  \mu\in\evalThree{\ma P}{\ma G}{\sem}
  \quad\Longleftrightarrow\quad
  \mu\in\evalThree{\widehat{\ma P}}{\ma G}{\{\emptyset\}}.
\]
Let $S_0:=\{\emptyset\}$. Its implicit membership predicate is
\[
  \mathsf{InS}_0(\eta)=\top
  \quad\Longleftrightarrow\quad
  \eta=\emptyset,
\]
which is decidable in polynomial space. Applying Lemma~\ref{lem:structural-membership-pspace-detailed} to $Q=\widehat{\ma P}$, $S=S_0$, and $\mathsf{InS}=\mathsf{InS}_0$ yields a polynomial-space test for the right-hand side. The equivalence above transfers the bound to the original pattern $\ma P$. For $\sem=\mathrm{wfs}$, this decides membership in the true component, as specified in Definition~\ref{def:membership}.
\end{proof}

\begin{theorem}[Combined \tx{PSPACE}-completeness]
\label{thm:pspace-complete-detailed}
Under Assumption~\ref{ass:scc-bounded-eval}, the evaluation-membership problem is \tx{PSPACE}-complete in combined complexity for $\sem\in\{\lfp,\mathrm{strat},\mathrm{wfs}\}$.
\end{theorem}

\begin{proof}
The lower bound is Theorem~\ref{thm:pspace-hardness-detailed}. The upper bound is Theorem~\ref{thm:pspace-upper-detailed}. Therefore the problem is \tx{PSPACE}-complete.
\end{proof}

\subsection{Proof of Theorem~\ref{thm:complexity}}

\begin{proof}[Proof of Theorem~\ref{thm:complexity}]
The bounded-generation statement in the main theorem is used only to obtain finite local candidate domains and polynomial-size encodings, as stated in Lemma~\ref{lem:det-to-scc-bounds}. The equality between accepted \tx{SCC} projections and declarative \tx{SCC} answers is the separate semantic exactness premise in Assumption~\ref{ass:scc-bounded-eval}. Under that exact \tx{SCC}-bounded regime, Lemma~\ref{lem:exact-scc-condensation} reduces membership for $\ma P$ to membership for the typed condensed pattern $\widehat{\ma P}$.

The data-complexity upper bound follows from Theorem~\ref{thm:data-upper-detailed}. The combined-complexity lower bound follows from Theorem~\ref{thm:pspace-hardness-detailed}. The combined-complexity upper bound follows from Theorem~\ref{thm:pspace-upper-detailed}. Combining the lower and upper bounds gives Theorem~\ref{thm:pspace-complete-detailed}.

The proof uses only local products
\[
  \prod_{v\in Y_C}D_C^\eta[v]
\]
and one-projection validation tests in the following existential generated-assignment form:
\[
\begin{aligned}
  \rho\in R_C(\eta)
  \Longleftrightarrow{}&
  \exists\nu:
  \domain(\nu)=Y_C
  \wedge
  \bigwedge_{v\in Y_C}\nu(v)\in D_C^\eta[v]\ \wedge\ \eta\sim\nu\ \wedge{}\\
  &\rho=(\eta\uplus\nu)_{|Y_C}
  \wedge
  \tx{SS}_C(\eta\uplus\nu).
\end{aligned}
\]
It therefore avoids both the invalid global generated-term universe argument and the invalid simplification that would identify the returned projection $\rho$ with the generated assignment without checking the projection equation.
\end{proof}

\end{appendix}
\end{document}